%% file: main_lossy_kernel.tex
\documentclass[svgnames, 11pt]{article}

\usepackage[latin2]{inputenc}
\usepackage[english]{babel}
\usepackage[all=normal,paragraphs=tight]{savetrees}
\usepackage{caption}
\usepackage{subcaption}
\usepackage{float}
\usepackage{enumitem}

\usepackage{amsfonts,amsthm,amsmath}
\usepackage{graphicx,tikz}
\RequirePackage{fancyhdr}
\usepackage{xcolor}
\usepackage{boxedminipage}

\usepackage{vmargin}
\setmarginsrb{1in}{1in}{1in}{1in}{0mm}{0mm}{0mm}{7mm}

\usepackage{tikz}
\usepackage{color}
\usepackage{tkz-tab}
\usepackage{subcaption}

\definecolor{Yellow}{cmyk}{0,0,1,0}
\definecolor{Magenta}{cmyk}{0,1,0,0}
\definecolor{Green}{cmyk}{1,0,1,0.25}
\definecolor{Blue}{cmyk}{1,1,0,0}
\definecolor{Orange}{cmyk}{0,0.5,1,0}
\definecolor{Red}{cmyk}{0,1,1,0}
\definecolor{Black}{cmyk}{0,0,0,1}

\newcommand{\onesafe}{$1$-safe}
\newcommand{\alphasafe}{$\alpha$-safe}

 \usepackage{algorithmic}

\newtheorem{theorem}{Theorem}
\newtheorem{lemma}{Lemma}[section]
\newtheorem{claim}{Claim}[section]

\newtheorem{definition}{Definition}[section]
\newtheorem{observation}{Observation}[section]
\newtheorem{proposition}{Proposition}[section]
\newtheorem{redrule}{Reduction Rule}[section]

\newtheorem{fact}{Fact}

\usepackage{amsthm}
\usepackage{amsmath}
\usepackage{amsfonts}
\usepackage{graphicx}
\RequirePackage{fancyhdr}
\usepackage{xcolor}
\usepackage{boxedminipage}
\usepackage{todonotes}
\usepackage{cite}
\usepackage[bookmarks=false,colorlinks,breaklinks]{hyperref}
\hypersetup{linkcolor=blue,citecolor=blue,filecolor=blue,urlcolor=blue}

\usepackage[font={scriptsize,it}]{caption}

\newcommand{\disfac}{\sc Disjoint Factors}

\newcommand{\defproblem}[3]{
  \vspace{1mm}
\noindent\fbox{
  \begin{minipage}{0.96\textwidth}
  \begin{tabular*}{\textwidth}{@{\extracolsep{\fill}}lr} #1 \\ \end{tabular*}
  {\bf{Input:}} #2  \\
  {\bf{Output:}} #3
  \end{minipage}
  }
  \vspace{1mm}
}

\usepackage{lmodern}

\usepackage[ruled,vlined,linesnumbered]{algorithm2e}

\usepackage{amstext}

\newcommand{\coNPbypoly}{${\sf co\- NP/Poly}$}

\newcommand{\NP}{\textsf{NP}}
\newcommand{\NPp}{${\sf NP}$}
\newcommand{\classP}{\textsf{P}}
\newcommand{\APX}{\textsf{APX}}
\newcommand{\FPT}{FPT}

\newcommand{\WOne}{\textsf{W[1]}}

\newcommand{\Yes}{{\sc Yes}}
\newcommand{\No}{{\sc No}}
\newcommand{\FF}{{\cal F}}
\newcommand{\BB}{{\cal B}}
\newcommand{\CC}{{\cal C}}
\newcommand{\OO}{{\cal O}}
\newcommand{\PP}{{\cal P}}
\newcommand{\QQ}{{\cal Q}}
\newcommand{\SSS}{{\cal S}}
\newcommand{\AAA}{{\cal A}}
\newcommand{\coNPsubNPbyPoly}{${\sf co\-NP\subseteq NP/Poly}$}

\newcommand{\cP}{\mathcal{P}}

\newcommand{\cQ}{\mathcal{Q}}

\newcommand{\PVC}{{\sc Partial Vertex Cover}} 
\newcommand{\df}{{\sc Disjoint Factors}}
\newcommand{\cvc}{{\sc Connected Vertex Cover}}

\newcommand{\ULI}[1]{#1-universal labelled independent set covering}
\newcommand{\scULI}[1]{{\sc #1-Universal Labelled Independent Set Covering}}
\newcommand{\shortuli}[1]{#1-ulisc}
\newcommand{\shortULI}[1]{{\sc #1-ULISC}}

\newcommand{\CP}{{\sc Disjoint Cycle Packing}} 
\newcommand{\lop}{{\sc Longest Path}}

\newcommand{\STREE}{{\sc Steiner Tree}}
\newcommand{\Setcover}{{\sc Set Cover}}
\newcommand{\SC}{{\sc SC}}
\newcommand{\Hittingset}{{\sc Hitting Set}}
\newcommand{\HS}{{\sc HS}}

\begin{document}


\title{Lossy Kernelization}
 \author{Daniel Lokshtanov\thanks{University of Bergen, Norway.
     \texttt{\{daniello,fahad.panolan\}@ii.uib.no}} 
	  \and Fahad Panolan\addtocounter{footnote}{-1}\footnotemark 
     \and M. S. Ramanujan\thanks{Technische Universit\"{a}t Wien, Vienna, Austria. 
     \texttt{ramanujan@ac.tuwien.ac.at}}\addtocounter{footnote}{-2}
 \and  Saket Saurabh\footnotemark\addtocounter{footnote}{1}~\thanks{The Institute of Mathematical Sciences, HBNI, Chennai, India. ~\texttt{saket@imsc.res.in} } 
 }

\maketitle

\thispagestyle{empty}


\begin{abstract}
\input{abstract.tex}

\end{abstract}

\newpage
\pagestyle{plain}
\setcounter{page}{1}

\input{introduction.tex}

\input{prelims}

\input{theoryFramework.tex}

%
\input{CVC.tex}


\input{labelled_interval_graph}

\input{disjoint_factor.tex}
\input{cycle_packing.tex}


%
\input{prevWork}

\input{lowerbound}

\input{setcover.tex}

\input{hittingset.tex}

%
\input{conclusion}


\bibliographystyle{siam}
\bibliography{apprx}


\end{document}

%% file: abstract.tex
In this paper we propose a new framework for analyzing the performance of preprocessing algorithms. Our framework builds on the notion of kernelization from parameterized complexity. However, as opposed to the original notion of kernelization, our definitions combine well with approximation algorithms and heuristics. The key new definition is that of a polynomial size $\alpha$-approximate kernel. Loosely speaking,
a polynomial size $\alpha$-approximate kernel is a polynomial time pre-processing algorithm that takes as input an instance $(I,k)$ to a parameterized problem, and outputs another instance $(I',k')$ to the same problem, such that $|I'|+k' \leq k^{\OO(1)}$. Additionally, for every $c \geq 1$, a $c$-approximate solution $s'$ to the pre-processed instance $(I',k')$ can be turned in polynomial time into a $(c \cdot \alpha)$-approximate solution $s$ to the original instance $(I,k)$. 

Our main technical contribution are $\alpha$-approximate kernels of polynomial size for three problems, namely {\cvc}, {\CP} and {\df}. These problems are known not to admit any polynomial size kernels unless \NP{} $\subseteq$ \coNPbypoly{}. Our approximate kernels simultaneously beat both the lower bounds on the (normal) kernel size, and the hardness of approximation lower bounds for all three problems. On the negative side we prove that {\lop} parameterized by the length of the path and \Setcover{} parameterized by the universe size do not admit even an $\alpha$-approximate kernel of polynomial size, for any $\alpha \geq 1$, unless \NP{} $\subseteq$ \coNPbypoly{}. In order to prove this lower bound we need to combine in a non-trivial way the techniques used for showing kernelization lower bounds with the methods for showing hardness of approximation. 

%% file: introduction.tex

\section{Introduction}
Polynomial time preprocessing is one of the widely used methods to tackle \NP-hardness in practice.  However, for decades there was no mathematical framework to analyze the performance of preprocessing heuristics. The advent of parameterized complexity made such an analysis possible. In parameterized complexity every instance $I$ comes with an integer {\em parameter} $k$, and the goal is to efficiently solve the instances whose parameter $k$ is small. Formally a parameterized decision problem $\Pi$ is a subset of $\Sigma^* \times \mathbb{N}$, 
where $\Sigma$ is a finite alphabet. The goal of parameterized algorithms is to determine whether an instance $(I, k)$ given as input belongs to $\Pi$ or not. 

On an intuitive level, a low value of the parameter $k$ should reflect that the instance $(I,k)$ has some additional structure that can be exploited algorithmically. Consider an instance $(I,k)$ such that $k$ is very small and $I$ is very large. Since $k$ is small, the instance is supposed to be easy. If $I$ is large and easy, this means that large parts of $I$ do not contribute to the computational hardness of the instance $(I,k)$. The hope is that these parts can be identified and reduced in polynomial time. This intuition is formalized as the notion of {\em kernelization}. Let $g : \mathbb{N} \rightarrow \mathbb{N}$ be a function. A {\em kernel of size $g(k)$} for a parameterized problem $\Pi$ is a polynomial time algorithm that takes as input an instance $(I,k)$ and outputs another instance $(I',k')$ such that $(I,k) \in \Pi$ if and only if $(I',k') \in \Pi$ and $|I'|+k' \leq g(k)$. If $g(k)$ is a linear, quadratic or polynomial function of $k$, we say that this is a linear, quadratic or polynomial kernel, respectively. 

The study of kernelization has turned into an active and vibrant subfield of parameterized complexity, especially since the development of complexity-theoretic tools to show that a problem does not admit a polynomial kernel~\cite{BodlaenderDFH09,BJK11,Drucker12,FortnowS11,HermelinKSWW15}, or a kernel of a specific size~\cite{DellM12,DellM10,HermelinW12}. Over the last decade many new results and several new techniques have been discovered, see the survey articles by Kratsch~\cite{Kratsch14} or Lokshtanov {\em et al.}~\cite{LNS2012} for recent developments, or the textbooks~\cite{CyganFKLMPPS15,DF2012Book} for an introduction to the field.

Despite the success of kernelization, the basic definition has an important drawback: {\em it does not combine well with approximation algorithms or with heuristics.} 
%
%
%
%
This is a serious problem since after all the ultimate goal of parameterized algorithms, or for that matter of any algorithmic paradigm, is to eventually solve the given input instance.  Thus, the application of a pre-processing algorithm is always followed by an algorithm that finds a solution to the reduced instance. In practice, even after applying a pre-processing procedure, the reduced instance may not be small enough to be solved to optimality within a reasonable time bound. In these cases one gives up on optimality and resorts to approximation algorithms or heuristics instead. Thus it is {\em crucial} that the solution obtained by an approximation algorithm or heuristic when run on the reduced instance provides a good solution to the original instance, or at least {\em some} meaningful information about the original instance. 
The current definition of kernels allows for kernelization algorithms with the unsavory property that running an approximation algorithm or heuristic on the reduced instance provides  {\em no insight whatsoever} about the original instance.
In particular, the {\em only} thing guaranteed by the definition of a kernel is that the reduced instance $(I',k')$ is a yes instance if and only if the original instance $(I,k)$ is. If we have an $\alpha$-approximate solution to $(I',k')$ there is no guarantee that we will be able to get an $\alpha$-approximate solution to $(I,k)$, or even able to get any feasible solution to $(I,k)$. 

\medskip
There is a lack of, and a real need for, a mathematical framework for analysing the performance of preprocessing algorithms, such that the framework not only combines well with parameterized and exact exponential time algorithms, but also with 
approximation algorithms and heuristics. {\em Our main conceptual contribution is an attempt at building such a framework.}
\medskip

The main reason that the existing notion of kernelization does not combine well with approximation algorithms is that the definition of a kernel is deeply rooted in decision problems.
The starting point of our new framework is an extension of kernelization to optimization problems.
This allows us to define $\alpha$-approximate kernels. Loosely speaking an $(\alpha)$-approximate kernel of size $g(k)$ is a polynomial time algorithm that given an instance $(I,k)$ outputs an instance $(I',k')$ such that $|I'|+ k' \leq g(k)$ and any $c$-approximate solution $s'$ to the instance $(I',k')$ can be turned in polynomial time into a $(c \cdot \alpha)$-approximate solution $s$ to the original instance $(I,k)$. In addition to setting up the core definitions of the framework we demonstrate that our formalization of lossy pre-processing is {\em robust}, {\em versatile} and {\em natural}. 

To demonstrate {\em robustness} we show that the key notions behave consistently with related notions from  parameterized complexity, kernelization, approximation algorithms and \FPT{}-approximation algorithms. More concretely we show that a problem admits an $\alpha$-approximate kernel \emph{if and only if} it is \FPT-$\alpha$-approximable, mirroring the 
equivalence between \FPT{} and kernelization~\cite{CyganFKLMPPS15}. Further, we show that the existence of a polynomial time $\alpha$-approximation algorithm is equivalent to the existence of an $\alpha$-approximate kernel of constant size.  

To demonstrate {\em versatility} we show that our framework can be deployed to measure the efficiency of pre-processing heuristics both in terms of the value of the optimum solution, and in terms of structural properties of the input instance that do not necessarily have any relation to the value of the optimum. In the language of parameterized complexity, we show that framework captures approximate kernels both for problems parameterized by the value of the optimum, and for structural parameterizations.

In order to show that the notion of $\alpha$-approximate kernelization is {\em natural}, we point to several examples in the literature where approximate kernelization has already been used implicitly to design approximation algorithms and \FPT{}-approximation algorithms. 
%
%
In particular, we show that the best known approximation algorithm for {\sc Steiner Tree}~\cite{ByrkaGRS13}, and \FPT{}-approximation for {\sc Partial Vertex Cover}~\cite{MarxFPT-AS} and for {\sc Minimal Linear Arrangement} parameterized by the vertex cover number~\cite{FellowsHRS13} can be 
re-interpreted as running an approximate kernelization first and then running an \FPT-approximation algorithm on
the preprocessed instance.  
%

A common feature of the above examples of $\alpha$-approximate kernels is that they
%
%
beat both the known lower bounds on kernel size of traditional kernels and the lower bounds on approximation ratios of approximation algorithms. Thus, it is quite possible that many of the problems for which we have strong inapproximability results and lower bounds on kernel size admit small approximate kernels with approximation factors as low as $1.1$ or $1.001$. If this is the case, it would offer up at least a partial explanation of why pre-processing heuristics combined with brute force search perform so much better than what is predicted by hardness of approximation results and kernelization lower bounds. This gives another compelling reason for a systematic investigation of lossy kernelization of parameterized optimization problems.




The observation that a lossy pre-processing can simultaneously achieve a better size bound than normal kernelization algorithms as well as a better approximation factor than the ratio of the best approximation algorithms is not new. In particular, motivated by this observation Fellows {\em et al.}~\cite{FellowsKRS12} initiated the study of lossy kernelization. Fellows {\em et al.}~\cite{FellowsKRS12} proposed a definition of lossy kernelization called $\alpha$-{\em fidelity kernels}. Essentially, an $\alpha$-fidelity kernel is a polynomial time pre-processing procedure such that an {\em optimal solution} to the reduced instance translates to an {\em $\alpha$-approximate solution} to the original. 
Unfortunately this definition suffers from the same serious drawback as the original definition of kernels - it does not combine well with approximation algorithms or with heuristics. Indeed, in the context of lossy pre-processing this drawback is even more damning, as there is no reason why one should allow a loss of precision in the pre-processing step, but demand that the reduced instance has to be solved to optimality. Furthermore the definition of $\alpha$-fidelity kernels is usable only for problems parameterized by the value of the optimum, and falls short for structural parameterizations. For these reasons we strongly believe that the notion of $\alpha$-approximate kernels introduced in this work is a better model of lossy kernelization than $\alpha$-fidelity kernels are.




It is important to note that even though the definition of $\alpha$-approximate kernels crucially differs from the definition of $\alpha$-fidelity kernels~\cite{FellowsKRS12}, it seems that most of the pre-processing algorithms that establish the existence of $\alpha$-approximate kernels can be used to establish the existence of $\alpha$-fidelity kernels and vice versa. In particular, all of the $\alpha$-fidelity kernel results of Fellows {\em et al.}~\cite{FellowsKRS12} can be translated to $\alpha$-approximate kernels.


\medskip
\noindent
{\bf Our Results.} Our main technical contribution is an investigation of the lossy kernelization complexity of several parameterized optimization problems, namely {\cvc}, {\CP}, {\df}, {\lop}, \Setcover{} and \Hittingset{}. For all of these problems there are known lower bounds~\cite{BodlaenderDFH09,BodlaenderTY09,DomLS14} precluding them from admitting polynomial kernels under widely believed complexity theoretic assumtions. Indeed, all of these six problems have played a central role in the development of the tools and techniques for showing kernelization lower bounds. 

For {\cvc}, {\CP} and {\df} we give approximate kernels that beat both the known lower bounds on kernel size and the lower bounds on approximation ratios of approximation algorithms. On the other hand, for {\lop} and \Setcover{} we show that even a constant factor approximate kernel of polynomial size would imply \NP{} $\subseteq$ \coNPbypoly{}, collapsing the polynomial hierarchy. For \Hittingset{} we show that a constant factor approximate kernel of polynomial size would violate the Exponential Time Hypothesis (ETH) of Impagliazzo, Paturi and Zane~\cite{ImpagliazzoPZ01}. Next we discuss our results for each of the six problems in more detail. An overview of the state of the art, as well as the results of this paper can be found in Table~\ref{fig:summary}.


\input{problemTable}

\medskip
\noindent
{\bf Approximate Kernels.} 
In the {\cvc} problem we are given as input a graph $G$, and the task is to find a smallest possible {\em connected vertex cover} $S \subseteq V(G)$. A vertex set $S$ is a connected vertex cover if $G[S]$ is connected and every edge has at least one endpoint in $S$. This problem is \NP-complete~\cite{ArkinHH93}, admits a factor $2$ approximation algorithm~\cite{ArkinHH93,Savage82}, and is known not to admit a factor $(2-\epsilon)$ approximation algorithm assuming the Unique Games conjecture~\cite{khot2008vertex}. Further, an approximation algorithm with ratio below $1.36$ would imply that \classP{} $=$ \NP{}~\cite{dinur2005hardness}. From the perspective of kernelization, it is easy to show that {\cvc} admits a kernel with at most $2^k$ vertices~\cite{CyganFKLMPPS15}, where $k$ is the solution size. On the other hand, Dom {\em et al.}~\cite{DomLS14} showed that {\cvc} does not admit a kernel of polynomial size, unless \NP{} $\subseteq$ \coNPbypoly{}. In this work we show that {\cvc} admits a {\em Polynomial Size Approximate Kernelization Scheme}, or PSAKS, the approximate kernelization analogue of a polynomial time approximation scheme (PTAS). In particular, for every $\epsilon > 0$, {\cvc} admits a simple $(1+\epsilon)$-approximate kernel of polynomial size. The size of the kernel is upper bounded by $k^{\OO(1/\epsilon)}$. Our results for {\cvc} show that 
allowing an arbitrarily small multiplicative loss in precision drastically improves the worst-case behaviour of preprocessing algorithms for this problem. 



In the {\CP} problem we are given as input a graph $G$, and the task is to find a largest possible collection ${\cal C}$ of pairwise disjoint vertex sets of $G$, such that every set $C \in {\cal C}$ induces a cycle in $G$. This problem admits a factor $\OO(\log n)$ approximation algorithm~\cite{SalavatipourV05}, and is known not to admit an approximation algorithm~\cite{FriggstadS11} with factor $\OO((\log n)^{\frac{1}{2}-\epsilon})$ for any $\epsilon > 0$, unless all problems in \NP{} can be solved in randomized quasi-polynomial time. With respect to kernelization, {\CP} is known not to admit a polynomial kernel~\cite{BodlaenderTY09} unless \NP{} $\subseteq$ \coNPbypoly{}.  We prove that \CP{} admits a PSAKS. More concretely we show that for every $\epsilon > 0$, {\CP} admits a $(1+\epsilon)$-approximate kernel of size $k^{\OO(\frac{1}{\epsilon \log \epsilon})}$. Again, relaxing the requirements of a kernel to allow an arbitrarily small multiplicative loss in precision yields a qualitative leap in the upper bound on kernel size from exponential to polynomial. Contrasting the simple approximate kernel for \cvc{}, the approximate kernel for {\CP} is quite complex.



On the way to obtaining a PSAKS for {\CP} we consider the {\df} problem. In {\df}, input is an alphabet $\Sigma$ and a string $s$ in $\Sigma^*$. For a letter $a \in \Sigma$, an $a$-{\em factor} in $s$ is a substring of $s$ that starts and ends with the letter $a$, and a {\em factor} in $s$ is an $a$-factor for some $a \in \Sigma$. Two factors $x$ and $y$ are {\em disjoint} if they do not overlap in $s$. In {\df} the goal is to find a largest possible subset $S$ of $\Sigma$ such that there exists a collection ${\cal C}$ of pairwise disjoint factors in $s$, such that for every $a \in S$ there is an $a$-factor in ${\cal C}$. This stringology problem shows up in the proof of the kernelization lower bound of Bodlaender {\em et al.}~\cite{BodlaenderTY09} for {\CP}.  Indeed, Bodlaenderr {\em et al.} first show that {\df} parameterized by alphabet size $|\Sigma|$ does not admit a polynomial kernel, and then reduce {\df} to {\CP} in the sense that a polynomial kernel for  {\CP} would yield a polynomial kernel for {\df}. Here we go in the other direction - first we obtain a PSAKS for {\df} parameterized by $|\Sigma|$, and then lift this result to {\CP} parameterized by solution size.

\medskip
\noindent
{\bf Lower Bounds for Approximate Kernels.} 
%
A {\em path} $P$ in a graph $G$ is a sequence $v_1v_2,\ldots v_t$ of distinct vertices, such that each pair of consecutive vertices in $P$ are adjacent in $G$. The {\em length} of the path $P$ is $t-1$, the number of vertices in $P$ minus one. In {\lop}, the input is a graph $G$ and the objective is to find a path of maximum length. The best approximation algorithm for {\lop}~\cite{AlonYZ95} has factor $\OO(\frac{n}{\log n})$, and the problem  cannot be approximated~\cite{KargerMR97} within a factor $2^{{(\log n)}^{1-\epsilon}}$ for any $\epsilon > 0$, unless \NP $=$ {\sf DTIME(}$2^{{\log n}^{O(1)}}${)}. Further, {\lop} is not expected to admit a polynomial kernel. In fact  it was one of the first \FPT{} problems for which the existence of a polynomial kernel was ruled out~\cite{BodlaenderDFH09}. We show that even within the realm of approximate kernelization, {\lop} remains hard. In particular we show that for any $\alpha \geq 1$, {\lop} does not admit an $\alpha$-approximate kernel of polynomial size unless \NP{} $\subseteq$ \coNPbypoly{}. 

In order to show the approximate kernelization lower bound for {\lop}, we extend the complexity-theoretic machinery for showing kernelization lower bounds~\cite{BodlaenderDFH09,BJK11,Drucker12,FortnowS11,HermelinKSWW15} to our framework of parameterized optimization problems. In particular we amalgamate the notion of {\em cross-compositions}, used to show kernelization lower bounds, with {\em gap-creating reductions}, used to show hardness of approximation bounds, and define  {\em gap creating cross-compositions}. Then, adapting the proofs of Fortnow and Santhanam~\cite{FortnowS11} and Bodlaender {\em et al.}~\cite{BJK11} to our setting, we show that this notion can be used to prove lower bounds on the size of approximate kernels. Once the framework of gap creating cross-compositions is set up, it trivially applies to {\lop}.

After setting up the framework for showing lower bounds for approximate kernelization, we consider the approximate kernelization complexity of two more problems, namely {\Setcover} and {\Hittingset}, both parameterized by universe size. In both problems input is a family ${\cal S}$ of subsets of a universe $U$. We use $n$ for the size of the universe $U$ and $m$ for the number of sets in $\SSS$. A {\em set cover} is a subfamily ${\cal F}$ of ${\cal S}$ such that $\bigcup_{S\in {\cal F}}S=U$. In the {\Setcover} problem the objective is to find a set cover ${\cal F}$ of minimum size. A {\em hitting set} is a subset $X$ of $U$ such that every $S \in \SSS$ has non-empty intersection with $X$, and in the {\Hittingset} problem the goal is to find a hitting set of minimum size.

The two problems are dual to each other in the following sense: given $(\SSS,U)$ we can define the {\em dual family} $(\SSS^*,U^*)$ as follows. $U^*$ has one element $u_X$ for every set $X \in \SSS$, and $\SSS^*$ has one set $S_v$ for every element $v \in U$. For every $X \in \SSS$ and $v \in U$ the set $S_v \in \SSS^*$ contains the element $u_X$ in $U^*$ if and only if $v \in X$. It is well known and easy to see that the dual of the dual of $(\SSS,U)$ is $(\SSS,U)$ itself, and that hitting sets of  $(\SSS,U)$ correspond to set covers in $(\SSS^*,U^*)$ and vice versa. 
This duality allows us to translate algorithms and lower bounds between {\Setcover} to {\Hittingset}. However, this translation {\em switches the roles of $n$ (the universe size) and $m$ (the number of sets)}. For example, {\Setcover} is known to admit a factor $(\ln n)$-approximation algorithm~\cite{bookApprox}, and known not to admit a  $(c\ln n)$-approximation algorithm for any $c < 1$ unless $\classP = \NP$~\cite{Moshkovitz15}. The duality translates these results to a $(\ln m)$-approximation algorithm, and a lower bound ruling out $(c\ln m)$-approximation algorithms for any $c < 1$ for {\Hittingset}. Nelson~\cite{Nelson07} gave a $O(\sqrt{m})$-approximation algorithm, as well as a lower bound ruling out a polynomial time $O(2^{(\log m)^c})$-approximation for any $c < 1$ for {\Setcover}, assuming the ETH. The duality translates these results to a $O(\sqrt{n})$-approximation algorithm, as well as a lower bound under ETH ruling out a polynomial time $O(2^{(\log n)^c})$-approximation for any $c < 1$ for {\Hittingset}. Observe that even though {\Setcover} and {\Hittingset} are dual to each other they behave very differently with respect to approximation algorithms that measure the quality of the approximation in terms of the universe size $n$.

For {\em kernelization} parameterized by universe size $n$, the two problems behave in a more similar fashion. Both problems admit kernels of size $O(2^n)$, and both problems have been shown not to admit kernels of size $n^{O(1)}$~\cite{DomLS14} unless \NP{} $\subseteq$ \coNPbypoly{}. However, the two lower bound proofs are quite different, and the two lower bounds do not follow from one another using the duality.

For {\Setcover} parameterized by $n$, we deploy the framework of gap creating cross-compositions to show that the problem does not admit an $\alpha$-approximate kernel of size $n^{O(1)}$ for any constant $\alpha$. This can be seen as a significant strengthening of the lower bound of Dom et al.~\cite{DomLS14}. While the gap creating cross-composition for {\lop{}} is very simple, the gap creating cross-composition for {\Setcover} is quite delicate, and relies both on a probabilistic construction and a de-randomization of this construction using co-non-determinism.

Our lower bound for {\Setcover} parameterized by universe size $n$ translates to a lower bound for {\Hittingset} parameterized by the number $m$ of sets, but says nothing about {\Hittingset} parameterized by $n$. We prove that for every $c < 1$, even a $O(2^{(\log n)^c})$-approximate kernel of size $n^{O(1)}$ for {\Hittingset} would imply a $O(2^{(\log n)^{c'}})$-approximation algorithm for {\Hittingset} for some $c' < 1$. By the result of Nelson~\cite{Nelson07} this would in turn imply that the ETH is false. Hence, {\Hittingset} does not admit a $O(2^{(\log n)^c})$-approximate kernel of size $n^{O(1)}$ assuming the ETH. 

We remark that the lower bounds proved using the framework of gap creating cross compositions, and in particular the lower bounds for {\lop} and {\Setcover}, also rule out approximate {\em compressions} to any other parameterized optimization problems. On the other hand, our lower bound for {\Hittingset} only rules out approximate {\em kernels}. As a consequence the  lower bounds for {\lop} and {\Setcover} have more potential as starting points for reductions showing that even further problems do not admit approximate kernels.

\medskip 
\noindent
{\bf Summary.} In this paper we set up a new framework for the study of lossy pre-processing algorithms, and demonstrate that the framework is natural, versatile and robust.  For several well studied problems, including {\sc Steiner Tree}, {\sc Connected Vertex Cover} and {\sc Cycle Packing} we show that a ``barely lossy'' kernelization can get dramatically better upper bounds on the kernel size than what is achievable with normal kernelization. We extend the machinery for showing kernelization lower bounds to the setting of approximate kernels, and use these new methods to prove lower bounds on the size of approximate kernels for {\lop} parameterized by the objective function value, and {\Setcover} and {\Hittingset} parameterized by universe size. Especially {\Setcover} parameterized by universe size has been a useful starting point for reductions showing lower bounds for traditional kernelization~\cite{BredereckCHKNSW14,DeyMN16,DomLS14,FernauFPS10,Guo2015Split,MisraPRRS13,YangG13}. We are therefore confident that our work lays a solid foundation for future work on approximate kernelization. Nevertheless, {\em this paper raises many more questions than it answers. It is our hope that our work will open the door for a systematic investigation of lossy pre-processing.} 

\medskip
\noindent
{\bf Organization of the Paper.}
In Section~\ref{sec:prelim} we set up notations. In section~\ref{sec:allTheDefinitions} we set up the necessary definitions to formally define and discuss approximate kernels, and relate the new notions to well known definitions from approximation algorithms and parameterized complexity. In section~\ref{sec:cvc} we give a PSAKS for \cvc{}. In section~\ref{sec:DfAndCp} we give PSAKSes for \df{} and \CP{}. In section~\ref{sec:previousWork} we show how (parts of) existing approximation algorithms for \PVC{}, \STREE{} and {\sc Optimal Linear Arrangement} can be re-interpreted as approximate kernels for these problems. In section~\ref{sec:sizelb} we set up a framework for proving lower bounds on the size of $\alpha$-approximate kernels for a parameterized optimization problem. In sections~\ref{sec:lp} and~\ref{sec:setCover} we deploy the new framework to prove lower bounds for approximate kernelization of {\lop} and {\Setcover}. In section~\ref{sec:hitSet} we give a lower bound for the approximate kernelization of \Hittingset{} by showing that a ``too good'' approximate kernel would lead to a ``too good'' approximation algorithm. We conclude with an extensive list of open problems in section~\ref{sec:conclusion}.

\medskip
\noindent
{\bf A Guide to the Paper.} 
In order to read any of the sections on concrete problems, as well as our lower bound machinery
(sections~\ref{sec:cvc}-\ref{sec:hitSet}) one needs more formal definitions of approximate kernelization and related concepts than what is given in the introduction. These definitions are given in section~\ref{sec:allTheDefinitions}. 

We have provided informal versions of the most important definitions in subsection~\ref{sec:quickAndDirty}. It should be possible to read subsection~\ref{sec:quickAndDirty} and then proceed directly to the technical sections (\ref{sec:cvc}-\ref{sec:hitSet}), only using the rest of section~\ref{sec:allTheDefinitions} occasionally as a reference. Especially the positive results of sections~\ref{sec:cvc}-\ref{sec:previousWork} should be accessible in this way. However, a reader interested in how approximate kernelization fits within a bigger picture containing approximation algorithms and kernelization should delve deeper into Section~\ref{sec:allTheDefinitions}.

All of the approximate kernelization results in sections~\ref{sec:cvc}-\ref{sec:previousWork} may be read independently of each other, except that the kernels for \df{} and \CP{} in section~\ref{sec:DfAndCp} are related. The approximate kernel for \cvc{} given in section~\ref{sec:cvc} gives a simple first example of an approximate kernel, in fact a PSAKS. The approximate kernels for \df{} and \CP{} given in section~\ref{sec:DfAndCp} are the most technically interesting positive results in the paper.

Section~\ref{sec:sizelb} sets up the methodology for proving lower bounds on approximate kernelization, this methodology is encapsulated in Theorem~\ref{thm:gapcrossComposition}. The statement of Theorem~\ref{thm:gapcrossComposition} together with the definitions of all objects in the statement are necessary to read the two lower bound sections (\ref{sec:lp} and~\ref{sec:setCover}) that apply this theorem. The lower bound for \lop{} in section~\ref{sec:lp} is a direct application of Theorem~\ref{thm:gapcrossComposition}. The lower bound for \Setcover{} in section~\ref{sec:setCover} is the most technically  interesting lower bound in the paper. The lower bound for \Hittingset{} in Section~\ref{sec:hitSet} does not rely on Theorem~\ref{thm:gapcrossComposition}, and may be read immediately after subsection~\ref{sec:quickAndDirty}

\medskip

%% file: problemTable.tex

\begin{figure}[t]
\centering
\scriptsize
{
\begin{tabular}{|c|c|c|c|c|c|}
\hline
Problem Name  &  Apx.  &  Apx. Hardness & Kernel & Apx. Ker. Fact. & Appx. Ker. Size \\ \hline \hline
{\sc Connected V.C.} & $2$\cite{ArkinHH93,Savage82} & $(2-\epsilon)$~\cite{khot2008vertex} & no $k^{\OO(1)}$~\cite{DomLS14} &  $1 < \alpha$  & $k^{f(\alpha)}$  \\ \hline                                                                                                                       
{\sc Cycle Packing} & $\OO(\log n)$\cite{SalavatipourV05} & $(\log n)^{\frac{1}{2}-\epsilon}$~\cite{FriggstadS11} & no $k^{\OO(1)}$~\cite{BodlaenderTY09} & $1 < \alpha$ & $k^{f(\alpha)}$ \\ \hline                                                                                                                       
%
%
{\sc Disjoint Factors} & $2$ & no PTAS & no $|\Sigma|^{\OO(1)}$~\cite{BodlaenderTY09} & $1 < \alpha$ & $|\Sigma|^{f(\alpha)}$  \\ \hline                                                                                                                       
{\sc Longest Path} & $\OO(\frac{n}{\log n})$~\cite{AlonYZ95} & $2^{{(\log n)}^{1-\epsilon}}$~\cite{KargerMR97} & no $k^{\OO(1)}$~\cite{BodlaenderDFH09} & any $\alpha$ & no $k^{\OO(1)}$ \\ \hline                                                                                                                       
{\sc Set Cover}/n & $\ln n$~\cite{bookApprox} & $(1-\epsilon)\ln n$~\cite{Moshkovitz15} & no $n^{\OO(1)}$~\cite{DomLS14} & any $\alpha$ & no $n^{\OO(1)}$ \\ \hline   
{\sc Hitting Set}/n & $\OO({\sqrt{n}})$~\cite{Nelson07} & $2^{{(\log n)}^{1-\epsilon}}$~\cite{Nelson07} & no $n^{\OO(1)}$~\cite{DomLS14} & any $\alpha$ & no $n^{\OO(1)}$ \\ \hline                                                                                                                      \hline

{\sc Vertex Cover} & $2$\cite{bookApprox} & $(2-\epsilon)$~\cite{dinur2005hardness,khot2008vertex} & $2k$~\cite{CyganFKLMPPS15} & $1 < \alpha < 2$ & $2(2-\alpha)k$~\cite{FellowsKRS12}  \\ \hline                                                                                                                        
{\sc $d$-Hitting Set} & $d$\cite{bookApprox} & $d-\epsilon$~\cite{DinurGKR05,khot2008vertex} & $\OO(k^{d-1})$~\cite{Abu-Khzam10} & $1 < \alpha < d$ & 
$\OO((k \cdot \frac{d-\alpha}{\alpha-1})^{d-1})$~\cite{FellowsKRS12}  \\ \hline \hline                                                                                                                      
{\sc Steiner Tree} & $1.39$\cite{ByrkaGRS13} & no PTAS~\cite{ChlebikC08} & no $k^{\OO(1)}$~\cite{DomLS14} & $1 < \alpha$ & $k^{f(\alpha)}$  \\ \hline                                                                                                                       
{\sc OLA/v.c.} & $\OO(\sqrt{\log n}\log\log n)$~\cite{FeigeL07} & no PTAS~\cite{AmbuhlMS11} & $f(k)$~\cite{loksQP15} & $1 < \alpha < 2$ & $f(\alpha)2^k k^4$ \\ \hline                                                                                                                       
{\sc Partial V.C.} & $(\frac{4}{3}-\epsilon)$~\cite{FeigeL01} & no PTAS~\cite{Patrank94} & no $f(k)$~\cite{GuoNW07} & $1 < \alpha$ & $f(\alpha)k^{5}$  \\ \hline                                                                                                                       
\end{tabular}
}
\caption{\label{fig:summary} Summary of known and new results for the problems considered in this paper. The columns show respectively: the best factor of a known approximation algorithm, the best known lower bound on the approximation ratio of polynomial time approximation algorithms, the best known kernel (or kernel lower bound), the approximation factor of the relevant approximate kernel, and the size of that approximate kernel. In the problem name column, V.C. abbreviates vertex cover. For {\sc Set Cover} and {\sc Hitting Set}, $n$ denotes the size of the universe. The approximate kernelization results for the top block of problems constitute our main technical contribution. The middle block re-states the results of Fellows et al.~\cite{FellowsKRS12} in our terminology. For the bottom block, the stated approximate kernelization results follow easily by re-interpreting in our terminology a pre-processing step within known approximation algorithms (see Section~\ref{sec:previousWork}).
} 
\end{figure}

%% file: prelims.tex
\section{Preliminaries}\label{sec:prelim}

We use ${\mathbb N}$ to denote the set of natural numbers. For a graph $G$ we use $V(G)$ and $E(G)$, to denote the
vertex and edge sets of the graph $G$  respectively. 
We use
standard terminology from the book of Diestel~\cite{Diestel} for
those graph-related terms which are not explicitly defined here.
For a vertex $v$ in $V(G)$, we use $d_G(v)$ to denote the degree of $v$, i.e the number 
of edges incident on $v$, in the (multi) graph $G$. 
For a vertex subset $S\subseteq V(G)$, we use $G[S]$ and $G-S$
to be the graphs induced on $S$ and $V(G)\setminus S$ respectively. 
For a graph $G$ and an induced subgraph $G'$ of $G$, we use $G-G'$ to denote 
the graph $G-V(G')$. 
For a vertex subset $S\subseteq V(G)$, we use $N_G(S)$ and $N_G[S]$ 
to denote the open neighbourhood and closed neighbourhood of $S$ in $G$. 
That is, $N_G(S)=\{v~\,\vert\, (u,v)\in E(G), u\in S\} \setminus S$ and 
$N_G[S]=N_G(S)\cup S$. 
For a graph $G$ and an edge $e\in E(G)$, we use $G/e$ to denote the graph obtained by contracting $e$ in $G$.  
If \(P\) is a path from a vertex \(u\) to a vertex \(v\) in graph \(G\) 
then we say that 
$u,v$ are the {\em end} vertices of the path $P$ and $P$ is a $(u,v)$-path.  
For a path $P$, we use $V(P)$ to denote the set of vertices in the path $P$ 
and the length of $P$ is denoted by $\vert P\vert $
(i.e, $\vert P \vert = \vert V(P)\vert-1$).  
For a cycle $C$, we use $V(C)$ to denote the set of vertices in the cycle $C$ 
and length of $C$, denoted by $\vert C\vert$, is $\vert V(C)\vert$.
Let \(P_{1}=x_{1}x_{2}\dotso{}x_{r}\) and
\(P_{2}=y_{1}y_{2}\dotso{}y_{s}\) be two paths in a graph $G$, 
\(V(P_{1})\cap{}V(P_{2})=\emptyset\) and $x_ry_1\in E(G)$, then we use \(P_{1}P_{2}\)
to denote the path \(x_{1}x_{2}\dotso{}x_{r}y_{1}\dotso{}y_{s}\). 
We say that \(P=x_{1}x_{2}\dotso{}x_{r}\) is an induced path in a multigraph  $G$, 
if $G[V(P)]$ is same as the simple graph $P$.   
We say that a path $P$ is a non trivial path if $\vert V(P)\vert \geq 2$.
For a path/cycle $Q$ we use $N_G(Q)$ and $N_G[Q]$ to denote the 
set $N_G(V(Q))$ and $N_G[V(Q)]$ respectively. 
For a set of paths/cycles $\QQ$, we use $\vert \QQ\vert$ and $V(\QQ)$  
to denote the number of paths/cycles in $\QQ$ and the set $\bigcup_{Q\in \QQ}V(Q)$ respectively.
The chromatic 
number of a graph $G$ is denoted by $\chi(G)$. 
An undirected graph $G$ is called an interval graph, if it is formed from  set ${\cal I}$ of intervals 
by creating one vertex $v_I$ for each interval $I\in {\cal I}$ and adding edge between two vertices 
$v_I$ and $v_{I'}$, if $I\cap I'\neq \emptyset$.  An interval representation of an interval graph $G$ 
is a set of intervals from which $G$ can be formed as described above.  
The following facts are useful in later sections. 
\begin{fact}
\label{fact1}
For any positive reals $x,y,p$ and $q$, $\min\left(\frac{x}{p},\frac{y}{q}\right)\leq \frac{x+y}{p+q} \leq \max\left(\frac{x}{p},\frac{y}{q}\right)$
\end{fact} 
\begin{fact}
\label{fact2}
For any $y\leq \frac{1}{2}$, $\left(1-y\right)\geq \left(\frac{1}{4}\right)^{y}$. 
\end{fact}

%% file: theoryFramework.tex

\section{Setting up the Framework}\label{sec:allTheDefinitions}
For the precise definition of approximate kernels, all its nuances, and how this new notion relates to approximation algorithms, \FPT{} algorithms, \FPT{}-approximation algorithms and kernelization, one should read Subsection~\ref{sec:allTheDefinitionsTheRealDeal}. For the benefit of readers eager to skip ahead to the concrete results of the paper, we include in Subsection~\ref{sec:quickAndDirty} a ``definition'' of $\alpha$-approximate kernelization that should be sufficient for reading the rest of the paper and understanding most of the arguments.

\subsection{Quick and Dirty ``Definition'' of Approximate Kernelization}\label{sec:quickAndDirty}
Recall that we work with {\em parameterized problems}. That is, every instance comes with a parameter $k$. Often $k$ is ``the quality of the solution we are looking for''. For example, does $G$ have a connected vertex cover of size at most $k$? Does $G$ have at least $k$ pairwise vertex disjoint cycles? When we move to optimization problems, we change the above two questions to: Can you find a connected vertex cover of size at most $k$ in $G$? If yes, what is the smallest one you can find? Or, can you find at least $k$ pairwise vertex disjoint cycles? If no, what is the largest collection of pairwise vertex disjoint cycles you can find? Note here the difference in how minimization and maximization problems are handled.
For minimization problems, a  bigger objective function value is undesirable, and $k$ is an ``upper bound on the `badness' of the solution''. That is, solutions worse than $k$ are so bad we do not care precisely how bad they are. For maximization problems, a  bigger objective function value is desirable, and $k$ is an ``upper bound on how good the solution has to be before one is fully satisfied''. That is, solutions better than $k$ are so good that we do not care precisely how good they are.

In many cases the parameter $k$ does not directly relate to the quality of the solution we are looking for. Consider  for example, the following problem.  Given a graph $G$ and a set $Q$ of $k$ terminals, find a smallest possible Steiner tree $T$ in $G$ that contains all the terminals. In such cases, $k$ is called a {\em structural parameter}, because $k$ being small restricts the structure of the input instance. In this example, the structure happens to be the fact that the number of terminals is `small'.

Let $\alpha \geq 1$ be a real number. We now give an informal definition of $\alpha$-approximate kernels. The kernelization algorithm should take an instance $I$ with parameter $k$, run in polynomial time, and produce a new instance $I'$ with parameter $k'$. Both $k'$ and the size of $I'$ should be bounded in terms of just the parameter $k$. That is, there should exist a function $g(k)$ such that $|I'| \leq g(k)$ and $k' \leq g(k)$. This function $g(k)$ is the {\em size} of the kernel. Now, a solution $s'$ to the instance $I'$ should be useful for finding a good solution $s$ to the instance $I$. What precisely this means depends on whether $k$ is a structural parameter or the ``quality of the solution we are looking for'', and whether we are working with a maximization problem or a minimization problem.

\begin{itemize}
\item If we are working with a structural parameter $k$ then we require the following from $\alpha$-approximate kernels:
For every $c \geq 1$, a $c$-approximate solution $s'$ to $I'$ can be transformed in polynomial time into a $(c \cdot \alpha)$-approximate solution to $I$.

\item If we are working with a minimization problem, and $k$ is the quality of the solution we are looking for, then $k$ is an ``upper bound on the badness of the solution''. In this case we require the following from $\alpha$-approximate kernels: For every $c \geq 1$, a $c$-approximate solution $s'$ to $I'$ can be transformed in polynomial time into a $(c \cdot \alpha)$-approximate solution $s$ to $I$. However, if the quality of $s'$ is ``worse than'' $k'$, or $(c \cdot \alpha) \cdot OPT(I) > k$, the algorithm that transforms $S'$ into $S$ is allowed to fail. Here $OPT(I)$ is the value of the optimum solution of the instance $I$.

{\em The solution lifting algorithm is allowed to fail precisely if the solution $S'$ given to it is ``too bad'' for the instance $I'$, or if the approximation guarantee of being a factor of $c \cdot \alpha$ away from the optimum for $I$ allows it to output a solution that is ``too bad'' for $I$ anyway.}

\item If we are working with a maximization problem, and $k$ is the quality of the solution we are looking for, then $k$ is an ``upper bound on how good the solution has to be before one is fully satisfied''. In this case we require the following from $\alpha$-approximate kernels:
For every $c \geq 1$, if $s'$ is a $c$-approximate solution $s'$ to $I'$ {\em or} the quality of $s'$ is at least $k'/c$, then $s'$ can be transformed in polynomial time into a $(c \cdot \alpha)$-approximate solution $s$ to $I$. 
However, if $OPT(I) > k$ then instead of being a $(c \cdot \alpha)$-approximate solution $s$ to $I$, the output solution $s$ can be any solution of quality at least $k / (c \cdot \alpha)$.

{\em In particular, if $OPT(I') > k'$ then the optimal solution to $I'$ is considered ``good enough'', and the approximation ratio $c$ of the solution $s'$ to $I'$ is computed as ``distance from being good enough'', i.e as $k' / s'$. Further, if $OPT(I) > k$ then we think of the optimal solution to $I$ as ``good enough'', and measure the approximation ratio of $s$ in terms of ``distance from being good enough'', i.e as $k / s$.}
\end{itemize}

We encourage the reader to instantiate the above definitions with $c \in \{1, 2\}$ and $\alpha \in \{1, 2\}$. That is, what happens to optimal and $2$-approximate solutions to the reduced instance when the approximate kernel incurs no loss ($\alpha = 1$)? What happens to optimal and $2$-approximate solutions to the reduced instance when the approximate kernel incurs a factor $2$ loss (i.e $\alpha = 2$)?

Typically we are interested in $\alpha$-approximate kernels of {\em polynomial size}, that is kernels where the size function $g(k)$ is upper bounded by $k^{O(1)}$. Of course the goal is to design $\alpha$-approximate kernels of smallest possible size, with smallest possible $\alpha$. Sometimes we are able to obtain a $(1+\epsilon)$-approximate kernel of polynomial size for every $\epsilon > 0$. Here the exponent and the constants of the polynomial may depend on $\epsilon$. We call such a kernel a {Polynomial Size Approximate Kernelization Scheme}, and abbreviate it as PSAKS. If only the constants of the polynomial $g(k)$ and not the exponent depend on $\epsilon$,  we say that the PSAKS is {\em efficient}. All of the positive results achieved in this paper are PSAKSes, but not all are efficient.


\subsection{Approximate Kernelization, The Real Deal.}\label{sec:allTheDefinitionsTheRealDeal}
We will be dealing with approximation algorithms and solutions that are not necessarily optimal, but at the same time relatively ``close'' to being optimal.  To properly discuss these concepts they have to be formally defined. Our starting point is a parameterized analogue of the notion of an {\em optimization problem} from the theory of approximation algorithms. 
\begin{definition}\label{def:paraOptProblem}
A parameterized optimization (minimization or maximization) problem $\Pi$ is a computable function
$$\Pi~:~\Sigma^*\times \mathbb{N}\times \Sigma^*\rightarrow {\mathbb R}\cup \{\pm\infty\}.$$
\end{definition}
The {\em instances} of a parameterized optimization problem $\Pi$ are pairs $(I,k) \in \Sigma^*\times \mathbb{N}$, and a {\em solution} to  $(I,k)$ is simply a string $s \in \Sigma^*$, such that $|s| \leq |I|+k$. The {\em value} of the solution $s$ is  $\Pi(I,k,s)$. Just as for ``classical'' optimization problems the instances of $\Pi$ are given as input, and the algorithmic task is to find a solution with the best possible value, where best means minimum for minimization problems and maximum for maximization problems. 
\begin{definition}\label{def:paraOpt}
For a parameterized minimization problem $\Pi$, the {\em optimum value} of an instance  $(I,k) \in \Sigma^*\times \mathbb{N}$ is
\begin{equation*}
OPT_\Pi(I,k) = \min_{\substack{s \in \Sigma^* \\ |s| \leq |I|+k}} \Pi(I,k,s).
\end{equation*}
For a parameterized maximization problem $\Pi$, the optimum value of $(I,k)$ is
\begin{equation*}
OPT_\Pi(I,k) = \max_{\substack{s \in \Sigma^* \\ |s| \leq |I|+k}} \Pi(I,k,s).
\end{equation*}
For an instance $(I,k)$ of a parameterized optimization problem $\Pi$, an {\em optimal solution} is a solution $s$ such that $\Pi(I,k,s) = OPT_\Pi(I,k)$.
\end{definition}
When the problem $\Pi$ is clear from context we will often drop the subscript and refer to $OPT_\Pi(I,k)$ as $OPT(I,k)$. Observe that in the definition of  $OPT_\Pi(I,k)$ the set of solutions over which we are minimizing/maximizing $\Pi$ is finite, therefore the minimum or maximum is well defined. We remark that the function $\Pi$ in Definition~\ref{def:paraOptProblem} depends {\em both} on $I$ and on $k$. Thus it is possible to define parameterized problems such that an optimal solution $s$ for $(I,k)$ is not necessarily optimal for $(I,k')$.

For an instance $(I,k)$ the {\em size} of the instance is $|I|+k$ while the integer $k$ is referred to as the {\em parameter} of the instance. Parameterized Complexity deals with measuring the running time of algorithms in terms of both the input size and the parameter. In Parameter Complexity a problem is {\em fixed parameter tractable} if input instances of size $n$ with parameter $k$ can be ``solved'' in time $f(k)n^{\OO(1)}$ for a computable function $f$. For decision problems ``solving'' an instance means to determine whether the input instance is a ``yes'' or a ``no'' instance to the problem. Next we define what it means to ``solve'' an instance of a parameterized optimization problem, and define fixed parameter tractability for parameterized optimization problems.
\begin{definition}\label{def:paraOptProb}
Let $\Pi$ be a parameterized optimization problem. An {\em algorithm for} $\Pi$ is an algorithm that given as input an instance $(I,k)$, outputs a solution $s$ and halts. The algorithm {\em solves} $\Pi$ if, for every instance $(I,k)$ the solution $s$ output by the algorithm is optimal for $(I,k)$. We say that a parameterized optimization problem $\Pi$ is {\em decidable} if there exists an algorithm that solves $\Pi$.
\end{definition}
\begin{definition}\label{def:paraOptFPT}
A parameterized optimization problem $\Pi$ is {\em fixed parameter tractable} (FPT) if there is an algorithm that solves $\Pi$, such that the running time of the algorithm on instances of size $n$ with parameter $k$ is upper bounded by $f(k)n^{\OO(1)}$ for a computable function $f$.
\end{definition}

We remark that Definition~\ref{def:paraOptProb} differs from the usual formalization of what it means to ``solve'' a decision problem. Solving a decision problem amounts to always returning ``yes'' on ``yes''-instances and ``no'' on ``no''-instances. For parameterized optimization problems the algorithm has to produce an optimal solution. This is analogous to the definition of optimization problems most commonly used in approximation algorithms.

We remark that we could have built the framework of approximate kernelization on the existing definitions of parameterized optimization problems used in parameterized approximation algorithms~\cite{MarxFPT-AS}, indeed the difference between our definitions of parameterized optimization problems and those currently used in parameterized approximation algorithms are mostly notational. 


\paragraph{Parameterizations by the Value of the Solution.}
At this point it is useful to consider a few concrete examples, and to discuss the relationship between parameterized optimization problems and decision variants of the same problem. For a concrete example, consider the {\sc Vertex Cover} problem. Here the input is a graph $G$, and the task is to find a smallest possible {\em vertex cover} of $G$: a subset $S \subseteq V(G)$ is a {\em vertex cover} if every edge of $G$ has at least one endpoint in $S$. This is quite clearly an optimization problem, the feasible solutions are the vertex covers of $G$ and the objective function is the size of $S$.

In the most common formalization of the {\sc Vertex Cover} problem as a {\em decision problem} parameterized by the solution size, the input instance $G$ comes with a parameter $k$ and the instance $(G,k)$ is a ``yes'' instance if $G$ has a vertex cover of size at most $k$. Thus, the parameterized decision problem ``does not care'' whether $G$ has a vertex cover of size even smaller than $k$, the only thing that matters is whether a solution of size at most $k$ is present.

To formalize {\sc Vertex Cover} as a parameterized optimization problem, we need to determine for every instance $(G,k)$ which value to assign to potential solutions $S \subseteq V(G)$. We can encode the set of feasible solutions by giving finite values for vertex covers of $G$ and $\infty$ for all other sets. We want to distinguish between graphs that do have vertex covers of size at most $k$ and the ones that do not. At the same time, we want the computational problem of solving the instance $(G,k)$ to become easier as $k$ decreases.  A way to achieve this is to assign $|S|$ to all vertex covers $S$ of $G$ of size at most $k$, and $k+1$ for all other vertex covers. Thus, one can formalize the {\sc Vertex Cover} problem as a parameterized optimization problem as follows.
\begin{equation*}
VC(G,k,S) = \left\{
\begin{array}{rl}
\infty & \text{if } $S$ \text{ is not a vertex cover of } $G$,\\
\min(|S|,k+1) & \text{ otherwise.}
\end{array} \right.
\end{equation*}
Note that this formulation of {\sc Vertex Cover} ``cares'' about  solutions of size less than $k$. One can think of $k$ as a threshold: for solutions of size at most $k$ we care about what their size is, while all solutions of size larger than $k$ are equally bad in our eyes, and are assigned value $k+1$. 

Clearly any FPT algorithm  that solves the parameterized optimization version of {\sc Vertex Cover} also solves the (parameterized) decision variant. Using standard self-reducibility techniques~\cite{S2012book} one can make an FPT algorithm for the decision variant solve the optimization variant as well.

We have seen how a minimization problem can be formalized as a parameterized optimization problem parameterized by the value of the optimum. Next we give an example for how to do this for maximization problems. In the {\sc Cycle Packing} problem we are given as input a graph $G$, and the task is to find a largest possible collection ${\cal C}$ of pairwise vertex disjoint cycles. Here a {\em collection of vertex disjoint cycles} is a collection ${\cal C}$ of vertex subsets of $G$ such that for every $C \in {\cal C}$, $G[C]$ contains a cycle and for every $C,C' \in {\cal C}$ we have $V(C) \cap V(C') = \emptyset$. We will often refer to a collection of vertex disjoint cycles as a {\em cycle packing}.

We can formalize the  {\sc Cycle Packing} problem as a parameterized optimization problem parameterized by the value of the optimum in a manner similar to what we did for {\sc Vertex Cover}. In particular, if ${\cal C}$ is a cycle packing, then we assign it value $|{\cal C}|$ if $|{\cal C}| \leq k$ and value $k+1$ otherwise. If $|{\cal C}|$ is not a cycle packing, we give it value $-\infty$.
\begin{equation*}
CP(G,k,{\cal C}) = \left\{
\begin{array}{rl}
-\infty & \text{if } {\cal C} \text{ is not a cycle packing},\\
\min(|{\cal C}|,k+1) & \text{ otherwise.}
\end{array} \right.
\end{equation*}
Thus, the only (formal) difference between the formalization of parameterized minimization and maximization problems parameterized by the value of the optimum is how infeasible solutions are treated. For minimization problems infeasible solutions get value $\infty$, while for maximization problems they get value $-\infty$. However, there is also a ``philosophical'' difference between the formalization of minimization and maximization problems. For minimization problems we do not distinguish between feasible solutions that are ``too bad''; solutions of size more than $k$ are all given the same value. On the other hand, for maximization problems all solutions that are ``good enough'', i.e. of size at least $k+1$, are considered equal.

Observe that the ``capping'' of the objective function at $k+1$ {\em does not make sense for approximation algorithms} if one insists on $k$ being the (un-parameterized) optimum of the instance $I$. The parameterization discussed above is {\em by the value of the solution that we want our algorithms to output}, not by the unknown optimum. We will discuss this topic in more detail in the paragraph titled ``{\bf Capping the objective function at $k+1$}'', after the notion of approximate kernelization has been formally defined.


\paragraph{Structrural Parameterizations.} We now give an example that demonstrates that the notion of parameterized optimization problems is robust enough to capture not only parameterizations by the value of the optimum, but also parameterizations by structural properties of the instance that may or may not be connected to the value of the best solution. In the {\sc Optimal Linear Arrangement} problem we are given as input a graph $G$, and the task is to find a bijection $\sigma : V(G) \rightarrow \{1,\ldots,n\}$ such that $\sum_{uv \in E(G)} |\sigma(u) - \sigma(v)|$ is minimized. A bijection $\sigma : V(G) \rightarrow \{1,\ldots,n\}$ is called a {\em linear layout}, and  $\sum_{uv \in E(G)} |\sigma(u) - \sigma(v)|$ is denoted by $val(\sigma, G)$ and is called the {\em value} of the layout $\sigma$.

We will consider the {\sc Optimal Linear Arrangement} problem for graphs that have a relatively small vertex cover. This can be formalized as a parameterized optimization problem as follows:
\begin{equation*}
OLA((G,S),k,\sigma) = \left\{
\begin{array}{rl}
-\infty & \text{if } ${\cal S}$ \text{ is not vertex cover of } G \text{ of size at most } k,\\ 
\infty & \text{if } \sigma \text{ is not a linear layout},\\
val(\sigma,G) & \text{ otherwise.}
\end{array} \right.
\end{equation*}
In the definition above the first case takes precendence over the second: if $S$ is not vertex cover of $G$ of size at most $k$ and $\sigma$ is not a linear layout, $OLA((G,S),k,\sigma)$ returns $-\infty$. This ensures that malformed input instances do not need to be handled.

Note that the input instances to the parameterized optimization problem described above are pairs $((G,S),k)$ where $G$ is a graph, $S$ is a vertex cover of $G$ of size at most $k$ and $k$ is the parameter. This definition allows algorithms for  {\sc Optimal Linear Arrangement}  parameterized by vertex cover to assume that the vertex cover $S$ is given as input.

\paragraph{Kernelization of Parameterized Optimization Problems.}
The notion of a kernel (or kernelization algorithm) is a mathematical model for polynomial time pre-processing for decision problems. We will now define the corresponding notion for parameterized optimization problems. To that end we first need to define a polynomial time pre-processing algorithm. 
%

\begin{definition}\label{def:polyTimePreProcess}
A {\bf polynomial time pre-processing algorithm} ${\cal A}$ for a parameterized optimization problem $\Pi$ is a pair of polynomial time algorithms. The first one is called the {\bf reduction algorithm}, and computes a map ${\cal R}_{\cal A} : \Sigma^* \times \mathbb{N} \rightarrow \Sigma^* \times \mathbb{N}$. Given as input an instance $(I,k)$ of $\Pi$ the reduction algorithm outputs another instance $(I',k') = {\cal R}_{\cal A}(I,k)$.

The second algorithm is called the {\bf solution lifting algorithm}. This algorithm
takes as input an instance $(I,k) \in \Sigma^* \times \mathbb{N}$ of $\Pi$, the output instance $(I',k')$ of the reduction algorithm, and a solution $s'$ to the instance $(I',k')$. The solution lifting algorithm works in time polynomial in $|I|$,$k$,$|I'|$,$k'$ and $s'$, and outputs a solution $s$ to $(I,k)$. Finally, if $s'$ is an optimal solution to $(I',k')$ then $s$ is an optimal solution to $(I,k)$.
\end{definition}


Observe that the solution lifting algorithm could contain the reduction algorithm as a subroutine. Thus, on input $(I,k,I',k',s')$ the solution lifting  algorithm could start by running the reduction algorithm $(I,k)$ and produce a transcript of \emph{how} the reduction algorithm obtains $(I',k')$ from $(I,k)$. Hence, when designing the solution lifting algorithm we may assume without loss of generality that such a transcript is given as input. For the same reason, it is not really necessary to include $(I',k')$ as input to the solution lifting algorithm. However, to avoid starting every description of a solution lifting algorithm with ``we compute the instance $(I',k')$ from $(I,k)$'', we include $(I',k')$ as input. 
The notion of polynomial time pre-processing algorithms could be extended to {\em randomized} polynomial time pre-processing algorithms, by allowing both the reduction algorithm and the solution lifting algorithm to draw random bits, and {\em fail} with a small probability. With such an extension it matters whether the solution lifting algorithm has access to the random bits drawn by the reduction algorithm, because these bits might be required to re-construct the transcript of how the reduction algorithm obtained $(I',k')$ from $(I,k)$. If the random bits of the reduction algorithm are provided to the solution lifting algorithm, the discussion above applies.

A kernelization algorithm is a polynomial time pre-processing algorithm for which we can prove an upper bound on the size of the output instances in terms of the parameter of the instance to be preprocessed. Thus, the {\em size} of a polynomial time pre-processing algorithm ${\cal A}$ is a function $\text{size}_{\cal A} : \mathbb{N} \rightarrow \mathbb{N}$ defined as follows.
$$\text{size}_{\cal A}(k) = \sup\{|I'| + k' : (I',k') = {\cal R}_{\cal A}(I,k), I \in \Sigma^*\}. $$

In other words, we look at all possible instances of $\Pi$ with a fixed parameter $k$, and measure the supremum of the sizes of the output of ${\cal R}_{\cal A}$ on these instances. At this point, recall that the {\em size} of an instance $(I,k)$ is defined as $|I|+k$. Note that this supremum may be infinite; this happens when we do not have any bound on the size of ${\cal R}_{\cal A}(I,k)$ in terms of the input parameter $k$ only. Kernelization algorithms are exactly these polynomial time preprocessing algorithms whose output size is finite and bounded by a computable function of the parameter.
\begin{definition}\label{def:kernel}
A {\bf kernelization} (or {\em kernel}) for a parameterized optimization problem $\Pi$ is a polynomial time pre-processing algorithm ${\cal A}$ such that $\text{size}_{\cal A}$ is upper bounded by a computable function $g : \mathbb{N} \rightarrow \mathbb{N}$.
\end{definition}
If the function $g$ in Definition~\ref{def:kernel} is a polynomial, we say that $\Pi$ admits a {\em polynomial kernel}. Similarly, if $g$ is a linear, quadratic or cubic function of $k$ we say that $\Pi$ admits a linear, quadratic, or cubic kernel, respectively.

One of the basic theorems in Parameterized Complexity is that a decidable parameterized decision problem admits a kernel if and only if it is fixed parameter tractable. We now show that this result also holds for parameterized optimization problems. We say that a parameterized optimization problem $\Pi$ is {\em decidable} if there exists an algorithm that solves $\Pi$, where the definition of ``solves'' is given in Definition~\ref{def:paraOptProb}.

\begin{proposition}\label{prop:optKernelEquiv}
A decidable parameterized optimization problem $\Pi$ is FPT if and only if it admits a kernel.
\end{proposition}

\begin{proof}
The backwards direction is trivial; on any instance $(I,k)$ one may first run the reduction algorithm to obtain a new instance $(I',k')$ of size bounded by a function $g(k)$. Since the instance $(I',k')$ has bounded size and $\Pi$ is decidable one can find an optimal solution $s'$ to $(I',k')$ in time upper bounded by a function $g'(k)$. Finally one can use the solution lifting algorithm to obtain an optimal solution $s$ to $(I,k)$.

For the forward direction we need to show that if a parameterized optimization problem $\Pi$ is FPT then it admits a kernel. Suppose there is an algorithm that solves instances $\Pi$ of size $n$ with parameter $k$ in time $f(k)n^c$. On input $(I,k)$ the reduction algorithm runs the FPT algorithm for $n^{c+1}$ steps. If the FPT algorithm terminates after at most $n^{c+1}$ steps, the reduction algorithm outputs an instance $(I',k')$ of constant size. The instance $(I',k')$ is hard-coded in the reduction algorithm and does not depend on the input instance $(I,k)$. Thus $|I'|+k'$ is upper bounded by a constant. If the FPT algorithm does not terminate after  $n^{c+1}$ steps the reduction algorithm halts and outputs the instance $(I,k)$. Note that in this case $f(k)n^c > n^{c+1}$, which implies that $f(k) > |I|$. Hence the size of the output instance is upper bounded by a function of $k$.

We now describe the solution lifting algorithm. If the reduction algorithm output $(I,k)$ then the solution lifting algorithm just returns the same solution that it gets as input. If the reduction algorithm output $(I',k')$ this means that the FPT algorithm terminated in polynomial time, which means that the solution lifting algorithm can use the FPT algorithm to output an optimal solution to $(I,k)$ in polynomial time, regardless of the solution to $(I',k')$ it gets as input. This concludes the proof.
\end{proof}

\paragraph{Parameterized Approximation and Approximate Kernelization.} 
For some parameterized optimization problems we are unable to obtain FPT algorithms, and we are also unable to find satisfactory polynomial time approximation algorithms. In this case one might aim for FPT-approximation algorithms, algorithms that run in time $f(k)n^c$ and provide good approximate solutions to the instance. 

\begin{definition}\label{def:fptAppx}
Let $\alpha \geq 1$ be constant. A fixed parameter tractable $\alpha$-approximation algorithm for a parameterized optimization problem $\Pi$ is an algorithm that takes as input an instance $(I,k)$, runs in time $f(k)|I|^{\OO(1)}$, and outputs a solution $s$ such that $\Pi(I,k,s) \leq \alpha \cdot OPT(I,k)$ if $\Pi$ is a minimization problem, and $\alpha \cdot \Pi(I,k,s) \geq OPT(I,k)$ if $\Pi$ is a maximization problem.
\end{definition}
Note that Definition~\ref{def:fptAppx} only defines constant factor FPT-approximation algorithms. The definition can in a natural way be extended to approximation algorithms whose approximation ratio depends on the parameter $k$, on the instance $I$, or on both.

We are now ready to define one of the key new concepts of the paper - the concept of an $\alpha$-approximate kernel. We defined kernels by first defining polynomial time pre-processing algorithms (Definition~\ref{def:polyTimePreProcess}) and then adding size constraints on the output (Definition~\ref{def:kernel}). In a similar manner we will first define $\alpha$-approximate polynomial time pre-processing algorithms, and then define $\alpha$-approximate kernels by adding size constraints on the output of the pre-processing algorithm.

\begin{definition}\label{def:polyTimePreProcessAppx}
Let $\alpha \geq 1$ be a real number and $\Pi$  be a parameterized optimization problem. An {\bf $\alpha$-approximate polynomial time pre-processing algorithm} ${\cal A}$ for $\Pi$ is a pair of polynomial time algorithms. The first one is called the {\bf reduction algorithm}, and computes a map ${\cal R}_{\cal A} : \Sigma^* \times \mathbb{N} \rightarrow \Sigma^* \times \mathbb{N}$. Given as input an instance $(I,k)$ of $\Pi$ the reduction algorithm outputs another instance $(I',k') = {\cal R}_{\cal A}(I,k)$.

The second algorithm is called the {\bf solution lifting algorithm}. This algorithm
takes as input an instance $(I,k) \in \Sigma^* \times \mathbb{N}$ of $\Pi$, the output instance $(I',k')$ of the reduction algorithm, and a solution $s'$ to the instance $(I',k')$. The solution lifting algorithm works in time polynomial in $|I|$,$k$,$|I'|$,$k'$ and $s'$, and outputs a solution $s$ to $(I,k)$. If $\Pi$ is a minimization problem then
\[
\frac{\Pi(I,k,s)}{OPT(I,k)} \leq \alpha \cdot \frac{\Pi(I',k',s')}{OPT(I',k')}.
\]
If $\Pi$ is a maximization problem then
\[
\frac{\Pi(I,k,s)}{OPT(I,k)} \cdot \alpha \geq \frac{\Pi(I',k',s')}{OPT(I',k')}.
\]
\end{definition}
Definition~\ref{def:polyTimePreProcessAppx} only defines constant factor approximate polynomial time pre-processing algorithms. The definition can in a natural way be extended approximation ratios that depend on the parameter $k$, on the instance $I$, or on both. Additionally, the discussion following Definition~\ref{def:polyTimePreProcess} also applies here. In particular we may assume that the solution lifting algorithm also gets as input a transcript of how the reduction algorithm obtains $(I',k')$ from $(I,k)$. The size of an $\alpha$-approximate polynomial time pre-processing algorithm is defined in exactly the same way as the size of a polynomial time pre-processing algorithm (from Definition~\ref{def:polyTimePreProcess}). 

\begin{definition}\label{def:appxkernel}
An {\bf $\alpha$-approximate kernelization} (or {\em $\alpha$-approximate kernel}) for a parameterized optimization problem $\Pi$, and real $\alpha \geq 1$, is an $\alpha$-approximate polynomial time pre-processing algorithm ${\cal A}$ such that $\text{size}_{\cal A}$ is upper bounded by a computable function $g : \mathbb{N} \rightarrow \mathbb{N}$.
\end{definition}
Just as for regular kernels, if the function $g$ in Definition~\ref{def:appxkernel} is a polynomial, we say that $\Pi$ admits an $\alpha$-approximate polynomial kernel. If $g$ is a linear, quadratic or cubic function, then $\Pi$ admits a linear, quadratic or cubic $\alpha$-approximate kernel, respectively. 

Proposition~\ref{prop:optKernelEquiv} establishes that a parameterized optimization problem $\Pi$ admits a kernel if and only if it is FPT. Next we establish a similar equivalence between \FPT-approximation algorithms and approximate kernelization.
\begin{proposition}\label{prop:appxKernelEquiv}
For every $\alpha \geq 1$ and decidable parameterized optimization problem $\Pi$, $\Pi$ admits a fixed parameter tractable $\alpha$-approximation algorithm if and only if $\Pi$ has an $\alpha$-approximate kernel.
\end{proposition}
The proof of Proposition~\ref{prop:appxKernelEquiv} is identical to the proof of Proposition~\ref{prop:optKernelEquiv}, but with the \FPT{} algorithm replaced by the fixed parameter tractable $\alpha$-approximation algorithm, and the kernel replaced with the $\alpha$-approximate kernel. 
On an intuitive level, it should be easier to compress an instance than it is to solve it. For $\alpha$-approximate kernelization this intuition can be formalized.
\begin{theorem}\label{prop:appxAlgConstKernelEquiv}
For every $\alpha \geq 1$ and decidable parameterized optimization problem $\Pi$, 
$\Pi$ admits a polynomial time $\alpha$-approximation algorithm if and only if $\Pi$ has an $\alpha$-approximate kernel of constant size.
\end{theorem}
The proof of Theorem~\ref{prop:appxAlgConstKernelEquiv} is simple; if there is an $\alpha$-approximate kernel of constant size one can brute force the reduced instance and lift the optimal solution of the reduced instance to an $\alpha$-approximate solution to the original. On the other hand, if there is a factor $\alpha$ approximation algorithm, the reduction algorithm can just output any  instance of constant size. Then, the solution lifting algorithm can just directly compute an $\alpha$-approximate solution to the original instance using the approximation algorithm.

We remark that Proposition~\ref{prop:appxKernelEquiv} and Theorem~\ref{prop:appxAlgConstKernelEquiv} also applies to approximation algorithms and  approximate kernels with super-constant approximation ratio.
We also remark that with our definition of $\alpha$-approximate kernelization, by setting $\alpha = 1$ we get essentially get back the notion of kernel for the same problem. The difference arises naturally from the different goals of decision and optimization problems. In decision problems we aim to correctly classify the instance as a ``yes'' or a ``no'' instance. In an optimization problem we just want as good a solution as possible for the instance at hand.
In traditional kernelization, a yes/no answer to the reduced instance translates without change to the original instance. With our definition of approximate kernels, a sufficiently good solution (that is, a witness of a yes answer) will always yield a witness of a yes answer to the original instance. However, the {\em failure} to produce a sufficiently good solution to the reduced instance does not stop us from {\em succeeding} at producing a sufficiently good solution for the original one. From the perspective of optimization problems, such an outcome is a win. 

\paragraph{Capping the objective function at $k+1$.}
We now return to the topic of parameterizing optimization problems by the value of the solution, and discuss the relationship between (approximate) kernels for such parameterized optimization problems and (traditional) kernels for the parameterized decision version of the optimization problem. 

Consider a traditional optimization problem, say {\sc Vertex Cover}. Here, the  input is a graph $G$, and the goal is to find a vertex cover $S$ of $G$ of minimum possible size. When {\em parameterizing} {\sc Vertex Cover} by the objective function value we need to provide a parameter $k$ such that solving the problem on the same graph $G$ becomes progressively easier as $k$ decreases. In parameterized complexity this is achieved by considering the corresponding {\em parameterized decision} problem where we are given $G$ and $k$ and asked whether there exists a vertex cover of size at most $k$. Here $k$ is the parameter.
If we also required an algorithm for {\sc Vertex Cover} to produce a solution, then the above parameterization can be interpreted as follows. Given $G$ and $k$, output a vertex cover of size at most $k$ or fail (that is, return that the algorithm could not find a vertex cover of size at most $k$.) If there exists a vertex cover of size at most $k$ then the algorithm is not allowed to fail.

A $c$-approximation algorithm for the {\sc Vertex Cover} problem is an algorithm that given $G$, outputs a solution $S$ of size no more than $c$ times the size of the smallest vertex cover of $G$. So, how do  approximation and parameterization mix? For $c \geq 1$, there are {\em two} natural ways to define a parameterized $c$-approximation algorithm for {\sc Vertex Cover}.
\begin{enumerate}
\setlength{\itemsep}{-2pt}
\item[(a)] Given $G$ and $k$, output a vertex cover of size at most $k$ or fail (that is, return that the algorithm could not find a vertex cover of size at most $k$.) If there exists a vertex cover of size at most $k/c$ then the  algorithm is not allowed to fail. 
\item[(b)] Given $G$ and $k$, output a vertex cover of size at most $ck$ or fail (that is, return that the algorithm could not find a vertex cover of size at most $ck$.) If there exists a vertex cover of size at most $k$ then the algorithm is not allowed to fail. 
\end{enumerate}

Note that if we required the approximation algorithm to run in {\em polynomial time}, then both   definitions above would yield exactly the definition of polynomial time $c$-approximation algorithms, by a linear search or binary search for the appropriate value of $k$. In the parameterized setting the running time depends on $k$, and the two formalizations are different, but nevertheless equivalent up to a factor $c$ in the value of $k$. That is $f(k)\cdot n^{\OO(1)}$ time algorithms and $g(k)$ size kernels for parameterization (b) translate to $f(ck)\cdot n^{\OO(1)}$ time algorithms and $g(ck)$ kernels for parameterization (a) and vice versa.

By defining the parameterized optimization problem for {\sc Vertex Cover} in such a way that the objective function depends on the parameter $k$, one can achieve either one of the two discussed formulations. By defining $VC(G,k,S) = \min\{|S|, k+1\}$ for vertex covers $S$ we obtain formulation (a). By defining 
$VC(G,k,S) = \min\{|S|, \lceil ck \rceil + 1\}$ for vertex covers $S$ we obtain formulation (b). It is more meaningful to define the computational problem {\em independently of the (approximation factor of) algorithms for the problem.} For this reason we stick to formulation (a) in this paper.


\paragraph{Reduction Rules and Strict $\alpha$-Approximate Kernels.} 
Kernelization algorithms in the literature~\cite{CyganFKLMPPS15,DF2012Book} are commonly described as a set of {\em reduction rules}. Here we discuss reduction rules in the context of parameterized optimization problems. A reduction rule is simply a polynomial time pre-processing algorithm, see Definition~\ref{def:polyTimePreProcess}. The reduction rule {\em applies} if the output instance of the reduction algorithm is not the same as the input instance. Most kernelization algorithms consist of a set of reduction rules. In every step the algorithm checks whether any of the reduction rules apply. If a reduction rule applies, the kernelization algorithm runs the reduction algorithm on the instance and proceeds by working with the new instance. This process is repeated until the instance is {\em reduced}, i.e. none of the reduction rules apply. To prove that this is indeed a kernel (as defined in Definition~\ref{def:kernel}) one proves an upper bound on the size of any reduced instance.

In order to be able to make kernelization algorithms as described above, it is important that reduction rules can be {\em chained}. That is, suppose that we have an instance $(I,k)$ and run a pre-processing algorithm on it to produce another instance $(I',k')$. Then we run another pre-processing algorithm on $(I',k')$ to get a third instance $(I^\star, k^\star)$. Given an optimal solution $s^\star$ to the last instance, we can use the solution lifting algorithm of the second pre-processing algorithm to get an optimal solution $s'$ to the instance $(I',k')$. Then we can use the solution lifting algorithm of the first pre-processing algorithm to get an optimal solution $s$ to the original instance $(I,k)$.

Unfortunately, one can not chain $\alpha$-approximate polynomial time pre-processing algorithms, as defined in Definition~\ref{def:polyTimePreProcessAppx}, in this way. In particular, each successive application of an $\alpha$-approximate pre-processing algorithm increases the gap between the approximation ratio of the solution to the reduced instance and the approximation ratio of the solution to the original instance output by the solution lifting algorithm. For this reason we need to define {\em strict} approximate polynomial time pre-processing algorithms.

\begin{definition}\label{def:strictPreprocess} Let $\alpha \geq 1$ be a real number, and $\Pi$ be a parameterized optimization problem. An $\alpha$-approximate polynomial time pre-processing algorithm is said to be {\bf strict} if, for every instance $(I,k)$, reduced instance $(I',k') = {\cal R}_{\cal A}(I,k)$ and solution $s'$ to $(I',k')$, the solution $s$ to $(I,k)$ output by the solution lifting algorithm when given $s'$ as input satisfies the following. 
\begin{itemize}\setlength\itemsep{-.7mm}
\item If $\Pi$ is a minimization problem then $\frac{\Pi(I,k,s)}{OPT(I,k)} \leq \max\left\{\frac{\Pi(I',k',s')}{OPT(I',k')}, \alpha \right\}$.
\item If $\Pi$ is a maximization problem then  $\frac{\Pi(I,k,s)}{OPT(I,k)} \geq \min\left\{\frac{\Pi(I',k',s')}{OPT(I',k')}, \frac{1}{\alpha} \right\}$.
\end{itemize}
\end{definition}

The intuition behind Definition~\ref{def:strictPreprocess} is that an $\alpha$-strict approximate pre-processing algorithm may incur error on near-optimal solutions, but that they have to preserve factor $\alpha$-approximation. If $s'$ is an $\alpha$-approximate solution to $(I',k')$ then $s$ must be a $\alpha$-approximate solution to $(I,k)$ as well. Furthermore, if the ratio of $\Pi(I',k',s')$ to $OPT(I',k')$ is {\em worse} than $\alpha$, then the ratio of $\Pi(I,k,s)$ to $OPT(I,k)$ should not be worse than the ratio of $\Pi(I',k',s')$ to $OPT(I',k')$.  

We remark that a reduction algorithm ${\cal R}_{\cal A}$ and a solution lifting algorithm that together satisfy the conditions of Definition~\ref{def:strictPreprocess}, also automatically satisfy the conditions of Definition~\ref{def:polyTimePreProcessAppx}. Therefore, to prove that  ${\cal R}_{\cal A}$ and solution lifting algorithm constitute a strict ${\alpha}$-approximate polynomial time pre-processing algorithm it is not necessary to prove that they constitute a ${\alpha}$-approximate polynomial time pre-processing algorithm first.
The advantage of Definition~\ref{def:strictPreprocess} is that strict $\alpha$-approximate polynomial time pre-processing algorithms do chain - the composition of two strict $\alpha$-approximate polynomial time pre-processing algorithms is again a strict $\alpha$-approximate polynomial time pre-processing algorithm. 

We can now formally define what a reduction rule is. A reduction rule for a parameterized optimization problem $\Pi$ is simply a polynomial time algorithm computing a map ${\cal R}_{\cal A} : \Sigma^* \times \mathbb{N} \rightarrow \Sigma^* \times \mathbb{N}$. In other words, a reduction rule is ``half'' of a polynomial time pre-processing algorithm. A reduction rule is only useful if the other half is there to complete the pre-processing algorithm. 
\begin{definition}\label{def:safe-Rule}
A reduction rule is said to be $\alpha$-{\bf safe for $\Pi$} if there exists a solution lifting algorithm, such that the rule together with the solution lifting algorithm constitute a strict $\alpha$-approximate polynomial time pre-processing algorithm for $\Pi$. A reduction rule is {\bf safe} if it is $1$-safe.
\end{definition}

In some cases even the final kernelization algorithm is a strict $\alpha$-approximate polynomial time pre-processing algorithm. This happens if, for example, the kernel is obtained only by applying $\alpha$-safe reduction rules. Strictness yields a tigher connection between the quality of solutions to the reduced instance and the quality of the solutions to the original instance output by the solution lifting algorithms. Thus we would like to point out which kernels have this additional property. For this reason we define strict $\alpha$-approximate kernels.
\begin{definition}\label{def:appxkernelStrict}
An $\alpha$-approximate kernel ${\cal A}$ is called {\bf strict} if ${\cal A}$ is a strict $\alpha$-approximate polynomial time pre-processing algorithm.
\end{definition} 

\paragraph{Polynomial Size Approximate Kernelization Schemes.} 
In approximation algorithms, the best one can hope for is usually an {\em approximation scheme}, that is an approximation algorithm that can produce a $(1+\epsilon)$-approximate solution for every $\epsilon > 0$. The algorithm runs in polynomial time for every fixed value of $\epsilon$. However, as $\epsilon$ tends to $0$ the algorithm becomes progressively slower in such a way that the algorithm cannot be used to obtain optimal solutions in polynomial time.

In the setting of approximate kernelization, we could end up in a situation where it is possible to produce a polynomial $(1+\epsilon)$-approximate kernel for every fixed value of $\epsilon$, but that the size of the kernel grows so fast when $\epsilon$ tends to $0$ that this algorithm cannot be used to give a polynomial size kernel (without any loss in solution quality). This can be formalized as a polynomial size approximate kernelization scheme.
\begin{definition}\label{def:psaks}
A {\bf polynomial size approximate kernelization scheme} (PSAKS) for a parameterized optimization problem $\Pi$ is a family of ${\alpha}$-approximate polynomial kernelization algorithms, with one such algorithm for every $\alpha > 1$.
\end{definition}
\noindent
Definition~\ref{def:psaks} states that a PSAKS is a {\em family} of algorithms, one for every $\alpha > 1$. However, many PSAKSes are {\em uniform}, in the sense that there exists an algorithm that given $\alpha$ outputs the source code of an $\alpha$-approximate polynomial kernelization algorithm for $\Pi$. In other words, one could think of a uniform PSAKS as a single ${\alpha}$-approximate polynomial kernelization algorithm where $\alpha$ is part of the input, and the size of the output depends on $\alpha$. 
From the definition of a PSAKS it follows that the size of the output instances of a PSAKS when run on an instance $(I,k)$ with approximation parameter $\alpha$ can be upper bounded by $f(\alpha) \cdot k^{g(\alpha)}$ for some functions $f$ and $g$ independent of $|I|$ and $k$. 
\begin{definition}\label{def:sizeEffPsaks}
A {\bf size efficient} PSAKS, or simply an {\bf efficient} PSAKS (EPSAKS) is a PSAKS such that the size of the instances output when the reduction algorithm is run on an instance $(I,k)$ with approximation parameter $\alpha$ can be upper bounded by $f(\alpha) \cdot k^c$ for a function $f$ of $\alpha$ and constant $c$ independent of $I$, $k$ 
and $\alpha$.
\end{definition}
\noindent
Notice here the analogy to efficient polynomial time approximation schemes, which are nothing but $\alpha$-approximation algorithms with running time $f(\alpha) \cdot n^c$. A PSAKS is required to run in polynomial time for every fixed value of $\alpha$, but the running time is allowed to become worse and worse as $\alpha$ tends to $1$. We can define {\em time-efficient} PSAKSes analagously to how we defined EPSAKSes.
\begin{definition}\label{def:timeEffPsaks}
A PSAKS is said to be {\bf time efficient} if (a) the running time of the reduction algorithm when run on an instance $(I,k)$ with approximation parameter $\alpha$ can be upper bounded by $f(\alpha) \cdot |I|^c$ for a function $f$ of $\alpha$ and constant $c$ independent of $I$, $k$, $\alpha$, and (b) the running time of the solution lifting algorithm when run on an instance $(I,k)$, reduced instance $(I',k')$ and solution $s'$ with approximation parameter $\alpha$ can be upper bounded by $f'(\alpha) \cdot |I|^c$ for a function $f'$ of $\alpha$ and constant $c$ independent of $I$, $k$ and $\alpha$.
\end{definition}
\noindent
Just as we distinguished between normal and strict $\alpha$-approximate kernels, we say that a PSAKS is {\bf strict} if it is a strict $\alpha$-approximate kernel for every $\alpha > 1$.

A {\em quasi-polynomial time} algorithm is an algorithm with running time $\OO(2^{(\log n)^c})$ for some constant $c$. In approximation algorithms, one is sometimes unable to obtain a PTAS, but still can make a $(1+\epsilon)$-approximation algorithm that runs in quasi-polynomial time for every $\epsilon > 1$. This is called a quasi-polynomial time approximation scheme. Similarly, one might be unable to give a PSAKS, but still be able to give a $\alpha$-approximate kernel of quasi-polynomial size for every $\alpha > 1$. 
\begin{definition}\label{def:qpsaks}
A {\bf quasi-polynomial size approximate kernelization scheme} (QPSAKS) for a parameterized optimization problem $\Pi$ is a family of ${\alpha}$-approximate kernelization algorithms, with one such algorithm for every $\alpha > 1$. The size of the kernel of the  ${\alpha}$-approximate kernelization algorithm should be upper bounded by $\OO(f(\alpha)2^{(\log k)^{g(\alpha)}})$ for functions $f$ and $g$ independent of $k$.
\end{definition}

%% file: CVC.tex
\section{Approximate Kernel for Connected Vertex Cover}\label{sec:cvc}
In this section we design a PSAKS for {\sc Connected Vertex Cover}. The parameterized optimization problem {\sc Connected Vertex Cover(CVC)} is defined as follows. 

\[
    CVC(G,k,S)=\left\{ 
\begin{array}{rl}
    \infty & \text{if $S$ is not a connected vertex cover of the graph $G$} \\
    \min\left\{|S|,k+1\right\} & \text{otherwise}
\end{array}\right.
\]

We show that {\sc CVC} has a polynomial size strict $\alpha$-approximate kernel for every $\alpha>1$. Let $(G,k)$ be the input instance. Without loss of generality assume that the input graph $G$ is connected. Let $d$ be the least positive integer such that $\frac{d}{d-1} \leq \alpha$. In particular, $d = \lceil \frac{\alpha}{\alpha-1} \rceil$. For a graph $G$ and an integer $k$, define $H$ to be the set of vertices of degree at least $k+1$. We define $I$ to be the set of vertices which are not in $H$ and whose neighborhood is a subset of $H$. 
That is $I =\{v\in V(G)\setminus H ~|~ N_G(v)\subseteq H\}$. 
The kernelization algorithm works by applying two reduction rules exhaustively. The first of the two rules is the following.
\begin{redrule}
\label{red:cvcdegree}
Let $v\in I$ be a vertex of degree $D \geq d$. Delete $N_G[v]$ from $G$ and add a vertex $w$ such that the neighborhood of $w$ is $N_G(N_G(v))\setminus \{v\}$. Then add $k$ degree $1$ vertices $v_1,\ldots,v_k$ whose neighbor is $w$. Output this graph $G'$, together with the new parameter $k'=k-(D-1)$.
\end{redrule}
\begin{lemma}
 Reduction Rule~\ref{red:cvcdegree} is \alphasafe.
\end{lemma}

\begin{proof}
To show that Rule~\ref{red:cvcdegree} is \alphasafe\ we need to give a solution lifting algorithm to go with the reduction. Given a solution $S'$ to the instance $(G',k')$, if $S'$ is a connected vertex cover of $G'$ of size at most $k'$ the algorithm returns the set $S = (S' \setminus  \{w, v_1, \ldots, v_k\}) \cup N_G[v]$. Otherwise the solution lifting algorithm returns $V(G)$. We now need to show that the reduction rule together with the above solution lifting algorithm constitutes a strict $\alpha$-approximate polynomial time pre-processing algorithm. 

First we show that $OPT(G',k') \leq OPT(G, k) - (D - 1)$. Consider an optimal solution $S^*$ to $(G, k)$. 
We have two cases based on the size of $S^*$. If $|S^*| > k$  then $CVC(G,k,S)=k + 1$; in fact $OPT(G, k)=k+1$. Furthermore,  any connected vertex cover of $G'$ has value at most $k' + 1 =k-(D-1)+1 \leq  OPT(G, k) - (D - 1)$. Now we consider the case when $|S^*|\leq k$. If $|S^*| \leq k$ then $N_G(v) \subseteq S^*$, since the degree of all the vertices in $N_G(v)$ is at least $k + 1$ and $S^*$ is a vertex cover of size at most $k$. Then $(S^* \setminus N_G[v]) \cup \{w\}$ is a connected vertex cover of $G'$ of size at most $|S^*| - (D - 1)=OPT(G, k) - (D - 1)$.

Now we show that $CVC(G,k,S) \leq CVC(G',k',S') + D$. If $S'$ is a connected vertex cover of $G'$ of size strictly more than $k'$ then $CVC(G,k,S) \leq k + 1=k'+D< k' +1 + D = CVC(G',k',S') + D$. Suppose now that $S'$ is a connected vertex cover of $G'$ of size at most $k'$. Then $w \in S'$ since $w$ has degree at least $k$ in $G'$. Thus $|S| \leq |S'| - 1 + D + 1 \leq |S'| + D $. Finally, $G[S]$ is connected because $G[N_G[v]]$ is connected and $N_G(N_G[v]) = N_{G'}(w) \setminus \{v_1, \ldots, v_k\}$. Hence $S$ is a connected vertex cover of $G$. Thus $CVC(G,k,S) \leq CVC(G',k',S') + D$. Therefore,  we have that
$$\frac{CVC{}(G,k,S)}{OPT(G,k)} \leq \frac{CVC{}(G',k',S') + D}{OPT(G',k') + (D - 1)} \leq \max\Bigg( \frac{CVC{}(G',k',S')}{OPT(G',k')},\alpha \Bigg).$$
The last transition follows from Fact~\ref{fact1}. 
This concludes the proof.
\end{proof}

The second rule is easier than the first, if any vertex $v$ has at least $k+1$ false twins, then remove $v$. A {\em false twin} of a vertex $v$ is a vertex $u$ such that $uv \notin E(G)$ and $N(u) = N(v)$.
\begin{redrule}\label{red:twins}
If a vertex $v$ has at least $k+1$ false twins, then remove $v$, i.e output $G' = G - v$ and $k' = k$.
\end{redrule}

\begin{lemma}
 Reduction Rule~\ref{red:twins} is \onesafe.
\end{lemma}
\begin{proof}
The solution lifting algorithm takes as input a set $S'$ to the reduced instance and returns the same set $S'=S$ as a solution to the original instance. To see that $OPT(G',k) \leq OPT(G, k)$, consider a smallest connected vertex cover $S^*$ of $G$. Again, we will distinguish between two cases either 
$|S^*|>k$ or $|S^*|\leq k$. If $|S^*| > k$ then 
$OPT(G',k) \leq k+1 = OPT(G, k)$. Thus, assume $|S^*| \leq k$. Then there is a false twin $u$ of $v$ that is not in $S^*$. Then $S^* \setminus \{v\} \cup \{u\}$ is a connected vertex cover of $G - v$ of size at most $k$.

Next we show that  $CVC(G,k,S) \leq CVC(G',k',S')$. If $|S'|>k'=k$ then clearly,   
$CVC(G,k,S)\leq k+1 =k'+1 = CVC(G',k',S')$. So let us assume that $|S'|\leq k$. 
Observe that, as $v$ has $k+1$ false twins, all vertices in $N(v)$ have degree at least $k+1$ in $G - v$. Thus, $N(v) \subseteq S' = S$ and $S$ is a connected vertex cover of $G$, and hence $CVC(G,k,S) \leq CVC(G',k',S')$. As a result, $$\frac{CVC(G,k,S)}{OPT(G,k)} \leq \frac{CVC(G',k',S')}{OPT(G',k')}$$ This concludes the proof.
\end{proof}

\begin{lemma}\label{lem:cvcUpperBound}
Let $(G,k)$ be an instance irreducible by rules~\ref{red:cvcdegree} and~\ref{red:twins}, such that $OPT(G,k) \leq k$. Then $|V(G)| \leq \OO(k^{d} + k^2)$.
\end{lemma}

\begin{proof}
Since $OPT(G,k) \leq k$, $G$ has a connected vertex cover $S$ of size at most $k$. We analyze separately the size of the three sets $H$, $I$ and $V(G) \setminus (H \cup I)$. First $H \subseteq S$ so $|H| \leq k$. Furthermore, every vertex in $I$ has degree at most $d-1$, otherwise Rule~\ref{red:cvcdegree} applies. Thus, there are at most ${k \choose d-1}$ different subsets $X$ of $V(G)$ such that there is a vertex $v$ in $I$ such that $N(v) = I$. Since each vertex $v$ has at most $k$ false twins it follows that $|I| \leq {k \choose d-1} \cdot (k+1) = \OO(k^d)$. 

Finally, every edge that has no endpoints in $H$ has at least one endpoint in $S \setminus H$. Since each vertex in $S \setminus H$ has degree at most $k$ it follows that there are at most $k|S| \leq k^2$ such edges. Each vertex that is neither in $H$ nor in $I$ must be incident to at least one edge with no endpoint in $H$. Thus there are at most $2k^2$ vertices in $V(G) \setminus (I \cup H)$ concluding the proof.
\end{proof}

\begin{theorem}\label{thm:cvcKernel}
\cvc{} admits a strict time efficient PSAKS with $\OO(k^{\lceil \frac{\alpha}{\alpha-1} \rceil} + k^2)$ vertices.
\end{theorem}

\begin{proof}
The kernelization algorithm applies the rules \ref{red:cvcdegree} and~\ref{red:twins} exhaustively. If the reduced graph $G$ has more than $\OO(k^{d} + k^2)$ vertices then, by Lemma~\ref{lem:cvcUpperBound}, $OPT(G,k) = k+1$ and the algorithm may return any conneccted vertex cover of $G$ as an optimal solution. Thus the reduced graph has at most $\OO(k^{d} + k^2)$ vertices, since $d = \lceil \frac{\alpha}{\alpha-1} \rceil$ the size bound follows. The entire reduction procedure runs in polynomial time (independent of $\alpha$), hence the PSAKS is time efficient.
\end{proof}

%% file: labelled_interval_graph.tex

\section{Disjoint Factors and Disjoint Cycle Packing}\label{sec:DfAndCp}
In this section we give PSAKes for  \df{} and \CP. The main ingredient of our lossy kernels is a combinatorial object 
that ``preserves'' labels of all the independent sets of a labelled graph.  We will make this precise in the next section and then use this crucially to design PSAKes for both  \df{} and \CP. 

\subsection{Universal independent set covering}
We start the subsection by defining a combinatorial object, which we call,  
{\em \ULI{$\epsilon$} (\shortuli{$\epsilon$})}.  After formally defining it, we give an efficient 
construction for finding this objects when the input graph enjoys some special properties. 
Informally, \shortuli{$\epsilon$} of a labelled graph $G$ 
is an induced subgraph of $G$ which preserves approximately 
{\em all the labelled independent sets}. The formal definition is given below.  
Here, $\epsilon>0$ is a fixed constant. 

\smallskip

\defproblem{\scULI{$\epsilon$} (\shortULI{$\epsilon$})}{A graph $G$, an integer $q\in {\mathbb N}$ 
and a labelling function $\Gamma~:~V(G)\rightarrow [q]$}{A subset $X\subseteq V(G)$ such that 
for any independent set $S$ in $G$, there is an independent set $S'$ in $G[X]$ with 
$\Gamma(S')\subseteq \Gamma(S)$ and $\vert \Gamma(S')\vert\geq  (1-\epsilon)\vert \Gamma(S)\vert $. The set $X$ is called \shortuli{$\epsilon$}.}

 \smallskip
  
Obviously, for any $\epsilon>0$ and a labelled graph $G$,  the whole graph $G$ itself is an 
\shortuli{$\epsilon$}.
Our objective here is to give \shortuli{$\epsilon$} with size as small as possible. Here, we design a polynomial time algorithm  which gives an \shortuli{$\epsilon$} for an interval graph $G$ 
of size at most $(q\cdot \chi(G))^{\OO(\frac{1}{\epsilon}\log \frac{1}{\epsilon})}$.   Here, $\chi(G)$ denotes the 
chromatic number of the graph $G$. 

\paragraph{\shortuli{$\epsilon$} for Interval Graphs.}
From now onwards, in this subsection whenever we will use graph we mean {\em an interval graph}. 
We also assume that we have an interval representation of the graph we are considering. 
We use the terms {\em vertex} as well as {\em interval}, interchangeably,  to denote the  
vertex of an interval graph. 
Let $(G,q,\Gamma)$ be an input instance of 
\shortULI{$\epsilon$}. 
 We first compute a proper 
coloring $\kappa : V(G)\rightarrow [\chi(G)]$ of $G$. 
It is well known that  a proper coloring of an interval graph with the minimum number of colors 
can be computed in polynomial time~\cite{Cormen01}. 
Now using the proper coloing function $\kappa$, we refine the labelling $\Gamma$ to another labelling 
$\Lambda$ of $G$ as follows: for any $u\in V(G)$, $\Lambda(u)=(\Gamma(u),\kappa(u))$. 
An important property of the labelling 
$\Lambda$ is the following: for any $i\in \{(a,b)~|~a\in [q], b\in[\chi(G)] \}$, $\Lambda^{-1}(i)$ is an {\em independent set} 
in $G$. From now onwards, we will assume that we are working with the labelling function 
$\Lambda~:~V(G)\rightarrow [k]$, where \fbox{$k=q \cdot \chi(G)$} and 
$\Lambda^{-1}(i)$ is an independent set  in $G$ for any $i\in [k]$. 
We say that a subset of labels $Z\subseteq [k]$ is {\em realizable} in a graph $G$, if there exists 
an independent set $S\subseteq V(G)$ such that $\Lambda(S)=Z$. We first show that an 
\shortuli{$\epsilon$} for $G$ with respect to the labelling $\Lambda$ refining $\Gamma$ is also an 
\shortuli{$\epsilon$} for $G$ with respect to the labelling $\Gamma$.
\begin{lemma}
\label{lem:equivlabel}
Let $X$ be a vertex subset of $G$. If $X$ is an \shortuli{$\epsilon$} for $(G,\Lambda)$ then it is also  an \shortuli{$\epsilon$} for $(G,\Gamma)$.
\end{lemma}
\begin{proof}
Let $I$ be an independent set of $G$ and let $\Gamma(I)$ denote the set of labels on the vertices of $I$. 
We first compute a subset $I'$ of $I$ by selecting exactly one vertex from $I$ for each label in  
$\Gamma(I)$. Clearly, $\Gamma(I)=\Gamma(I')$. Observe that since any label is used 
at most once on any vertex in $I'$ we have that $|I'|=|\Gamma(I')|=\vert \Lambda(I')\vert $. In particular,  
for any pair of vertices $u, v\in I'$, $\Gamma(u)\neq \Gamma(v)$. 
By the property of the set $X$, we have 
an independent set $S'$ in $G[X]$ with $\Lambda(S')\subseteq \Lambda(I')$ and 
$\vert \Lambda(S')\vert\geq  (1-\epsilon)\vert \Lambda(I')\vert $.  Since, for any pair of vertices 
$u, v\in I'$, $\Gamma(u)\neq \Gamma(v)$, we have that  $\Gamma(S') \subseteq \Gamma(I')=\Gamma(I)$ and $\vert \Gamma(S')\vert\geq  (1-\epsilon)\vert \Gamma(I')\vert =(1-\epsilon)\vert \Gamma(I)\vert $. 
This concludes the proof. 
 \end{proof}

Lemma~\ref{lem:equivlabel} implies that we can assume that the input labelling is also a proper coloring of $G$ by increasing the number of labels by a multiplicative factor of $\chi(G)$. We first define notions of 
 {\em rich} and {\em poor} labels; which will be crucially used in our algorithm. 
 
\begin{definition}
 For any induced subgraph $H$ of $G$  we say that a  
label $\ell \in [k]$ is {\em rich} in $H$, if there are at least $k$ vertices in $H$ that are  
labelled $\ell$. Otherwise, the label $\ell$ is called {\em poor} in $H$. 
\end{definition}

We start with a simple lemma that shows that in an interval graph all the rich labels are realizable by an independent set of $G$. In particular we show the following lemma.

\begin{lemma}
\label{obs:rich:greedy:packing}
Let  $H$ be an induced subgraph of $G$, $\Lambda$ be a labelling function as defined above 
and $R$ be the set of rich labels in $H$.
Then $R$ is realizable in $H$. Moreover, an independent set $S$ such that 
$\Lambda(S)=R$ can be computed in polynomial time. 
\end{lemma} 
\begin{proof}
For our proof we will design an algorithm which constructs an independent set $S$ such that 
$\Lambda(S)=R$. We assume that we have an interval representation of $H$ and for 
any  vertex $v\in V(H)$, let $I_v$ be the interval corresponding to the vertex $v$. 
Our algorithm is recursive and as an input takes the tuple $(H, \Lambda, R)$.  In the base case it checks whether there is a vertex $w\in V(H)$ such that 
$\Lambda(w)\in R$. If there is no $w$ such that $\Lambda(w)\in R$ then the algorithm outputs 
an $\emptyset$. Otherwise, 
pick a vertex $I_{u}$  in $H'$ such that $\Lambda(u)\in R$ and the value of the 
right endpoint of the interval $I_u$ is minimum among all the vertices that are labelled with 
labels from $R$
in $H$.
Having found $I_u$ (or $u$), 
we recursively solve the problem on the input $(H'=H-N[u], \Lambda|_{V(H')}, R\setminus \{\Lambda(u)\})$. 
Here,    $\Lambda|_{V(H')}$ is the labelling $\Lambda$ restricted to the vertices in $V(H')$. 
Let $S'$ be the output of the recursive call on the input  $(H', \Lambda|_{V(H')}, R\setminus \{\Lambda(u)\})$. Our algorithm will output $S'\cup \{u\}$. 

Now we prove the correctness of the algorithm. Towards this we prove 
the following statement using induction on $\vert R \vert$: 
for any induced subgraph  $H$ of $G$ and 
$R\subseteq [k]$ 
such that for any $j\in R$, the  number of vertices in $H$ labelled with $j$ is at least $\vert R\vert$, 
then the above  algorithm on input $(H,\Lambda,R)$ will output an independent set $S$ such that $\Lambda(S)=R$. 
The base case is when $\vert R\vert=0$, and statement holds trivially.  Now consider the 
induction step. 
Let $u$ be the vertex picked by the algorithm such that $\Lambda(u)\in R$
and the value of the right endpoint of the interval $I_u$, corresponding to $u$, is the minimum among all 
such intervals. 
Since for any $j\in R$, $\Lambda^{-1}(j)$  is independent, and 
$I_u$ is an interval (vertex) 
with minimum right endpoint value, 
we have that for any $i\in R$, the number of intervals labelled with $i$ that intersects with $I_u$ is at most 
$1$. This implies that for any $i\in R \setminus \{\Lambda(u)\}$, the number of vertices labelled with $i$ in $H-N[u]$ is 
at least $\vert R\vert -1$ (because the number of vertices labelled with $i$ in $H$ is at least $\vert R\vert$). 
Hence, by the induction hypothesis, the  recursive call on the input  $(H'=H-N[u], \Lambda|_{V(H')}, R\setminus \{\Lambda(u)\})$ will output an independent set $S'$ such that $\Lambda(S')=R\setminus \{\Lambda(u)\}$. 
Since $S'\cap N[u]=\emptyset $, we have that $S'\cup \{u\}$ is the required 
independent set for the input $(H,\Lambda, R)$. This completes the correctness proof.  
\end{proof}

Before we give the formal construction for the desired \shortuli{$\epsilon$} for $G$, we first give an intuitive explanation of our strategy. 
If we are seeking for an upper bound on \shortuli{$\epsilon$} in terms of $k$, then Lemma~\ref{obs:rich:greedy:packing} suggests the following natural strategy: for rich labels, we find an independent set, say $I_{\sf rich}$,  of size at most $k$ (as the number of labels itself is upper bounded by $k$)  that realizes it and add all the vertices in this set to \shortuli{$\epsilon$} we are constructing. Let us denote the the  \shortuli{$\epsilon$} we are constructing by  $X$. 
For the poor labels, we know that by definition each label appears on at most $k$ vertices and thus in total the number of vertices that have poor labels is upper bounded by $k^2$. We include all the vertices that have poor labels to \shortuli{$\epsilon$} (the set $X$)  we are constructing. So at this stage if $G$ has an independent set 
$S$ such that all the labels used on the vertices in $S$ ($\Lambda(S)$) are rich then we can find an appropriate independent subset of $I_{\sf rich}$ that realizes all the labels of $\Lambda(S)$. On the other hand, if we have an independent set $S$ such that all the labels used on the vertices in $S$ are poor 
then it self realizes itself. That is, since we have kept all the vertices that are poor, the set $S$ itself is a contained inside the  \shortuli{$\epsilon$} we are constructing and thus it self realizes itself. The problem arises when we have an independent set $S$ that has vertices having both rich labels as well as poor labels. We deal with this case essentially by the following case distinctions. Let 
$\Lambda(S)$ denote  the set of labels used on the vertices in $S$ and $\Lambda(S)_{\sf rich}$ and $\Lambda(S)_{\sf poor}$ denote the set of rich and poor labels in $\Lambda(S)$, respectively. 

\begin{figure}[t]
\begin{center}
\includegraphics[height=10cm,width=14cm]{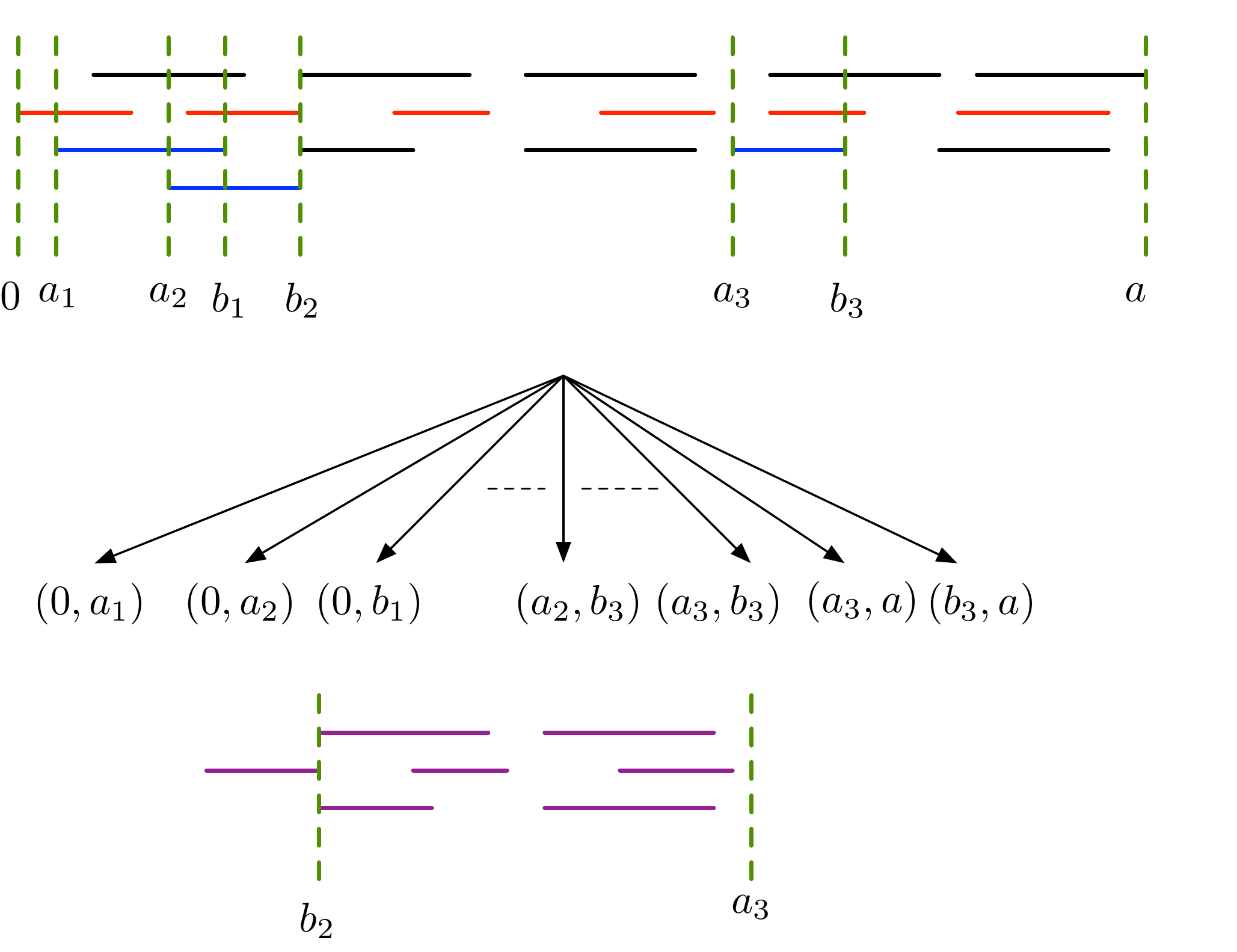}
\end{center}
\caption{An illustration for the process of obtaining \shortuli{$\epsilon$}. Intervals colored with red denote $I_{\sf rich}$ and intervals colored with blue denote vertices labelled with poor label. The instance below corresponds to branching on $(a_2,b_3)$. \label{fig:illus-ip-branchingalgo}}
\end{figure}

\begin{enumerate}
\setlength{\itemsep}{-2pt}
\item If $|\Lambda(S)_{\sf rich}| \geq (1-\epsilon) |\Lambda(S)|$, then we are again done as we can find an 
appropriate independent subset of $I_{\sf rich}$ that realizes all  the labels of $\Lambda(S)_{\sf rich}$. 
\item Since the first case does not arise we have that the number of rich labels in $\Lambda(S)$, that is, 
$|\Lambda(S)_{\sf rich}|$ is upper bounded by $(1-\epsilon) |\Lambda(S)|$  and that 
$|\Lambda(S)_{\sf poor}|\geq \epsilon |\Lambda(S)|$. Thus, in this case 
it is possible that $|\Lambda(S)_{\sf poor}|=|\Lambda(S)_{\sf rich}|=\frac{1}{2}|\Lambda(S)|$ and hence it is {\em possible} that there is no independent set  $S'$ in  $G[X]$ (the set $X$ constructed so far)  with 
$\Lambda(S')\subseteq \Lambda(S)$ and $\vert \Lambda(S')\vert\geq  (1-\epsilon)\vert \Lambda(S)\vert $. 
Thus, we need to enrich the set $X$ further. Towards this we use the following strategy. Let $Q$ be the set of endpoints of the intervals labelled with poor labels. Furthermore, assume that all the intervals of $G$ are between $(0,a)$. Now for every $p,q\in Q\cup\{0,a\}$, let $Y_{p,q}$ denote the set of intervals of $G$ which is fully contained in the open interval $(p,q)$. For every, $p,q\in Q\cup\{0,a\}$, we recursively find the desired \shortuli{$\epsilon$}  in $G[Y_{p,q}]$ and then take the union. Clearly, this is a branching algorithm with every node in the recursion tree having $\OO(k^4)$ children. See Figure~\ref{fig:illus-ip-branchingalgo} for an illustration of the process. The {\em idea} of this branching procedure is that given 
an  independent set $S$ we would like to pack all the vertices in $S$ with poor labels and then having made this choice we get disjoint induced subgraph of $G$ (by removing all the vertices in $S$ with poor labels and their neighorhiood)  where we ``would like to pack'' the vertices in $S$ that have rich labels. By our construction it is evident that the set $S'$ we will obtain by packing labels in the disjoint induced subgraphs of $G$ is compatible with the choice of packing all the vertices with poor labels in $S$. To get an upper bound on the size of the set $X$ we are constructing we show that the recursion tree can be truncated at the depth of $\lceil \OO(\frac{1}{\epsilon}\log \frac{1}{\epsilon} )\rceil$
\end{enumerate}

Now we give our main lemma that gives an algorithm for finding the desired \shortuli{$\epsilon$} for $G$ with respect to the labelling function $\Lambda$.

\begin{lemma}
\label{lemma:interval_graph}
Let $G$ be an interval graph, $k\in {\mathbb N}$ and $\epsilon'>0$. 
Let $\Lambda~:~V(G)\rightarrow [k]$ be a labelling function such that for any $i\in [k]$, 
$\Lambda^{-1}(i)$ is an independent set in $G$.  Then, there is a polynomial time algorithm that finds a set $X\subseteq V(G)$ 
of cardinality $k^{\OO(\frac{1}{\epsilon'}\log \frac{1}{\epsilon'})}$ such that for any realizable set $Z\subseteq [k]$ in $G$, 
there is a realizable subset $Z'\subseteq Z$  of cardinality at least $(1-2\epsilon')\vert Z \vert$ 
 in $G[X]$. 
\end{lemma}
\begin{proof}
Our  polynomial time algorithm is a bounded depth recursive procedure.
Towards that we define a recursive marking procedure 
{\sc Mark-Interval} 
which takes as input an induced subgraph of $G$ 
and a positive integer and marks intervals of $G$. Vertices corresponding to the marked intervals will correspond to the desired set $X$. Our algorithm is called {\sc Mark-Interval}. See Algorithm~\ref{alg:mark:int} for a detailed formal description of the algorithm. 
We call the procedure 
{\sc Mark-Interval} 
on input $(G,d=\lceil \frac{1}{\epsilon'}\log \frac{1}{\epsilon'}\rceil )$ to get the required set $X$, which is the set of vertices marked by the procedure. 
Without loss of generality we assume that all the intervals in $G$ are contained in $(0,a)$ for some $a\in {\mathbb N}$. 

\begin{algorithm}[t]
\If {$d'=1$}
\Return
Let $R$ be the set of  rich labels in $H$.\\
Apply the algorithm  mentioned in Lemma~\ref{obs:rich:greedy:packing} 
and let $S$ be its output (Note that $S$ is an independent set and $\Lambda(S)=R$).\label{step:rich:computing} \\
Mark all the intervals in $S$. \label{step:rich:marking} \\
Let $P$ be the set of intervals which are labelled with poor labels in $H$.\\
Mark all the intervals in $P$. \label{step:poor:marking}\\ 
Let $Q$ be the set of endpoints of the intervals in $P$. \\
\ForAll{ $p,q\in Q\cup\{0,a\}$}
{
{\sc Mark-Interval}($H[Y_{p,q}],d'-1$), where $Y_{p,q}$ is the set of intervals of $H$ which is fully contained in the open interval $(p,q)$.
} 
\caption{{\sc Mark-Interval} ($H$, $d'$), where $H$ is an induced subgraph of $G$ and  $d'\in {\mathbb N}$}
\label{alg:mark:int}
\end{algorithm}
 
We first show that the procedure {\sc Mark-Interval} on input $(G,d)$ 
marks at most $k^{\OO(d)}=k^{\OO(\frac{1}{\epsilon'}\log \frac{1}{\epsilon'})}$ intervals. Let $X$ be the set of marked intervals. 
In Step~\ref{step:rich:marking}, {\sc Mark-Interval} marks at most $k$ intervals, one 
for each rich label in $G$. In Step~\ref{step:poor:marking}, {\sc Mark-Interval} mark 
all intervals which are labelled with poor labels in $G$ and the number of such intervals 
is at most $k^2$. This implies that number of  points in $Q$ is at most $2k^2+2$. 
Hence the procedure makes at most $\binom{2k^2+2}{2}$ recursive calls. 
Thus, the total number 
of marked intervals  is bounded by the recurrence relation,   
$T(d)\leq (k^2+k)+\binom{2k^2+2}{2} T(d-1)$ and $T(1)=0$. This recurrence relation solves to 
$k^{\OO(d)}$.
This implies 
that the cardinality of the set of marked vertices by  {\sc Mark-Interval}$(G,\lceil \frac{1}{\epsilon'}\log \frac{1}{\epsilon'}\rceil)$ 
is at most $k^{\OO(\frac{1}{\epsilon'}\log \frac{1}{\epsilon'})}$.

Now we show the correctness of the algorithm. 
Towards that we first prove the following claim. 
\begin{claim}
\label{claim:markingbound}
Let $H$ be an  induced subgraph of $G$,  $d'\leq d$ be a positive integer and $X'$ be the set of 
marked vertices by the procedure  {\sc Mark-Interval} on input $(H,d')$. 
If $W\subseteq [k]$ is realizable in $H$, 
then there is a subset $W'\subseteq W$ such that $W'$ is realizable in $H[X']$ 
and $\vert W' \vert \geq (1-\epsilon'-(1-\epsilon')^{d'}) \vert W\vert$. 
\end{claim}
\begin{proof}
We prove the claim using induction on $d'$. The base case is when $d'=1$. When $d'=1$,  
$(1-\epsilon'-(1-\epsilon')^{d'}) \vert W\vert=0$ and empty set is the required set $W'$ of labels. 
Now consider the induction step. We assume that the claim is true for any $1\leq d''<d'$.  If 
at least $(1-\epsilon')\vert W\vert$ labels in $W$ are rich then in Step~\ref{step:rich:computing}, 
the procedure {\sc Mark-Interval} computes an independent set $S$ such that $\Lambda(S)$ 
is the set of all rich labels in $H$ and vertices in $S$ is marked in Step~\ref{step:rich:marking}. 
This implies that at leasts $(1-\epsilon')\vert W\vert \geq (1-\epsilon'-(1-\epsilon')^{d'}) \vert W\vert$ labels in $W$
are realizable in $H[X']$. Now we are in the case where strictly less than $(1-\epsilon')\vert W\vert$ labels in 
$W$ are rich. That is, the number of poor labels contained in $W$ appearing on the vertices of $H$ is at least $\epsilon'\vert W \vert$. Let 
$U$ be an independent set in $H$ such that $\Lambda(U)=W$ and let $U_p$ be the subset of $U$ which are 
labeled with poor labels from $H$. Notice that $\vert U_p\vert\geq \epsilon'\vert W \vert$. In Step~\ref{step:poor:marking}, procedure  {\sc Mark-Interval} 
marks all the intervals in $U_p$. Let $[a_1,b_1],\ldots,[a_{\ell},b_{\ell}]$ be the set of intervals in $U_p$ 
such that $a_1<b_1<a_2<b_2<\ldots < b_{\ell}$. All the intervals in $U\setminus U_p$ are disjoint 
from $U_p$.  That is, there exists a family of intervals $\{V_0,V_1,\ldots,V_{\ell}\}$ such that $\bigcup_{i=0}^{\ell} V_i= U \setminus U_p$
and for any $i\in \{0,\ldots,\ell\}$, the intervals in $V_i$ are contained in $(b_i,a_{i+1})$, where $b_0=0$ 
and $a_{\ell+1}=a$. The recursive procedure  {\sc Mark-Interval} on input $(H,d')$ calls recursively with 
inputs $(H[Y_{b_i,a_{i+1}}],d'-1)$, $i\in \{0,\ldots,\ell\}$. Here $V_i\subseteq Y_{b_i,a_{i+1}}$. 
Let $W_i=\Lambda(V_i)$. Notice that $W_i\cap W_j=\emptyset$ for $i\neq j$ 
and $\Lambda(U_p)\cup \bigcup_{i=0}^{\ell} W_i =W$. 
By induction hypothesis,  
for any $i\in\{0,\ldots,\ell\}$, there exists $W_i' \subseteq W_i \subseteq W$ such that 
$\vert W_i' \vert \geq (1-\epsilon'-(1-\epsilon')^{d'-1}) \vert W_i\vert$ and $W_i' $ is realizable in $H[X_i]$ 
where $X_i$ is the set of vertices marked by {\sc Mark-Interval}$(H[Y_{b_i,a_{i+1}}],d'-1)$. 
This implies that $\Lambda(U_p)\cup \bigcup_{i=0}^{\ell} W_i' $ is  realizable in $H[X]$. Now we lower bound the size of $\Lambda(U_p)\cup \bigcup_{i=0}^{\ell} W_i' $. 
\begin{eqnarray*}
\vert \Lambda(U_p)\cup \bigcup_{i=0}^{\ell} W_i'\vert &\geq& \vert \Lambda(U_p)\vert + \sum_{i=0}^{\ell}  (1-\epsilon'-(1-\epsilon')^{d'-1}) \vert W_i\vert \\
&=&\vert \Lambda(U_p)\vert + (1-\epsilon'-(1-\epsilon')^{d'-1}) \sum_{i=0}^{\ell}   \vert W_i\vert \\
&=& \vert \Lambda(U_p)\vert + (1-\epsilon'-(1-\epsilon')^{d'-1})\vert W \setminus \Lambda(U_p)\vert + \vert W \setminus \Lambda(U_p)\vert - \vert W \setminus \Lambda(U_p)\vert  \\
&\geq& \vert W \vert - (\epsilon'+(1-\epsilon')^{d'-1})\vert W \setminus \Lambda(U_p)\vert \\
&\geq& \vert W \vert - \epsilon' \vert W \vert -(1-\epsilon')^{d'-1}\vert W \setminus \Lambda(U_p)\vert \\
&\geq& \vert W \vert - \epsilon' \vert W \vert -(1-\epsilon')^{d'}\vert W \vert \qquad\quad(\mbox{Because $\vert W \setminus \Lambda(U_p)\vert < (1-\epsilon') \vert W \vert$})\\
&\geq&(1-\epsilon'-(1-\epsilon')^{d'}) \vert W\vert
\end{eqnarray*}
This completes the proof of the claim.
\end{proof}

Let $Z\subseteq [k]$ be a set of labels which is realizable in $G$. 
Now, by Claim~\ref{claim:markingbound}, we have that there exists $Z'\subseteq Z$ such that 
$Z'$ is realizable in $G[X]$ and $\vert Z'\vert \geq (1-\epsilon'-(1-\epsilon')^{\frac{1}{\epsilon'}\log \frac{1}{\epsilon'}}) \vert Z\vert\geq (1-2\epsilon')\vert Z \vert$.  This completes the proof.
\end{proof}
Now we are ready to prove the main result of this section.

\begin{lemma}
\label{lem:uni_lab_is}
Let $G$ be an interval graph,  $q\in {\mathbb N}$, $\epsilon>0$, and $\Gamma~:~V(G)\rightarrow [q]$ 
be a labelling function. Then there is a polynomial time algorithm which
finds a set $X\subseteq V(G)$ 
of cardinality $(q\cdot \chi(G))^{\OO(\frac{1}{\epsilon}\log \frac{1}{\epsilon})}$ such that $X$ 
\shortuli{$\epsilon$} of $G$.
\end{lemma}
\begin{proof}
We start by refining the labelling $\Gamma$ to $\Lambda$ such that $\Lambda$ is a proper coloring of $G$. As explained before, we first compute a proper 
coloring $\kappa : V(G)\rightarrow [\chi(G)]$ of $G$ in polynomial time~\cite{Cormen01}. 
Now using the proper coloing function $\kappa$ and the labelling $\Gamma$, we define labelling 
$\Lambda$ of $G$ as follows: for any $u\in V(G)$, $\Lambda(u)=(\Gamma(u),\kappa(u))$.  Now we set $\epsilon'=\frac{\epsilon}{2}$  and apply 
Lemma~\ref{lemma:interval_graph} on  $G$, 
$\Lambda$, $k=q\cdot \chi(G)$ and $\epsilon'$ to get a set $X\subseteq V(G)$ 
of cardinality $k^{\OO(\frac{1}{\epsilon'}\log \frac{1}{\epsilon'})}$ such that for any realizable set $Z\subseteq [k]$ in $G$, 
there is a realizable subset $Z'\subseteq Z$  of cardinality at least 
$(1-2\epsilon')\vert Z \vert$  in $G[X]$. That is, $X\subseteq V(G)$ is 
of cardinality $k^{\OO(\frac{1}{\epsilon}\log \frac{1}{\epsilon})}$ such that for any realizable set $Z\subseteq [k]$ in $G$, there is a realizable subset $Z'\subseteq Z$  of cardinality at least 
$(1-\epsilon)\vert Z \vert$  in $G[X]$. Note that a set $X$ is \shortuli{$\epsilon$} for $G$ if and only if for any realizable set $Z\subseteq [k]$ in $G$, there is a realizable subset $Z'\subseteq Z$  of cardinality at least 
$(1-\epsilon)\vert Z \vert$  in $G[X]$. This implies that $X$ is  \shortuli{$\epsilon$} for $(G,\Lambda)$. However, 
by Lemma~\ref{lem:equivlabel}, we know that if  $X$ is an \shortuli{$\epsilon$} for $(G,\Lambda)$ then it is also  an \shortuli{$\epsilon$} for $(G,\Gamma)$. This concludes the proof.
%
\end{proof}

%% file: disjoint_factor.tex
\subsection{Disjoint Factors}

In this section, we give a PSAKS for the parameterized optimization problem {\sc Disjoint Factors (DF)}. 
To define this problem we first need to set up some definitions. 
For a string $L=a_1a_2\ldots a_n$ over an alphabet $\Sigma$, we use $L[i,j]$, where $1\leq i\leq j\leq n$, to denote the substring 
$a_i\ldots a_j$. In this section we would like to distinguish between two substrings 
$L[i,j]$ and $L[i',j']$, where $i\neq i'$ or $j\ne j'$, even if the string $L[i,j]$ is exactly same as 
the string $L[i',j']$. Thus we call $L'$ is a ``position substring'' of $L$, to emphasize 
$L'$ is substring of $L$ associated with two indices. We say two position substrings $L_1$ and 
$L_2$ are disjoint if they do not overlap (even at the starting or at the ending of the substrings). 
For example $L[i,j]$ and $L[j,j']$ are overlapping and not disjoint.  
We say that  a string $L'$ is a {\em string minor} of $L$, if $L'$ can be obtained from $L$ by deleting some 
position substrings of $L$. 
A {\em factor} of a string $L$ is a position substring 
of length at least $2$ which starts and ends with the same letter (symbol). A factor is called {\em $x$-factor} if the factor starts and end at a letter $x\in \Sigma$. 
Two factors are called distinct if they start at different letters. A set $\SSS$ of factors in $L$ is called 
a set of {\em disjoint factors} if each pair of factors
in $\SSS$ are {\em disjoint and distinct}. That is, no two factors in $\SSS$ start at the same letter and 
pairwise they do not overlap. This immediately implies that for any string $L$ over $\Sigma$, the 
cardinality of any set of {\em disjoint factors} is at most $|\Sigma|$. 
For a set of disjoint factors $\SSS$ of $L$ and $\Sigma'\subseteq \Sigma$, we
say that $\SSS$ is a $\Sigma'$-factor if for each
element $x$ in $\Sigma'$, there is a factor in $\SSS$, starting and ending at $x$.
%

In the {\disfac} problem, introduced in \cite{BodlaenderTY09}, input is an alphabet $\Sigma$ and a string $L$ in $\Sigma^*$, the task is to find a maximum cardinality set of disjoint factors in $L$.
%
%
Bodlaender {\em et al.}~\cite{BodlaenderTY09} proved that {\disfac} is \NP-complete by reduction from $3$-SAT, and also that {\disfac} parameterized by $|\Sigma|$ does not admit a polynomial kernel unless \coNPsubNPbyPoly. The reduction of Bodlaender {\em et al.} started from a {\em gap} variant of $3$-SAT where every variable appears in at most a constant number of clauses~\cite{Trevisan01} shows that  {\disfac} is in fact \APX{}-hard, which means that it does not admit a PTAS unless \classP{} $=$ \NP{}. We will consider {\disfac} when parameterized by the alphabet size $|\Sigma|$. Formally, the parameterized optimization problem that we consider is defined as follows. 

\[
    DF(L,|\Sigma|,\SSS)= 
\begin{cases}
    -\infty & \text{if $\SSS$ is not a set of disjoint factors of $L$} \\
    |\SSS| & \text{otherwise}
\end{cases}
\]
We remark that in the original definition of {\df} of Bodlaender {\em et al.}~\cite{BodlaenderTY09}, the objective is simply to decide whether it is possible to find $|\Sigma|$ disjoint factors in the input string $L$.
The variant of the problem discussed here is the natural maximization variant of the original problem. Next we give a PSAKS for this problem, in other words a polynomial size $\alpha$-approximate kernel for any $\alpha>1$.

\begin{definition}
 Let ${\cal S}=\{S_1,\ldots,S_t\}$ be a set of mutually disjoint position substrings of  a string $L$. 
 Then we use $L/\SSS$ to denote the string obtained from $L$ after deleting all position substrings 
 in $\SSS$. For example if $L=a_1\cdots a_{11}$ and $\SSS=\{L[2,4], L[7,9]\}$, then 
 $L/\SSS=a_1a_5a_6a_{10}a_{11}$. 
\end{definition}

The following lemma states that we can pull back a solution of a string from a solution of its 
string minor.
\begin{lemma}
 \label{prop:minorDF}
Let $L$ be a string over an alphabet $\Sigma$ and $\SSS$ be a set containing distinct position substrings (non-overlapping strings).  
Let $L'$ be a string minor of  $L$ obtained by deleting position substrings in $\SSS$.  Then, there is a polynomial time algorithm, given $L,L',\SSS$ and a solution $\FF'$ of
$(L', \vert \Sigma\vert)$, 
computes a solution $\FF$ of $(L,\vert \Sigma \vert)$ of cardinality $\vert \FF'\vert$. 
\end{lemma}
\begin{proof}
The proof follows from the fact that for each string $w$ in $\FF'$ we can associate indices $i$ and $j$ in $L$ such that $L[i,j]$ is an $x$-factor if and only if $w$ is an $x$-factor in $L'$. Clearly, the algorithms runs in polynomial time. 
\end{proof}

\begin{theorem}
{\disfac} parameterized by $|\Sigma|$ admits a PSAKS. 
\end{theorem}
\begin{proof}
We need to show that for any $\epsilon>0$, there is a polynomial sized $(1-\epsilon)$-approximate kernel 
for  {\disfac}. Towards that given an instance of {\disfac}, we will construct a labelled interval graph $G$
and use \shortuli{$\epsilon$} of $G$ 
to reduce the length of the input string. 
Let $(L,|\Sigma|)$ be an input instance of {\disfac} and  $k=\vert \Sigma \vert$. Now we construct an instance $(G, |\Sigma|,\Gamma)$ of \shortULI{$\epsilon$}. We define  the graph and the labelling function 
 $\Gamma : V(G)\rightarrow \Sigma$  as follows. 
Let $L=a_1a_2\ldots a_n$ where $a_i\in \Sigma$.   For any $i\neq j$ such that $a_i=a_j$ and 
$a_r\neq a_i$ for all $i<r<j$, we construct an interval $I_{i,j}=[i,j]$ on real line and label it with $a_i$. Observe that since $a_i=a_j$, we have that it is an $a_i$-factor.  
The set of intervals constructed form the interval representation of $G$. 
Each interval in $G$ corresponds to a factor in $L$.  
By construction, 
we have that any point belongs to at most two intervals of the same label. 
This implies that the cardinality of largest clique in $G$, and hence $\chi(G)$, 
is upper bounded by $2\vert \Sigma \vert$ (because interval graphs are perfect graphs). 
%
Now we apply Lemma~\ref{lem:uni_lab_is} on input 
$(G,\vert \Sigma\vert, \Gamma)$ and $\epsilon$. Let $X\subseteq V(G)$
be the output of the algorithm. By Lemma~\ref{lem:uni_lab_is}, we have that $\vert X \vert = \vert \Sigma\vert^{\OO(\frac{1}{\epsilon}\log \frac{1}{\epsilon})}$. Let $P=\{q : q \mbox{ is an endpoint of an interval in } $X$\}$ and ${\cal S}=\{L[j,j]~:~j\in [n]\setminus P\}$. The reduction algorithm will output $(L'=L/\SSS,\vert \Sigma\vert)$ as the 
reduced instance. Since $\vert P\vert=\vert \Sigma \vert^{\OO(\frac{1}{\epsilon}\log \frac{1}{\epsilon})}$, we have that the 
length of $L/\SSS$ is at most $\vert \Sigma \vert^{\OO(\frac{1}{\epsilon}\log \frac{1}{\epsilon})}$. 

The solution lifting algorithm is same as the one mentioned in Lemma~\ref{prop:minorDF}. 
Let ${\cal F}'$ be a set of disjoint factors of the reduced instance $(L',\vert \Sigma\vert)$ and let ${\cal F}$ be the output of 
solution lifting algorithm. By Lemma~\ref{prop:minorDF}, we have that  $\vert{\cal F}\vert=\vert{\cal F}'\vert$. 
To prove the correctness 
we need to prove the approximation guarantee of 
${\cal F}$. 
Towards that we first show that $OPT(L/\SSS,\vert \Sigma \vert) \geq (1-\epsilon)OPT(L,\vert \Sigma \vert)$. 
Let ${\cal P}$ be a set of maximum sized disjoint factors 
in $L$. Without loss of generality we can assume that for each factor $L[i,j]$ in ${\cal P}$, $L[i']\neq L[i]$ for all $i<i'<j$. 
This implies that each factor in ${\cal P}$ corresponds to an interval in $G$. Moreover, these set of intervals $U$ 
(intervals corresponding to ${\cal P}$) form 
an independent set in $G$ with distinct labels. By Lemma~\ref{lem:uni_lab_is}, there is an independent set 
$Y$ in $G[X]$ such that $\Gamma(Y)\subseteq \Gamma(U)$ and $\vert \Gamma(Y)\vert \geq (1-\epsilon) \vert \Gamma(U)\vert$. Each interval in $Y$ corresponds to a factor in $L/\SSS$ and its label corresponds to the starting symbol of the factor. This implies that  $L/\SSS$ has a set of disjoint factors of cardinality at least 
$(1-\epsilon)\vert {\cal P} \vert=(1-\epsilon)OPT(L,\vert \Sigma \vert)$.  Hence, we have 
$$\frac{\vert{\cal F}\vert}{OPT(L,\vert \Sigma \vert)}\geq (1-\epsilon) \frac{\vert{\cal F}'\vert}{OPT(L/\SSS,\vert \Sigma \vert)}.$$
This concludes the proof. 
\end{proof}

%% file: cycle_packing.tex
\subsection{Disjoint Cycle Packing}
In this subsection we  design a PSAKS for the {\sc Disjoint Cycle Packing ($CP$)} problem. 
The parameterized optimization problem {\sc Disjoint Cycle Packing ($CP$)} is formally defined as,
\[
    CP(G,k,P)=\left\{ 
\begin{array}{rl}
    -\infty & \text{if $P$ is not a set of vertex disjoint cycles in $G$} \\
    \min\left\{|P|,k+1\right\} & \text{otherwise}
\end{array}\right.
\]

%
%
%

We start by defining feedback vertex sets of a graph. Given a graph $G$ and a vertex subset $F\subseteq V(G)$, $F$ is called a {\em feedback veretx set} of $G$ if $G-F$ is a forest. 
We will make use of the following well-known Erd\H{o}s-P\'{o}sa Theorem relating feedback vertex set and the number of vertex disjoint  cycles in a graph. 
\begin{lemma}[\cite{ErdosPosa1965}]
\label{lemma:erdosposa}
There exists a constant $c$ such that for each positive integer $k$, every (multi) graph either contains $k$ vertex disjoint cycles or it has a feedback vertex set of size at most $ck\log k$. Moreover, there is a polynomial time algorithm that takes a graph $G$ and an integer $k$ as input, and outputs either $k$ vertex disjoint cycles or a feedback vertex set of size at most $ck\log k$. 
\end{lemma}

 The following lemma allows us to reduce the size of the input 
graph $G$ if it has a small feedback vertex set. 

\begin{lemma}
 \label{lemma:leaf}
 Let $(G,k)$ be an instance of \CP{} and $F$ be a feedback vertex set of $G$. Suppose there are strictly  
 more than $\vert F\vert^2(2\vert F\vert+1)$ vertices in $G-F$ whose degree in $G-F$ is at most $1$. 
Then there is a polynomial time 
algorithm ${\cal A}$ that, given an instance $(G,k)$ and a feedback vertex set satisfying the above properties,
  returns a graph $G'$ (which is a minor of $G$) such that $OPT(G,k)=OPT(G',k)$, $\vert V(G')\vert =\vert V(G)\vert -1$ and $F \subseteq V(G')$ is still a feedback vertex set of $G'$.  
Further, given a cycle packing ${\cal S'}$ in $G'$, there is a polynomial time algorithm ${\cal B}$ which 
outputs a cycle packing ${\cal S}$ in $G$ such that $\vert {\cal S}\vert=\vert {\cal S'}\vert$. 
\end{lemma}

\begin{proof}
The algorithm ${\cal A}$ works as follows. 
 Let $|F|=\ell$ and  for $(u,v)\in F\times F$, let $L(u,v)$ be the set of vertices of degree at most $1$ in $G-F$ such that each $x\in L(u,v)$ is adjacent to both $u$ and $v$ (if $u=v$, then $L(u,u)$
is the set of vertices which have degree at most $1$ in $G-F$ and at least two edges to $u$). Suppose that the number of vertices of degree at most $1$ in $G-F$ is strictly more than
$\ell^2(2\ell+1)$. For each pair $(u,v)\in F\times F$, if $L(u,v)>2\ell+1$ then we mark an arbitrary  set of 
$2\ell+1$ vertices from $L(u,v)$, else we mark all the vertices in $L(u,v)$. Since there are at most $\ell^2(2\ell+1)$ marked vertices,
there exists an unmarked vertex $w$ in $G-F$ such that 
 $d_{G-F}(w)\leq 1$. 
 If $d_{G-F}(w)=0$, then algorithm ${\cal A}$
returns $(G-w,k)$.
Suppose $d_{G-F}(w)=1$. Let  $e$ be the unique edge in $G-F$ which is incident to $w$. Algorithm ${\cal A}$  
returns $(G/e,k)$. Clearly $F\subseteq V(G')$ and $F$ is a feedback vertex set of $G'$. 

 Let $(G',k)$ be the instance returned by algorithm ${\cal A}$.
Since $G'$ is a minor of $G$, $OPT(G,k)\geq OPT(G',k)$.   (A graph $H$ is called a minor of an undirected  graph $G^\star$, if we  can obtain $H$ from $G^\star$ by a sequence of edge deletions, vertex deletions and edge contractions.) 
%
Now we show that $OPT(G,k)\leq OPT(G',k)$. 
Let $G'=G/e$, $e=(w,z)$, $d_{G-F}(w)=1$ and $w$ is an unmarked vertex.   
 Let $\CC$ be a maximum set of vertex disjoint cycles in $G$. 
 Observe that if $\CC$ does not contain a pair of cycles each intersecting a different endpoint of $e$, then contracting $e$ will keep the resulting cycles vertex disjoint in $G/e$. Therefore, we may assume that $\CC$ contains 2 cycles $C_w$ and $C_z$ where $C_w$ contains $w$ and $C_z$ contains $z$. Now, the neighbor(s) of $w$ in $C_w$ must lie in $F$. Let these neighbors be $x$ and $y$ (again, $x$ and $y$ are not necessarily distinct). Since $w\in L(x,y)$ and it is unmarked, there are $2\ell+1$ vertices in $L(x,y)$ which are already marked by the marking procedure. 
Further, since for each vertex $u\in V(\CC)$, with $d_{G-F}(u)\leq 1$, 
at least one neighbour of $u$ in the cycle packing $\CC$ is from $F$  and each  
vertex  $v\in V(\CC)\cap F$ can be adjacent to at most $2$ vertices from 
$L(x,y)$, we have that at most $2\ell$ vertices from $L(x,y)$ are in $V(\CC)$. This implies that 
at least one vertex (call it $w'$), marked for $L(x,y)$ is not in   $V(\CC)$. 
Therefore we can route the cycle $C_w$ through $w'$ instead of $w$, which gives us a set of $\vert \CC\vert$ vertex disjoint cycles in $G/e$. 
Suppose $G'=G-w$ and $d_{G-F}(w)=0$. Then by similar arguments we can show that 
$OPT(G,k)=OPT(G-w,k)$. 

Algorithm ${\cal B}$ takes a solution ${\cal S'}$ of the instance $(G',k)$ and outputs a solution ${\cal S}$ of $(G,k)$ 
as follows. If $G'$ is a subgraph of $G$ (i.e, $G'$ is obtained by deleting a vertex), then ${\cal S}={\cal S'}$. Otherwise, 
let $G'=G/e$, $e=(u,v)$ and let $w$ be the vertex in $G'$ created by contracting $(u,v)$. If $w\notin V({\cal S}')$, then 
${\cal S}={\cal S}'$. Otherwise let $C=wv_1\ldots v_{\ell}$ be the cycle in ${\cal S}'$ containing $w$. 
We know that $v_1,v_\ell \in N_G(\{u,v\})$. 
If $v_1,v_{\ell}\in N_G(u)$, then $C'=uv_1\ldots v_{\ell}$ is a cycle in $G$ which is vertex disjoint 
from ${\cal S'}\setminus \{C\}$. If $v_1,v_{\ell}\in N_G(v)$,  then $C'=vv_1\ldots v_{\ell}$ is a cycle in $G$ which is vertex disjoint 
from ${\cal S'}\setminus \{C\}$. In either case ${\cal S}=({\cal S'}\setminus \{C\})\cup \{C'\}$. 
If $v_1\in N_G(u)$ and $v_{\ell}\in N_G(v)$, then $C''=uv_1\ldots v_\ell vu$ is 
a cycle in $G$ which is vertex disjoint 
from ${\cal S'}\setminus \{C\}$. In this case ${\cal S}=({\cal S'}\setminus \{C\})\cup \{C''\}$.
This completes the proof of the lemma.
\end{proof}

Lemma~\ref{lemma:leaf} leads to the following reduction rule which is \onesafe{} (follows from Lemma~\ref{lemma:leaf}). 

\begin{redrule}
\label{rule:leaf}
 Let $(G,k)$ be an instance of \CP{} and let $F$ be a feedback vertex set of $G$ 
such that the forest $G-F$ contains strictly  
 more than $\vert F\vert^2(2\vert F\vert+1)$ vertices of degree at most $1$.  
Then run the algorithm ${\cal A}$ mentioned in Lemma~\ref{lemma:leaf} 
on $(G,k)$ and $F$, and return $(G',k)$, where $G'$, a minor of $G$, is the output of the algorithm ${\cal A}$.   
\end{redrule}

The following observation follows from Lemma~\ref{lemma:leaf}. 
\begin{observation}
\label{obs:leaf_opt}
Let $(G,k)$ be an instance of \CP{} and $(G',k)$ be the instance obtained after applying  
Reduction Rule~\ref{rule:leaf}. Then $OPT(G,k)=OPT(G',k)$. 
\end{observation}

The Reduction Rule~\ref{rule:leaf}, may create multiple edges in the reduced instance.  
To bound the number of multi-edges between a pair of vertices, we use the following 
simple reduction rule. 
\begin{redrule}
 \label{rule:multiedge}
 Let $(G,k)$ be an instance of \CP{} and there exist two vertices $u,v\in G$ such that 
 there are at least $3$ edges between $u$ and $v$. Then delete all but two edges between $u$ and 
 $v$.  
\end{redrule} 

Since any set of vertex disjoint cycles in $G$ can use at most two edges between $u$ and $v$, 
it is safe to delete remaining edges between them and hence Reduction Rule~\ref{rule:multiedge} 
is \onesafe.  Hence, in the rest of the section we always assume that 
the number of edges between any pair of vertices is at most $2$. 
The following lemma allows us to find a subset $F'$, of a feedback vertex set $F$, of cardinality at most $OPT(G,k)$ such that the large portion of the graph $G-F$ is connected to $F'$ and not to $F\setminus F'$.

\input{figure_CP.tex}

\begin{lemma}
\label{lemma:mark_in_tree}
Let $(G,k)$ be an instance of \CP{} and let $F$ be a feedback vertex set of $G$. Then there is 
a polynomial time algorithm ${\cal B}$ that given $(G,k)$ and $F$, either outputs $k$ 
vertex disjoint cycles in $G$ or two sets $F'\subseteq F$ and $S\subseteq V(G-F)$ such that 
$(i)$ $\vert F'\vert, \vert S\vert \leq OPT(G,k)$ and $(ii)$ for any $w\in F \setminus F'$ and any connected component $C$ of $G-(F\cup S)$, 
$\vert N(w)\cap V(C)\vert \leq 1$. 
\end{lemma}
\begin{proof}
We know that $G-F$ is a forest. We consider each tree in $G-F$ as a rooted tree, where the root is chosen arbitrarily. Now we create a 
{\em dummy root $r$} and connect to all the roots in $G-F$. The resulting graph $T$ with vertex set $(V(G)\cup \{r\})\setminus F$ is 
a tree rooted at $r$. The level of a vertex $v\in V(T)$ is the distance between $r$ and $v$, denoted by $d_T(r,v)$. Let $T'$ be a rooted tree, then for a vertex $v\in V(T')$  
we use $T'_v$ to denote the subtree of $T'$ rooted at $v$.  

Now we are ready to give a procedure to find the desired sets 
$F'$ and $S$. Initially we set $T':=T$, $F':=\emptyset$ and $S:=\emptyset$. Let $u\in V(T')$ such that 
$d_{T'}(r,u)$ is maximized and there is a vertex $w\in F\setminus F'$ with the property that $G[V(T'_u)\cup \{w\}]$ has a cycle.  
Then, we set $T':=T'-T'_u$, $F':=F'\cup\{w\}$ and $S:=S\cup \{u\}$. We continue this procedure until $\vert F'\vert=\vert S \vert=k$ 
or the above step is not applicable. Let $F'=\{w_1,\ldots,w_{k'}\}$.
Notice that by the above process there are vertex disjoint subtrees 
$T_1,\ldots,T_{k'}$ of $T-r$  such that for each $i\in [k']$, $G[V(T_i)\cup \{w_i\}]$ has a cycle. 
%
Thus when $k'=k$, our algorithm 
${\cal B}$ will output one cycle from each $G[V(T_i)\cup \{w_i\}]$, $i\in [k]$ as the output. Otherwise, since in each step the 
algorithm picks a vertex with highest level, each connected component $C$ of $T-S$ and $w\in F \setminus F'$,  
$\vert N(w)\cap V(C)\vert \leq 1$. 
Algorithm ${\cal B}$ will output 
$F'$ and $S$ are the required sets. Notice that, in this case $\vert F'\vert=\vert S \vert =k'<k$. We have seen that there are $\vert F'\vert$ vertex disjoint cycles in $G$. This implies that   $\vert F'\vert=\vert S \vert \leq OPT(G,k)$.  An illustration is given in Figure~\ref{fig:treemarking}.  Figure~\ref{fig:atreemarking} depicts a graph $G$ with a feedback vertex set $F$ and 
the sets $F'$ and $S$ chosen by the algorithm. In Figure~\ref{fig:btreemarking}, the graph $G-(F'\cup S)$ is drawn separately 
to see the properties mentioned in the lemma.  
 
\end{proof}

Using Lemma~\ref{lemma:mark_in_tree}, we will prove the following decomposition lemma and after this the structure of the reduced graph becomes ``nice'' and our algorithm boils down to applications of \shortULI{$\epsilon$} on multiple auxiliary interval graphs.  

\begin{lemma}
\label{lem:CP_path_creation}
Let $(G,k)$ be an instance of \CP. Then there is a polynomial time algorithm ${\cal A}$ which either 
outputs $k$ vertex disjoint cycles or 
a minor $G'$ of $G$, and $Z,R\subseteq V(G')$ with the following properties.
\begin{itemize}
\setlength{\itemsep}{-2pt}
\item[$(i)$] $OPT(G,k)=OPT(G',k)$, 
\item[$(ii)$] $\vert Z \vert \leq OPT(G,k)$, $\vert R\vert = \OO(k^4\log^4 k)$, 
\item[$(iii)$] $G'-(Z\cup R)$ is a  collection $\cP$ of  $\OO(k^4\log^4 k)$ non trivial paths, and 
\item[$(iv)$] for each path $P=u_1\cdots u_r$ in $\cP$, no internal vertex is adjacent to a vertex in $R$, \;
$d_{G'[R\cup \{u_1\}]}(u_1)  \leq 1$, and $d_{G'[R\cup \{u_r\}]}(u_r) \leq 1$.
\end{itemize}
Furthermore, given a cycle packing ${\cal S'}$ in $G'$, there is a polynomial time algorithm ${\cal D}$ which 
outputs a cycle packing ${\cal S}$ in $G$ such that $\vert {\cal S}\vert=\vert {\cal S'}\vert$. 
\end{lemma}

\begin{proof}
We first give a description of the polynomial time algorithm ${\cal A}$ mentioned in the statement of the lemma. 
It starts by running  the algorithm mentioned in Lemma~\ref{lemma:erdosposa} on input $(G,k)$ and if it returns $k$ vertex disjoint cycles,  then ${\cal A}$ returns 
$k$ vertex disjoint cycles in $G$ 
and stops. Otherwise, let $F$ be a 
feedback vertex set of $G$. 
Now  ${\cal A}$ applies Reduction Rule~\ref{rule:leaf} repeatedly using the feedback vertex set $F$ until 
Reduction Rule~\ref{rule:leaf} is no longer applicable. Let $(G',k)$ be the reduced instance after the exhaustive application of 
Reduction Rule~\ref{rule:leaf}. By Lemma~\ref{lemma:leaf}, we have that $F\subseteq V(G')$ and $G'-F$ is a forest. 
Now,  ${\cal A}$ runs the algorithm ${\cal B}$ mentioned in Lemma~\ref{lemma:mark_in_tree} on input $(G',k)$ and $F$. If ${\cal B}$ 
returns $k$ vertex disjoint cycles in $G'$, then ${\cal A}$ also returns $k$ vertex disjoint cycles in $G$. The last assertion follows from the fact that  Reduction Rule~\ref{rule:leaf} (applied to get $G'$) is \onesafe. 
Otherwise, let $F'\subseteq F$ and $S\subseteq V(G'-F)$ be the output of ${\cal B}$. Next we define a 
few sets that will be used by $\cal A$ to construct its output.
\begin{enumerate}
\setlength{\itemsep}{-2pt}
\item Let $Q$ be the set 
of vertices of $G'-F$ whose  degree in $G'-F$ is at least $3$. 
\item Let $O=\bigcup_{w\in F\setminus F'} N(w)\cap V(G'-F)$; and 
\item let $W$ be the vertices of degree $0$ in $G'-(F\cup Q\cup O \cup S)$. 
\end{enumerate}
Algorithm ${\cal A}$ returns $G'$, $Z=F'$ and 
$R=Q\cup O \cup S\cup W \cup (F\setminus F')$ as output. In the example given in Figure~\ref{fig:treemarking}, $Z=\{z_1,z_2,z_3\}, F\setminus F'=\{f_1,f_2\}, S=\{s_1,s_2,s_3\}, Q=\{q_1,s_2,s_3\}, O=\{o_1,o_2\}$ and $W=\{w_1\}$.

Now we prove the correctness of the algorithm. If ${\cal A}$ outputs $k$ vertex disjoint  cycles in $G$, then we are done. Otherwise, let $Z=F'$ and $R=Q\cup O \cup S\cup W \cup (F\setminus F')$ be the output of ${\cal A}$. 
Now we prove $G', Z$ and $R$  indeed satisfy the properties mentioned in the statement of lemma.
Since $G'$ is obtained after repeated applications of Reduction Rule~\ref{rule:leaf}, by Observation~\ref{obs:leaf_opt}, we get 
that $OPT(G,k)=OPT(G',k)$ and hence proving property $(i)$.

By Lemma~\ref{lemma:mark_in_tree}, 
we have that $\vert F'\vert= \vert S \vert \leq OPT(G,k)$. Hence the size of $Z(=F')$ is as desired. 
Next we bound the size of $R$.  
By  Lemma~\ref{lemma:erdosposa}, we have that $\vert F\vert\leq ck\log k$, where $c$ is a fixed constant.   
By Lemma~\ref{lemma:leaf}, we have that $F\subseteq V(G')$, $G'-F$ is a forest, and 
the number of vertices of  degree at most $1$ in $G'-F$ is upper bounded by $\vert F\vert^2(2\vert F\vert+1)=\OO(k^3\log^3 k)$. 
Since the number of vertices of degree at least $3$ in a forest is at most the number of leaves in the forest, we can conclude 
that cardinality of $Q$, the set of vertices of degree at least $3$ in $G'-F$ is upper bounded by $\OO(k^3\log^3 k)$. It is well-known that the number of maximal degree $2$ paths in a forest is  upper bounded by the sum of the number of leaves and the vertices of degree at least $3$ 
(for example see~\cite{RamanSS06} for a proof). This immediately implies the following claim. 

%
\begin{claim}
 \label{claim:no_of_paths}
 $G'-(F\cup Q)$ is a collection of $\OO(k^3\log^3 k)$ paths. 
\end{claim}
The following claim proves properties $(ii)$ and $(iii)$ stated in the lemma. 
\begin{claim}
\label{claim:CPboundR}
$\vert R\vert = \OO(k^4\log^4 k)$ and the number of paths in $\cP$ is at most  $\OO(k^4\log^4 k)$. 
\end{claim}
\begin{proof}
Observe that $G'-(F\cup Q)$ is a 
collection of $\OO(k^3\log^3 k)$ paths and thus it has at most  $\OO(k^3\log^3 k)$ connected components. This implies that, $G'-(F\cup Q \cup S)$ has at most  $\OO(k^3\log^3 k)$ connected components and in particular $G'-(F \cup S)$ has at most  $\OO(k^3\log^3 k)$ connected components. 
Let $\gamma_s$ denote the number of connected components of $G'-(F \cup S)$. 
By Lemma~\ref{lemma:mark_in_tree}, we have that 
for any $w\in F \setminus F'$ and any connected component $C$ of $G'-(F\cup S)$, 
$\vert N_{G'}(w)\cap V(C) \vert \leq 1$.  Thus, for every vertex  $w\in F \setminus F'$ we have that 
$|N(w)\cap V(G'-F)|\leq |S|+ \gamma_s=\OO(k^3\log^3 k)$. 
This
implies that the cardinality of $O$, the set $\bigcup_{w\in F\setminus F'} N(w)\cap V(G'-F)$,  
is upper bounded by $\OO( \vert F \setminus F'\vert \cdot k^3\log^3 k)=\OO( k^4\log^4 k)$.  
By Lemma~\ref{lemma:mark_in_tree}, we have that $\vert S\vert\leq OPT(G,k)\leq k+1$. 
Since $\vert O \cup S  \vert =\OO( k^4\log^4 k)$ and by Claim~\ref{claim:no_of_paths}, we can conclude 
that the number of paths in $G'-(F\cup Q \cup O\cup S)$ is at most $\OO( k^4\log^4 k)$. 
Notice that $W$ is the family of paths on single vertices in the collection of paths of 
$G'-(F\cup Q \cup O\cup S)$. 
Since the number of maximal paths in $G'-(F\cup Q \cup O\cup S)$ is at most $\OO( k^4\log^4 k)$, 
we have that $\vert W \vert= \OO( k^4\log^4 k)$ and 
the number of maximal paths in $G'-(F\cup Q \cup O\cup S \cup W)=G'-(Z\cup R)$ (i.e, the 
number of paths in $\cP$) is at most $\OO( k^4\log^4 k)$. 

Since $\vert S\vert\leq OPT(G,k)\leq k+1$, $\vert F\vert\leq ck\log k$, $\vert Q\vert=\OO(k^3\log^3 k)$, $\vert O \cup S \vert =\OO( k^4\log^4 k)$ and $\vert W \vert= \OO( k^4\log^4 k)$, we can conclude that the cardinality of 
$R=Q\cup O \cup S \cup W \cup (F\setminus F')$, is upper bounded by $\OO(k^4\log^4 k)$. This concludes the proof. 
\end{proof}
Finally, we will show the last property stated in the lemma. 
Since $G'-F$ is a forest and $Q$ is the set of vertices of degree at least $3$ in the forest $G'-F$, 
we have that  any internal vertex of any path in $G'-(Q\cup F)$ is not adjacent to $Q$. 
Also, since any vertex $w$, which is an internal vertex of a path in $G'-(Q\cup F)$ and 
adjacent to a vertex in $F\setminus F'$, belongs to $O$, we can conclude that no internal vertex of any path in $G'-(Q\cup O \cup F)$ is adjacent to 
$Q\cup O \cup (F\setminus F')$. This implies that  
no internal vertex of any path in $G'-(Q\cup O\cup S \cup W \cup F)= G'-(Z\cup R)$ is adjacent to 
$Q\cup O \cup S \cup W \cup (F\setminus F') = R$.  
Now we claim that an endpoint $u$ of a path $P$ in $\cP$ has at most 
 one edge between $u$ and $R$. 
Let $u$ be an endpoint of $P$. 
Since  $O=\bigcup_{w\in F\setminus F'} N(w)\cap V(G'-F)$ and $u\notin O$, we can conclude that $u$ is not adjacent to 
any vertex in $F\setminus F'$. 
Since $u\in V(G'-(F\cup Q))$, the degree of 
$u$ in $G'-F$ is at most $2$. Since $P$ is a non trivial path $\vert N(u)\cap (V(G'-F)\setminus V(P))\vert \leq 1$. 
Since $G'-F$ is a forest, $\vert N(u)\cap (V(G'-F)\setminus V(P))\vert \leq 1$, and $u$ is not adjacent to any vertex in $F\setminus F'$, 
we conclude that $d_{G'[R\cup \{u\}]}(u)\leq 1$.  

The solution lifting algorithm, ${\cal D}$,  is basically obtained by solution lifting algorithm 
used in the Reduction Rule~\ref{rule:leaf}. That is,   given a cycle packing ${\cal S'}$ in $G'$,  $\cal D$ 
repeatedly applies the solution lifting algorithm of  Reduction Rule~\ref{rule:leaf} to obtain a cycle packing ${\cal S}$ in $G$ such that $\vert {\cal S}\vert=\vert {\cal S'}\vert$. The correctness of the algorithm ${\cal D}$ follows 
from the fact that Reduction Rule~\ref{rule:leaf} is \onesafe, 
and $G'$ is obtained from $G$ by repeated application of 
Reduction Rule~\ref{rule:leaf}.  
An illustration of a path $P\in \cP$, $Z$ and $R$ can be found in Figure~\ref{figure_CP_paths}. 
This completes the proof of the lemma. 
\end{proof}

\input{figure_CPPath.tex}

Observe that Lemma~\ref{lem:CP_path_creation} decomposes the graph into $k^{\OO(1)}$ simple structures, namely, paths in $\cal P$  combined together with a set of size $k^{\OO(1)}$. Note that the only {\em unbounded} objects in $G'$ are the paths in $\cal P$. The reason we can not reduce the size of $P$ is that a vertex in $Z$ can have unbounded neighbors on it. See Figure~\ref{figure_CP_paths} for an illustration. However,  
Lemma~\ref{lem:CP_path_creation} still  provides us the required decomposition which will be used to cast several instances of 
\shortULI{$\epsilon$}. In particular for every path $P \in \cal P$, we will have one instance of 
\shortULI{$\epsilon$}. We will compute \shortuli{$\epsilon$} for each of these instances and reduce the path size to get the desired kernel. 

  
\begin{theorem}
\label{thm:CPapprx}
For any $\epsilon>0$, there is polynomial sized $(1-\epsilon)$-approximate kernel for \CP. That is, \CP{} admits a PSAKS. 
\end{theorem}
\begin{proof}
Let $(G,k)$ be an input instance of \CP.  
The reduction algorithm ${\cal R}$ works as follows. It first runs the algorithm ${\cal A}$ mentioned in Lemma~\ref{lem:CP_path_creation}. If the algorithm returns $k$ vertex disjoint cycles, then ${\cal R}$ return these cycles. Otherwise, let $G'$, $Z$ and $R$ be the output of ${\cal A}$, satisfying four properties mentioned in Lemma~\ref{lem:CP_path_creation}. Important properties that will be most useful in our context are: 

\begin{itemize}
\setlength{\itemsep}{-2pt}
\item $\vert Z \vert \leq OPT(G,k)$, $\vert R\vert = \OO(k^4\log^4 k)$; and 
\item  $G'-(Z\cup R)$ is a collection $\cP$ of non trivial  
paths such that for any path $P\in \cP$ we have that no internal vertex of $P$ is adjacent to  
any vertex of $R$.
\end{itemize}
\noindent 
Now ${\cal R}$ will solve several instances of  \shortULI{$\epsilon$} to bound the length of each path in 
$\cP$. Towards this we fix a path $$P=v_1v_2\ldots v_{\ell}  \mbox{ in }  
\cP.$$ 
Our objective is to apply Lemma~\ref{lem:uni_lab_is} 
to reduce the length of $P$. 
Our algorithm finds a set of small number of relevant vertices on $P$ 
and reduces $P$ in {\em a single step} even though 
we use Lemma~\ref{lem:uni_lab_is} several times to identify relevant vertices.  Next we give the  construction for applying Lemma~\ref{lem:uni_lab_is} in order to find the relevant vertices. To find relevant vertices of $P$, we create $(\vert Z \vert +1)^2$ labelled interval graphs, one for every 
$(x,y)\in Z\cup \{\clubsuit\} \times Z \cup \{\clubsuit\}$  
with $Z\times Z$ being the set of labels. That is, for  the path $P$ and $(x,y)\in Z\cup \{\clubsuit\} \times Z \cup \{\clubsuit\}$ we create a  labelled interval graph $H_P^{(x,y)}$ as follows. Our labelling function will be denoted by  $\Gamma_P^{(x,y)}$.

\begin{enumerate}
\setlength{\itemsep}{-2pt}
\item The set of labels is $\Sigma=Z \times Z$. 
\item Let $P^{(x,y)}= v_r\ldots v_{r'}$ be the subpath of $P$ such that $v_{r-1}$ is the first vertex in $P$ adjacent to $x$ and $v_{r'+1}$ is the last vertex on $P$ adjacent to $y$.  If $x=\clubsuit$, then $v_r=v_1$ and if $y=\clubsuit$, then $v_{r'}=v_\ell$. Indeed, if $x=\clubsuit$ and $y=\clubsuit$ then $v_r=v_1$ and $v_{r'}=v_\ell$. 
\item We say that a subpath $Q'$ of $P^{(x,y)}$ is a {\em potential $(u_1,u_2)$-subpath}, 
where $(u_1,u_2)\in Z \times Z$, if   
either $u_1Q'u_2$ or $u_2Q'u_1$ is an induced path (induced cycle when $u_1=u_2$) in $G'$. Essentially, the potential subpath is trying to capture the way a cycle can interact with a subpath in $P$ with its neighbors on the cycle being $u_1$ and $u_2$.
\item For each $(u_1,u_2)\in Z \times Z$ 
and a potential $(u_1,u_2)$-subpath $Q'=v_i\ldots v_j$ we create an interval  
$I_{Q'}^{(u_1,u_2)}=[i,j]$ and  label it with $(u_1,u_2)$. That is, $\Gamma_P^{(x,y)}(I_{Q'}^{(u_1,u_2)})=(u_1,u_2)$. 
We would {\em like to emphasize} that when $u_1=u_2$ and $v_i=v_j$, we create an interval $I_{Q'}^{(u_1,u_2)}=[i,j]$ only if 
there are two edges between $u_1$ and $v_i$.  
Also notice that if we have created an interval $I_{Q'}^{(u_1,u_2)}=[i,j]$ with label $(u_1,u_2)$, then we have created an interval $I_{Q'}^{(u_2,u_1)}=[i,j]$ with label $(u_2,u_1)$ as well.  
\end{enumerate}

\begin{figure}

        \centering

\begin{tikzpicture}[ scale=1.3]

\node [] (a) at (0.4,0) {$P$};
\node [] (a) at (0.75,0) {$\bullet$};
\node [] (a) at (1.5,0) {$\bullet$};
\node [] (a) at (2.25,0) {$\bullet$};
\node [] (a) at (3,0) {$\bullet$};
\node [] (a) at (3.75,0) {$\bullet$};
\node [] (a) at (4.5,0) {$\bullet$};
\node [] (a) at (5.25,0) {$\bullet$};
\node [] (a) at (6,0) {$\bullet$};
\node [] (a) at (6.75,0) {$\bullet$};
\node [] (a) at (7.5,0) {$\bullet$};

\node [blue] (a) at (1.5,1.5) {$Z$};
\node [blue] (a1) at (2,1.5) {$\bullet$};
\node [blue] (a2) at (2.75,1.5) {$\bullet$};
\node [blue] (a) at (3.5,1.5) {$\bullet$};
\node [blue] (a) at (4.25,1.5) {$\bullet$};
\node [blue] (a) at (5,1.5) {$\bullet$};

\node [] (a) at (2,1.8) {$u_1$};
\node [] (a) at (2.75,1.8) {$u_2$};
\node [] (a) at (3.5,1.8) {$u_3$};
\node [] (a) at (4.25,1.8) {$u_4$};
\node [] (a) at (5,1.8) {$u_5$};


 \draw[] (0.75,0)-- (1.5,0) --(2.25,0) --
(3,0) --
 (3.75,0)-- 
 (4.5,0)-- 
 (5.25,0)--(6,0) --(6.75,0)--(7.5,0); 

 \draw[black]
(2,1.5) -- (1.5,0)
(5,1.5)-- (6,0)
(2,1.5) -- (2.25,0)
(2,1.5) -- (4.5,0)
(2.75,1.5)-- (3.75,0)
(2.75,1.5)-- (5.25,0)
(4.25,1.5)to[out=-10,in=60](5.25,0)
(4.25,1.5)to[out=-90,in=100](5.25,0)
(3.5,1.5)--(0.75,0)
;

\node [] (a) at (3,-0.4) {$(u_1,u_2)$};
\draw[line width=0.35mm,red] (2.25,-0.53)--(3.75,-0.53); 
\node [] (a) at (4.9,-0.4) {$(u_1,u_2)$};
\draw[line width=0.35mm,red] (2.25,-0.65)--(3.75,-0.65);
\node [] (a) at (3,-0.8) {$(u_2,u_1)$};
\node [] (a) at (4.9,-0.8) {$(u_2,u_1)$};
\draw[line width=0.35mm,red] (4.5,-0.53)-- (5.25,-0.53);
\draw[line width=0.35mm,red] (4.5,-0.65)-- (5.25,-0.65);
\node [] (a) at (5.8,-1.4) {$(u_4,u_4)$};
\node [red] (a) at (5.25,-1.4) {$\bullet$};

\draw[line width=0.35mm,red] (2.25,-1.2)--(4.5,-1.2);
\node [] (a) at (3.5,-1.05) {$(u_1,u_1)$};

\node [] (a) at (4.2,-1.45) {$(u_1,u_2)$};
\draw[line width=0.35mm,red] (3.75,-1.6)--(4.5,-1.6);
\draw[line width=0.35mm,red] (3.75,-1.75)--(4.5,-1.75);
\node [] (a) at (4.2,-1.85) {$(u_1,u_2)$};

\draw[thick, dotted] (0,-2)--(8,-2);
\node [] (a) at (.75,-2.15) {$1$};
\node [] (a) at (1.5,-2.15) {$2$};
\node [] (a) at (2.25,-2.15) {$3$};
\node [] (a) at (3,-2.15) {$4$};
\node [] (a) at (3.75,-2.15) {$5$};
\node [] (a) at (4.5,-2.15) {$6$};
\node [] (a) at (5.25,-2.15) {$7$};
\node [] (a) at (6,-2.15) {$8$};

\end{tikzpicture}

\caption{An example of $H_P^{(u_1,u_5)}$.  The interval representation of $H_P^{(u_1,u_5)}$ along with labels is drawn below 
the path $P$. The real line is represented using a dotted line.}
\label{figure_HP}
\end{figure}

This completes the construction of  $H_P^{(x,y)}$   and the labelling function
 $\Gamma_P^{(x,y)}$.  
See Figure~\ref{figure_HP} for an illustration.  
The fact that $H_P^{(x,y)}$  is an interval graph follows from the fact that in fact to construct  $H_P^{(x,y)}$, we have given an interval representation for it. 
Having, created the interval graph and a labelling function $\cal R$ runs the following steps.

\begin{enumerate}
\setlength{\itemsep}{-2pt}
\item Now using Lemma~\ref{lem:CP_path_creation},  ${\cal R}$ computes 
a set $X^{(x,y)}_P$ such that $X^{(x,y)}_P$ is a \shortuli{$\frac{\epsilon}{2}$} of $H^{(x,y)}_P$. 
Now we define a few sets. 
\begin{eqnarray*}
S^{(x,y)}_P & = & \{v_i : i \mbox{ is an endpoint of an interval in } H^{(x,y)}_P \}\\
K_P& = & \{v_1,v_{\ell}\} \cup \bigcup_{u\in Z} \{ v : v \mbox{ is the first or last vertex on $P$ such that } uv\in E(G')\}\\ 
 S_P& = & \bigcup_{(x,y)\in  Z\cup \{\clubsuit\} \times Z \cup \{\clubsuit\}} S^{(x,y)}_P \\
 D_P & = & V(P)\setminus (S_P\cup K_P).
\end{eqnarray*}


\item Now, ${\cal R}$ will do the following modification to  shorten $P$: delete all the edges between 
$D_P$ and $Z$, and then contract all the remaining edges incident with vertices in $D_P$. In other words, let $\{v_{i_1},\ldots v_{i_{\ell'}}\}=S_P\cup K_P$, where $1=i_1<i_2<\ldots <i_{\ell'}=\ell$. Then delete $D_P$ and add edges $v_{i_j}v_{i_{j+1}}, j\in [\ell'-1]$. 
Let $P'$ be the path obtained from $P$, by the above process. We use the same vertex names in $P'$ as well to represent a
vertex. That is, if a vertex $u$ in $V(P)$ is not deleted to obtain $P'$, we use $u$ to represent the same vertex. 
\item Let $G''$ be the graph obtained after this modification has been done for all paths $P\in \cP$. 
Finally,  ${\cal R}$ returns $(G'',k)$ as the reduced instance. 
\end{enumerate}
\paragraph{Solution Lifting Algorithm.} 
Notice that $G''$ is a minor of $G'$ and hence a minor of $G$.  Given a set $S'$ of vertex disjoint cycles in 
$G''$, the solution lifting algorithm computes a set $S$ of vertex disjoint cycles in $G$ of cardinality 
$\vert S'\vert$ by doing reverse of the minor operations used to obtain $G''$ from $G$. All this can be done in polynomial time because  the solution lifting algorithm knows the 
minor operations done to get $G''$ from $G$.    

Next we need to prove the correctness of the algorithm. Towards that we first bound the size of $G''$.  

\medskip
\noindent 
{\bf Bounding the size of $G''$.} As a first step to bound the size of $G''$, we bound the chromatic number of $H_P^{(x,y)}$, where $P\in \cP$ and 
$(x,y)\in Z\cup \{\clubsuit\} \times Z \cup \{\clubsuit\}$. In fact what we will bound is the size of the maximum clique of $H_P^{(x,y)}$. 
\begin{claim}
\label{claim:CPpw}
For any $P\in \cP$ and $(x,y)\in Z\cup \{\clubsuit\} \times Z \cup \{\clubsuit\}$, $\chi(H_P^{(x,y)})=\OO(k^2)$. 
\end{claim}
\begin{proof}
To prove the claim, it is enough to show that the size of a maximum clique in $H_P^{(x,y)}$ is at most $\OO(k^2)$. 
Let $P^{(x,y)}= v_r\ldots v_{r'}$. We know that in the interval representation of $H_P^{(x,y)}$,  
all the intervals are contained in $[r,r']$. We claim that for any point $p\in [r,r']$ and $(u_1,u_2)\in Z \times Z$,  
the number of intervals labelled $(u_1,u_2)$ and containing the point $p$ is at most $2$. Towards a contradiction assume that there are 
three intervals $I_1=[i_1,j_1],I_2=[i_2,j_2],I_3=[i_3,j_3]$ such that $\Gamma_P^{(x,y)}(I_1)=\Gamma_P^{(x,y)}(I_2)=\Gamma_P^{(x,y)}(I_3)=(u_1,u_2)$ and all the 
intervals $I_1$, $I_2$ and $I_3$ contain the point $p$. Since for each $r\leq i,j\leq r'$ and $(u_1,u_2)\in Z\times Z$ 
we have created at most one interval $[i,j]$ with label $(u_1,u_2)$, all the intervals $I_1,I_2$ and $I_3$ are distinct intervals in the real line. 

We first claim that no interval in $\{I_1,I_2,I_3\}$ is same as $[p,p]$. Suppose $I_3=[p,p]$. 
Since all the interval in $\{I_1,I_2,I_3\}$, are different and $I_3=[p,p]$ we have that 
$I_1\neq [p,p]$, but contains $p$. This implies that either $i_1\neq p$ or $j_1\neq p$. We consider 
the case $i_1\neq p$. The case that $j_1\neq p$  is symmetric. 
Let $Q_1=v_{i_1}v_{i_1+1}\ldots v_{j_1}$. 
We know that $u_1v_pu_2$ is an induced 
path (induced cycle when $u_1=u_2$ and two edges between $u_1$ and $p$). This 
implies that neither $u_1Q_1u_2$ nor $u_2Q_1u_1$ is an induced path, 
because $v_p\in \{v_{i_1+1}\ldots v_{j_1}\}$. We would like to clarify that  
when ${j_1}=p$ and $u_1=u_2$, $u_1Q_1u_1$ is cycle 
and there are two edges between $v_{j_1}$ and $u_1$. 
This implies that $u_1Q_1u_1$ is a not an induced cycle.
 See Figure~\ref{figure_pointinterval} for illustration.  

\begin{figure}
\begin{subfigure}[b]{0.5\textwidth}
        \centering

\begin{tikzpicture}[ scale=1]

\node[blue] at (-0.8,1.5) {$u_1$};
\node [blue] (a) at (-0.5,1.5) {$\bullet$};
\node [blue] (a) at (2.5,1.5) {$\bullet$};
\node[blue] at (2.8,1.5) {$u_2$};

\draw[line width=0.35mm,red] (0,0)--(2,0);
\node [red] (a) at (0,0) {$\bullet$};
\node [red] (a) at (2,0) {$\bullet$};
\node [red] (a) at (1,0.8) {$\bullet$};
\node [red] (a) at (1,0.5) {$I_3$};
\node [red] (a) at (1,-0.3) {$I_1$};
\node [red] (a) at (1,0) {$\bullet$};

\draw[thick, dotted] (-0.5,1.5) -- (0,0);
\draw[thick, dotted] (-0.5,1.5) -- (1,0);
\draw[thick, dotted] (2.5,1.5) -- (2,0);
\draw[thick, dotted] (2.5,1.5) -- (1,0);


\node[blue] at (4.2,1.5) {$u_1$};
\node [blue] (a) at (4.5,1.5) {$\bullet$};

\draw[line width=0.35mm,red] (5,0)--(7,0);
\node [red] (a) at (5,0) {$\bullet$};
\node [red] (a) at (7,0) {$\bullet$};
\node [red] (a) at (7,0.8) {$\bullet$};
\node [red] (a) at (7,0.5) {$I_3$};
\node [red] (a) at (6,-0.3) {$I_1$};
\draw[thick, dotted] (4.5,1.5) -- (5,0);
\draw[thick, dotted] (4.5,1.5) -- (7,0);
\draw[thick, dotted] (4.5,1.5)  to[out=0,in=120] (7,0);
\end{tikzpicture}
\end{subfigure}
\begin{subfigure}[b]{0.5\textwidth}
        \centering

\begin{tikzpicture}[ scale=1]

\node[blue] at (-0.8,1.5) {$u_1$};
\node [blue] (a) at (-0.5,1.5) {$\bullet$};
\node [blue] (a) at (2.5,1.5) {$\bullet$};
\node[blue] at (2.8,1.5) {$u_2$};

\draw[line width=0.35mm,red] (0,0)--(2,0);
\draw[line width=0.35mm,red] (1,0.7)--(1.8,0.7);
\draw[line width=0.35mm,red,dotted] (1.8,0.7)--(3,0.7);
\node [red] (a) at (0,0) {$\bullet$};
\node [red] (a) at (2,0) {$\bullet$};
\node [red] (a) at (1,0.5) {$I_1$};
\node [red] (a) at (1,-0.3) {$I_2$};
\node [red] (a) at (1,0) {$\bullet$};

\draw[thick, dotted] (-0.5,1.5) -- (0,0);
\draw[thick, dotted] (-0.5,1.5) -- (1,0);
\draw[thick, dotted] (2.5,1.5) -- (2,0);

\end{tikzpicture}
\end{subfigure}
\caption{Illustration of 
proof of Claim~\ref{claim:CPpw}. The case when $I_3=[p,p]$ is drawn in the left and middle figures. 
The figure in the middle represents the 
case when $j_1=p$ and $u_1=u_2$.
The case when $I_1$ and $I_2$ intersects at strictly more than one point can be seen in the  right most figure.  
The black dotted curves represent the edges in the graph. }
\label{figure_pointinterval}
\end{figure}

Since $p$ is a common point in $I_1$, $I_2$ and $I_3$ and none of these intervals is equal to $[p,p]$, there are two intervals in 
$\{I_1,I_2,I_3\}$ such that they intersect at strictly more than one point. 
Without loss of generality we assume that the intersection of  $I_1$ and $I_2$
contains at least $2$ points. Also, since $I_1$ and $I_2$ are different intervals 
on the real line, one endpoint of an interval is fully inside another interval (not as the endpoint of the other interval). 
Let $Q_2=v_{i_2}v_{i_2+1}\ldots v_{j_2}$. 
Assume that $i_1\in (i_2,j_2)$. All other cases are symmetric to this case.  
We know that  $u_1v_{i_1}\in E(G')$ or $u_2v_{i_1}\in E(G')$.
This implies that 
neither $u_1Q_2u_2$ nor $u_2Q_2u_1$ is an induced path. This contradicts the fact that we created an interval 
$[i_2,j_2]$ with label $(u_1,u_2)$. See Figure~\ref{figure_pointinterval} for illustration.

We have proved that for any point $p\in [r,r']$, the number of intervals containing   $p$ with the same label is upper bounded by  $2$. This implies that the cardinality of a largest clique in $H_P^{(x,y)}$ is at most twice the number of labels. Thus, the size of the largest cliques is upper bounded by $\OO(|Z|^2)$. 
By  Lemma~\ref{lem:CP_path_creation}, we know that $\vert Z \vert \leq OPT(G,k) \leq k+1$ and thus 
$\OO(|Z|^2)$ is bounded by $\OO(k^2)$. Since the chromatic number of an interval graph is upper bounded by the size of a maximum clique, the proof of the claim follows.  
\end{proof}

By 
Lemma~\ref{lem:CP_path_creation}, we know that 
$\vert \cP \vert= \OO(k^4\log^4 k)$. 
For each $P\in {\cal P}$ and $(x,y)\in Z\cup \{\clubsuit\} \times Z \cup \{\clubsuit\}$, we created a subset 
$S^{(x,y)}_P$ of $V(P)$ of cardinality $(|Z|^2 \cdot \chi(H_P^{(x,y)}) ^{\OO(\frac{2}{\epsilon}\log \frac{2}{\epsilon})}=k^{\OO(\frac{1}{\epsilon}\log \frac{1}{\epsilon})}$. Hence, the cardinality of $S_P$ is also upper bounded by $k^{\OO(\frac{1}{\epsilon}\log \frac{1}{\epsilon})}$. 
The cardinality of  $K_P$ is 
at most $2|Z|+2=\OO(k)$. 
This implies that the reduced path $P'$ has at most $k^{\OO(\frac{1}{\epsilon}\log \frac{1}{\epsilon})}$ vertices. 
Also, we know that $\vert \cP\vert =\OO(k^4\log^4 k)$, hence the total number of vertices across all the paths of $\cP$ after the reduction is upper bounded by $k^{\OO(\frac{1}{\epsilon}\log \frac{1}{\epsilon})}$. 
This together with the fact that $\vert Z \vert \leq OPT(G,k)$ and  $\vert R\vert = \OO(k^4\log^4 k)$ imply that  $\vert V(G'')\vert$ is upper bounded by $k^{\OO(\frac{1}{\epsilon}\log \frac{1}{\epsilon})}$. This completes the proof of upper bound on the size of $G''$.

\medskip
\noindent 
{\bf Correctness of lossy reduction.} Finally, we show that indeed $(G'',k)$ is a  
$(1-\epsilon)$-approximate kernel for \CP. Towards this we show the following claim. 

\begin{claim}
\label{claim:optCP}
$OPT(G'',k)\geq (1-\epsilon)OPT(G',k)$. 
\end{claim}
\begin{proof}
Let $\CC$ be an optimum solution to $(G',k)$. Without loss of generality we can assume that 
each cycle in $\CC$ is a {\em chordless cycle}.  
Let $\cQ$ be the non-empty subpaths of cycles in $\CC$ induced in the graph $G'-(Z\cup R)$. 
That is, $\cQ$ is the collection of supaths in the intersection of $\CC$ and $\cP$.
For any $Q\in \cQ$, 
there exists two vertices $u,v\in R \cup Z$ such that $uQv$ is a subpath in $\CC$.
Because of property $(iv)$ of Lemma~\ref{lem:CP_path_creation}, for any $Q\in \cQ$ with $\vert V(Q)\vert =1$, at least one of the endpoint of $Q$ is connected to a vertex from $Z$ in the cycle packing $\CC$. 
We say a path $Q'$ is a {\em substitute} for $Q\in \cQ$ if $uQ'v$ is a subpath in $G''$ where 
$u,v\in R\cup Z$ and $uQv$ is a subpath in $\CC$. 
In what follows,   for at least $(1-\epsilon)\vert \cQ\vert$ paths in $\cQ$, we identify substitutes in 
the reduced graph $G''$ which are pairwise vertex disjoint.
%

We partition the paths in $\cQ$ into $\cQ_1$ and $\cQ_2$. 
Notice that $Q_i\in \cQ$ is a subpath of a cycle $C\in \CC$ and the neighbors (could be the same) 
of  both the endpoints  of $Q_i$ on $C$ are in 
$R\cup Z$. If the neighbors of both endpoints of $Q_i$  on $C$ are in $Z$, then we include $Q_i$ in $\cQ_1$.  Otherwise $Q_i$ is in $\cQ_2$. 
See Figure~\ref{figure_cyclepathintersection} for an illustration. 
For each $Q\in \cQ_2$, we give a substitute path 
as follows. 
We know that there is a path $P\in\cP$ such that  either $P=QQ'$ or $P=Q'Q$ for some $Q'$ where $V(Q')$ can be 
an empty set too.
If $Q=P$, then we replace $Q$ with $P'$ 
(Note that $P'$ is the path obtained from $P$ in the reduction process). Also, notice that end vertices of $P$ and $P'$ are same 
(because endvertices of $P$ belong to $K_P$) 
and hence $P'$ is a substitute for 
$Q$. Suppose $P=QQ'$ where $V(Q')\neq \emptyset$. Let $C_Q$ be the cycle in $\CC$ such that $Q$ is a subpath of $C_Q$. 
Let $Q=v_1\ldots v_{d}$. Let $z$ be the neighbour of $v_d$ in $C_Q$ which is from $Z$ (recall that 
 no internal vertex of $P$ is adjacent to  any vertex of $R$).  
Since 
$C_Q$ is a chordless cycle, none of $v_1,\ldots, v_{d-1}$ is adjacent to $z$. This implies that  $v_1,v_{d}\in K_P$ and hence 
$P'$ contains a subpath $P'_{Q}$ from $v_1$ to $v_{d-1}$ with internal vertices from $\{v_2,\ldots,v_{d-1}\}$. In this case 
$P'_{Q}$ is a substitute for $Q$. In a similar way, we can construct a substitute for $Q$ when  $P=Q'Q$ where $V(Q')\neq \emptyset$. 
Let $\cQ'_2$ be the set of substitute paths constructed for paths in $\cQ_2$. Notice that $\cQ'_2$ 
is a collection of vertex disjoint paths in $G''-(Z\cup R)$ and it has one substitute path for each $Q\in \cQ_2$. 
See Figure~\ref{figure_cyclepathintersection} for an illustration. 

\begin{figure}

\begin{subfigure}[b]{0.5\textwidth}
        \centering
\begin{tikzpicture}[ scale=1]

\node [] (a) at (0.4,0) {$P$};
\node [] (a) at (0.75,0) {$\bullet$};
\node [] (a) at (1.5,0) {$\bullet$};
\node [] (a) at (2.25,0) {$\bullet$};
\node [] (a) at (3,0) {$\circ$};
\node [] (a) at (3.75,0) {$\bullet$};
\node [] (a) at (4.5,0) {$\circ$};
\node [] (a) at (5.25,0) {$\bullet$};
\node [] (a) at (6,0) {$\bullet$};
\node [] (a) at (6.75,0) {$\circ$};
\node [] (a) at (7.5,0) {$\bullet$};

\node [blue] (a) at (1.5,1.5) {$Z$};
\node [blue] (a1) at (2,1.5) {$\bullet$};
\node [blue] (a2) at (2.75,1.5) {$\bullet$};
\node [blue] (a) at (3.5,1.5) {$\bullet$};
\node [blue] (a) at (4.25,1.5) {$\bullet$};
\node [blue] (a) at (5,1.5) {$\bullet$};

\node [] (a) at (2,1.8) {$u_1$};
\node [] (a) at (2.75,1.8) {$u_2$};
\node [] (a) at (3.5,1.8) {$u_3$};
\node [] (a) at (4.25,1.8) {$u_4$};
\node [] (a) at (5,1.8) {$u_5$};

 \draw[] (0.75,0)-- (1.5,0) --(2.25,0) --
(3,0) --
 (3.75,0)-- 
 (4.5,0)-- 
 (5.25,0)--(6,0) --(6.75,0)--(7.5,0); 

 \draw[black]
(2,1.5) -- (1.5,0)
(5,1.5)-- (6,0)
(2,1.5) -- (2.25,0)
(2,1.5) -- (4.5,0)
(2.75,1.5)-- (3.75,0)
(2.75,1.5)-- (5.25,0)
(4.25,1.5)to[out=-10,in=60](5.25,0)
(4.25,1.5)to[out=-90,in=100](5.25,0)
(3.5,1.5)--(0.75,0)
;

 \draw[red] (0.7,-0.2)-- (0.8,-0.2) 
(2.25,-0.2) --(3.75,-0.2) 
 (4.5,-0.2)--(5.25,-0.2)
(6,-0.2) --(6.75,-0.2)--(7.5,-0.2); 

\node [] (a) at (0.78,-0.45) {$Q_1$};
\node [] (a) at (3,-0.45) {$Q_2$};
\node [] (a) at (5,-0.45) {$Q_3$};
\node [] (a) at (6.5,-0.45) {$Q_4$};
\end{tikzpicture}
\end{subfigure}
\begin{subfigure}[b]{0.5\textwidth}
        \centering
\begin{tikzpicture}[ scale=1]

\node [] (a) at (0.4,0) {$P'$};
\node [] (a) at (0.75,0) {$\bullet$};
\node [] (a) at (1.5,0) {$\bullet$};
\node [] (a) at (2.25,0) {$\bullet$};
\node [] (a) at (3.75,0) {$\bullet$};
\node [] (a) at (5.25,0) {$\bullet$};
\node [] (a) at (6,0) {$\bullet$};
\node [] (a) at (7.5,0) {$\bullet$};

\node [blue] (a) at (1.5,1.5) {$Z$};
\node [blue] (a1) at (2,1.5) {$\bullet$};
\node [blue] (a2) at (2.75,1.5) {$\bullet$};
\node [blue] (a) at (3.5,1.5) {$\bullet$};
\node [blue] (a) at (4.25,1.5) {$\bullet$};
\node [blue] (a) at (5,1.5) {$\bullet$};

\node [] (a) at (2,1.8) {$u_1$};
\node [] (a) at (2.75,1.8) {$u_2$};
\node [] (a) at (3.5,1.8) {$u_3$};
\node [] (a) at (4.25,1.8) {$u_4$};
\node [] (a) at (5,1.8) {$u_5$};

 \draw[] (0.75,0)-- (1.5,0) --(2.25,0) --
(3,0) --
 (3.75,0)-- 
 (4.5,0)-- 
 (5.25,0)--(6,0) --(6.75,0)--(7.5,0); 

 \draw[black]
(2,1.5) -- (1.5,0)
(5,1.5)-- (6,0)
(2,1.5) -- (2.25,0)
(2.75,1.5)-- (3.75,0)
(2.75,1.5)-- (5.25,0)
(4.25,1.5)to[out=-10,in=60](5.25,0)
(4.25,1.5)to[out=-90,in=100](5.25,0)
(3.5,1.5)--(0.75,0)
;

 \draw[red] (0.7,-0.2)-- (0.8,-0.2) 
(2.25,-0.2) --(3.75,-0.2) 
(6,-0.2) --(6.75,-0.2)--(7.5,-0.2); 

\node [] (a) at (0.78,-0.45) {$Q'_1$};
\node [] (a) at (3,-0.45) {$Q'_2$};
\node [] (a) at (6.5,-0.45) {$Q'_4$};
\end{tikzpicture}
\end{subfigure}

\caption{An example of intersection of $\CC$ and a path $P\in \PP$. 
The vertices in $P$ colored black belong to $V(P')$. 
The intersection of $\CC$ and $P$ are the set of paths 
$\QQ=\{Q_1,Q_2,Q_3,Q_4\}$. Here $Q_2,Q_3\in \QQ_1$ and $Q_1,Q_4\in \QQ_2$. 
The substitute paths for $\QQ$ is in the figure at the right hand side}
\label{figure_cyclepathintersection}
\end{figure}

Now we construct substitute paths for $\cQ_1$. Here, we construct substitute paths for at least $(1-{\epsilon})\vert \cQ_1\vert$ 
paths and these paths will be vertex disjoint. Moreover, these paths will be vertex disjoint from the paths in $\cQ_2'$ as well. 
Let $P$ be a path in ${\cal P}$ such that at least one path in $\cQ_1$ is a subpath of $P$. 
Let $\cQ_1(P)\subseteq \cQ_1$ be a subset of $\cQ_1$ containing all the paths in $\cQ_1$ that is a subpath of $P$. 
There are at most two paths in $\cQ_2$ which 
are subpaths of $P$. Let $F$ and $L$ be these paths, where $F$ and $L$ could be empty too. 
 Let the neighbours of $F$ and $L$ in $Z$ in the cycle packing $\CC$ be $x$ and $y$, respectively 
(here, $x=\clubsuit$ if $F=\emptyset$ and $y=\clubsuit$ if $L=\emptyset$).  
Then, consider the  following decomposition of path $P=FP^\star L$. We claim that $P^\star=P^{(x,y)}$. That is, $P^\star$ is a path for which we would have created the interval graph $H_P^{(x,y)}$.  
Observe that, if $F$ is non-empty then $x$  does not have any neighbor on $F$ as cycles in $\cal C$ are {\em chordless}. Similarly, if $L$ is non-empty then $y$ does not have any neighbor on $L$. 
Thus, if $F$ and $L$ are both non-empty then indeed the {\em last} vertex of $F$ is the {\em first}  vertex on $P$ that is a neighbor of $x$ and the {\em first} vertex of $L$ is the {\em last} vertex on $P$ that is a neighbor of $y$. This implies that indeed we would have created the interval graph $H_P^{(x,y)}$. We can argue similarly if either $F$ is empty or $L$ is empty.  
Now consider the interval graph $H_P^{(x,y)}$. This graph is constructed from $P^{(x,y)}$. 
Since 
each cycle in $\CC$ is a chordless cycle, we have that 
each subpath $Q\in \cQ_1(P)$ is a potential $(u_1,u_2)$-subpath of $P^{(x,y)}$ where 
either $u_1Qu_2$ or $u_2Qu_1$ is a subpath of $\CC$ and $u_1,u_2\in Z$.
Since each vertex in $V(\CC)$ has degree two in $\CC$, 
for a pair $(u_1,u_2)$ we have {\em at most two potential $(u_1,u_2)$ paths} 
$Q_1$ and $Q_2$ in $\cQ_1(P)$. Also note that these supaths $Q_1$ and $Q_2$ are 
potential $(u_2,u_1)$-subpaths as well. So when there are two paths $Q_1,Q_2\in \cQ_1(P)$ 
such that $u_1Q_1 u_2$ and $u_1 Q_2 u_2$ are subpaths of $\CC$, then we consider $Q_1$ 
as a potential $(u_1,u_2)$-subpath and $Q_2$ as  a potential $(u_2,u_1)$-subpath.  
Now we can consider $\cQ_1(P)$ as a set of potential subpaths of $P^{(x,y)}$. 
That is, for each 
$Q\in \cQ_1(P)$, there is an interval $I^{(u_1,u_2)}_Q$ with label $(u_1,u_2)$ 
and $(u_1,u_2)$ is not a label of any other intervals corresponding to a subpath in 
$\cQ_1(P)\setminus \{Q\}$.
Let ${\cal I}_1(P)$ be the set of interval created for the potential subpaths in $\cQ_1(P)$. 
We have explained that for 
any $(u_1,u_2)\in Z \times Z$, there is at most one potential $(u_1,u_2)$-subpath in $\cQ_1(P)$. 
Also notice that, since $\cQ_1(P)$ is a collection of vertex disjoint paths, the interval 
constructed for corresponding potential subpaths are disjoint. This 
implies that ${\cal I}_1(P)$ is an independent set in  $H_P^{(x,y)}$ and 
$\vert \Gamma_P^{(x,y)}({\cal I}_1(P))\vert= \vert {\cal I}_1(P)\vert$. 
By Lemma~\ref{lem:uni_lab_is}, we have that 
there is a subset $\Sigma'\subseteq \Gamma_P^{(x,y)}({\cal I}_1(P))$
such that there is an independent set $S$ of cardinality $(1-\frac{\epsilon}{2})\vert {\cal I}_1(P)\vert$
in $X_P^{(x,y)}$ and $\Gamma_P^{(x,y)}(S)=\Sigma'$. This implies that there are at least 
$(1-\frac{\epsilon}{2})\vert {\cal I}_1(P)\vert = (1-\frac{\epsilon}{2})\vert \cQ_1(P)\vert$ 
of paths in $\cQ_1(P)$ has {\em substitute paths} in $P'$ which are vertex disjoint from 
$F$ and $L$, where $P'$ is the path obtained from $P$ in the reduction process using Lemma~\ref{lem:uni_lab_is}. This implies that for 
each $P\in \cP$, at least $(1-\frac{\epsilon}{2})\vert \cQ_1(P)\vert$ paths has 
substitute paths in $G''$ and they are vertex disjoint subpaths of $P'$ and does not intersect 
with $F$ and $L$. We denote the set of substitute paths in $P'$ by $\cQ_1'(P')$. 
This implies that the substitute paths for $\cQ_1$ are vertex disjoint and they are vertex disjoint from the 
substitute paths for $\cQ_2$. Let these substitute paths form a set $\cQ_1'=\cup_{P\in \cP}\cQ_1'(P')$. 
Also notice that 
since each vertex $u\in Z$, has degree at most $2$ in $\CC$ and $\vert Z \vert\leq OPT(G',k)$, 
the total number of paths in $\cQ_1$ is at most $2OPT(G',k)$.
From each $\cQ_1(P)$, at least $(1-\frac{\epsilon}{2})\vert \cQ_1(P)\vert$ paths have 
substitute paths in $G''$. Recall that, 
\[ \cQ_1 =\biguplus_{P\in \cP}\cQ_1(P) \mbox{ and }  \cQ_1' =\biguplus_{P\in \cP}\cQ_1'(P').\]
That is, $\cQ_1$ ($\cQ_1'$) is the disjoint union of $\cQ_1(P)$ ($\cQ_1'(P')$) for $P\in \cP$. Thus, 
\begin{eqnarray*}
\vert \cQ_1\vert -  \vert \cQ_1' \vert &= & \sum_{P\in \cP} \vert \cQ_1(P) \vert -  \vert \cQ_1'(P') \vert \\
& \leq & \sum_{P\in \cP} \vert \cQ_1(P) \vert - \left(1- \frac{\epsilon}{2}\right) \vert \cQ_1(P) \vert \\
& = & \left(\sum_{P\in \cP} \frac{\epsilon}{2} \vert \cQ_1(P) \vert  \right)=  \frac{\epsilon}{2} \vert \cQ_1 \vert.
\end{eqnarray*}
This implies that $\vert \cQ_1\vert-\vert \cQ_1'\vert \leq  \frac{\epsilon}{2} \vert \cQ_1 \vert \leq \epsilon OPT(G',k)$. This implies that  $\cQ_1'\cup \cQ_2'$ contains at least $(1-\epsilon)OPT(G',k)$ substitute paths. Each path in $\cQ$ for which we do not have a substitute path can destroy at most one cycle in 
$\CC$. Recall that, $\CC$  is an optimum solution to $(G',k)$. This implies that $G''$
contains at least $(1-\epsilon)OPT(G',k)$ vertex disjoin cycles. 
This completes the proof of the claim. 
\end{proof}
By Lemma~\ref{lem:CP_path_creation}, we know that $OPT(G',k)=OPT(G,k)$ and hence 
by Claim~\ref{claim:optCP} we get that $OPT(G'',k)\geq (1-\epsilon)OPT(G,k)$.   
We know that given a solution ${\cal S}'$ of $(G'',k)$ the solution lifting algorithm will output a solution ${\cal S}$ of same cardinality
for $(G,k)$. Therefore, we have  
$$\frac{\vert{\cal S}\vert}{OPT(G,k)}\geq (1-\epsilon) \frac{\vert{\cal S}'\vert}{OPT(G'',k)}.$$
This gives the desired PSAKS for \CP{} and completes the proof.
 \end{proof}

%% file: figure_CP.tex
 \begin{figure}

 \begin{subfigure}[b]{0.5\textwidth}
        \centering

\begin{tikzpicture}[scale=1.3]

\node [] (a) at (0,0.75) {$\bullet$};
\node [] (a) at (1,0.75) {$\bullet$};
\node [] (a) at (2,0.75) {$\bullet$};
\node [] (a) at (2.3,0.8) {$o_2$};
\node [] (a) at (3,0.75) {$\bullet$};
\node [] (a) at (4,0.75) {$\bullet$};
\node [] (a) at (1.25,0.85) {$s_1$};

\node [] (a) at (0,1.25) {$\bullet$};
\node [] (a) at (1,1.25) {$\bullet$};
\node [] (a) at (0,1.75) {$\bullet$};
\node [] (a) at (1,1.75) {$\bullet$};
\node [] (a) at (0.5,2.25) {$\bullet$};
\node [] (a) at (0.25,2.27) {$s_2$};
\node [] (a) at (0.65,2.75) {$\bullet$};
\node [] (a) at (0.8,3.25) {$\bullet$};
\node [] (a) at (0.95,4) {$\bullet$};
\node [] (a) at (1.5,5) {$\bullet$};
\node [] (a) at (1.5,5.2) {$o_1$};

\draw (0,0.75)--
 (0,1.25) --
(0,1.75)  -- 
(0.5,2.25) -- (0.65,2.75) -- (0.8,3.25)--  (0.95,4) -- (1.3,4.5)--(1.5,5);
\draw (1,0.75) --
(1,1.25) -- (1,1.75)--(0.5,2.25);

\node [] (a) at (2.5,1.5) {$\bullet$};
\node [] (a) at (2.7,1.7) {$s_3$};
\node [] (a) at (2.3,2) {$\bullet$};

\node [] (a) at (2.1,2.5) {$\bullet$};
\node [] (a) at (1.9,3) {$\bullet$};
\node [] (a) at (1.7,3.5) {$\bullet$};
\node [] (a) at (1.5,4) {$\bullet$};
\node [] (a) at (1.3,4.5) {$\bullet$};
\node [] (a) at (1.6,4.6) {$q_1$};

\draw (2,0.75)--
(2.5,1.5) --(2.3,2) --
(2.1,2.5)--(1.9,3) -- (1.7,3.5) -- (1.5,4) -- (1.3,4.5);
\draw (3,0.75)--(2.5,1.5);


\node [] (a) at (4,1.25) {$\bullet$};
\node [] (a) at (4,1.75) {$\bullet$};
\node [] (a) at (4,2.25) {$\bullet$};
\node [] (a) at (4,2.75) {$\bullet$};
\node [] (a) at (4,3.25) {$\bullet$};
\node [] (a) at (4,4) {$\bullet$};
\node [] (a) at (4,4.5) {$\bullet$};
\node [] (a) at (4,5) {$\bullet$};

\node [] (a) at (2.75,5) {$\bullet$};

\draw (4,0.75)--
 (4,1.25) --(4,1.75) -- (4,2.25) -- (4,2.75) -- (4,3.25) -- (4,4) -- (4,4.5) -- (4,5)-- (2.75,5)  --  (1.5,5);

\node [blue] (a) at (-0.5,-0.25) {$\bullet$};
\node [blue] (a) at (0.5,-0.25) {$\bullet$};
\node [blue] (a) at (1.5,-0.25) {$\bullet$};
\node [blue] (a) at (2.5,-0.25) {$\bullet$};
\node [blue] (a) at (3.5,-0.25) {$\bullet$};

\node [blue] (a) at (-0.5,-0.5) {$f_1$};
\node [blue] (a) at (0.5,-0.5) {$z_2$};
\node [blue] (a) at (1.5,-0.5) {$z_1$};
\node [blue] (a) at (2.5,-0.5) {$f_2$};
\node [blue] (a) at (3.5,-0.5) {$z_3$};

\draw[blue] (2.5,-0.25) -- (3.5,-0.25);

\draw[red] (1.5,-0.25) to [out=150,in=270] (1,0.75);
\draw[red] (1.5,-0.25) to [out=70,in=-30] (1,0.75);
\draw[red] (3.5,-0.25) -- (4,1.75)
 (3.5,-0.25) -- (4,2.75) 
(3.5,-0.25) -- (4,4.5);
\draw[red] (1.5,-0.25) to [out=70,in=-30] (0.8,3.25);

\draw[red] 
(3.5,-0.25) --(2,0.75) 
 (3.5,-0.25) -- (3,0.75) 
 (3.5,-0.25) to [out=85,in=0] (2.5,1.5);  


\draw[red] (2,0.75) -- (2.5,-0.25);

\draw[red] (0.5,-0.25) -- (0,1.25);
\draw[red] (0.5,2.25)-- (0.5,-0.25);

\draw[red] (-0.5,-0.25)  to [out=90,in=180]  (1.5,5);

\end{tikzpicture}
\caption{$F=\{f_1,f_2,z_1,z_2,z_3\}$ is a feedback vertex set of $G$ and the tree $G-F$ is rooted 
at $o_1$. Algorithm ${\cal B}$ will output $F'=\{z_1,z_2,z_3\}$ and 
$S=\{s_1,s_2,s_3\}$ when $k>3$.} 
\label{fig:atreemarking}
\end{subfigure}
\begin{subfigure}[b]{0.5\textwidth}
        \centering

\begin{tikzpicture}[scale=1.2]


\node [] (a) at (-0.25,-1.2) {$s_2$};
\node [] (a) at (0,-1) {$\bullet$};
\node [] (a) at (1,-1) {$\bullet$};
\node [] (a) at (0.75,-1.2) {$s_1$};
\node [] (a) at (4.2,-1) {$\bullet$};
\node [] (a) at (4.4,-1.2) {$s_3$};

\draw[red] (1.5,-0.25) to [out=210,in=90] (1,-1);
\draw[red] (1.5,-0.25) to [out=-45,in=30] (1,-1);

\draw[red] (3.5,-0.25) -- (4,1.75)
 (3.5,-0.25) -- (4,2.75) 
(3.5,-0.25) -- (4,4.5);
\draw[red] (1.5,-0.25) to [out=70,in=-30] (0.8,3.25);

\draw[red] (3.5,-0.25) --(2,0.75) 
 (3.5,-0.25) -- (3,0.75) 
 (3.5,-0.25)--(4.2,-1);

\draw[red] (2,0.75) -- (2.5,-0.25);

\draw[red] (2,0.75) -- (2.5,-0.25);

\draw[red]  (0,-1) -- (0.5,-0.25);
\draw[red] (0,1.25)-- (0.5,-0.25);

\draw[red] (-0.5,-0.25)  to [out=90,in=180]  (1.5,5);

\node [] (a) at (0,0.75) {$\bullet$};
\node [] (a) at (2,0.75) {$\bullet$};
\node [] (a) at (3,0.75) {$\bullet$};
\node [] (a) at (4,0.75) {$\bullet$};
\node [] (a) at (2.75,0.75) {$w_1$};

\node [] (a) at (0,1.25) {$\bullet$};
\node [] (a) at (1,1.25) {$\bullet$};
\node [] (a) at (0,1.75) {$\bullet$};
\node [] (a) at (1,1.75) {$\bullet$};
\node [] (a) at (0.65,2.75) {$\bullet$};
\node [] (a) at (0.8,3.25) {$\bullet$};
\node [] (a) at (0.95,4) {$\bullet$};
\node [] (a) at (1.5,5) {$\bullet$};

\draw (0,0.75)-- (0,1.25)
 --(0,1.75)
(0.65,2.75) -- 
(0.8,3.25)--  
(0.95,4) -- (1.3,4.5)--(1.5,5);
\draw[green]
(0,1.75) to [out=210,in=90] (0,-1)
(1,1.75) to [out=250,in=-20] (0,-1)
(0.65,2.75) to [out=250,in=20] (0,-1);

\draw [green]
(1,-1)--
(1,1.25);
\draw (1,1.25)
 -- (1,1.75);

\node [] (a) at (2.3,2) {$\bullet$};

\node [] (a) at (2.1,2.5) {$\bullet$};
\node [] (a) at (1.9,3) {$\bullet$};
\node [] (a) at (1.7,3.5) {$\bullet$};
\node [] (a) at (1.5,4) {$\bullet$};
\node [] (a) at (1.3,4.5) {$\bullet$};

\draw[green] (2,0.75) to [out=-45,in=180]
(4.2,-1) to [out=50,in=-40] (2.3,2);
\draw
(2.3,2) --
(2.1,2.5)--(1.9,3) -- (1.7,3.5) -- (1.5,4) -- (1.3,4.5);
\draw[green] (3,0.75)  to [out=0,in=90] (4.2,-1);


\node [] (a) at (4,1.25) {$\bullet$};
\node [] (a) at (4,1.75) {$\bullet$};
\node [] (a) at (4,2.25) {$\bullet$};
\node [] (a) at (4,2.75) {$\bullet$};
\node [] (a) at (4,3.25) {$\bullet$};
\node [] (a) at (4,4) {$\bullet$};
\node [] (a) at (4,4.5) {$\bullet$};
\node [] (a) at (4,5) {$\bullet$};

\node [] (a) at (2.75,5) {$\bullet$};

\draw (4,0.75)--
 (4,1.25) --(4,1.75) -- (4,2.25) -- (4,2.75) -- (4,3.25) -- (4,4) -- (4,4.5) -- (4,5)-- (2.75,5)-- (1.5,5);

\node [blue] (a) at (-0.5,-0.25) {$\bullet$};
\node [blue] (a) at (0.5,-0.25) {$\bullet$};
\node [blue] (a) at (1.5,-0.25) {$\bullet$};
\node [blue] (a) at (2.5,-0.25) {$\bullet$};
\node [blue] (a) at (3.5,-0.25) {$\bullet$};

\node [blue] (a) at (0.25,-0.2) {$z_2$};
\node [blue] (a) at (1.75,-0.5) {$z_1$};
\node [blue] (a) at (3.5,-0.5) {$z_3$};

\draw[blue] (2.5,-0.25) -- (3.5,-0.25);

\end{tikzpicture}
\caption{The vertices of $S$ is drawn separately with edges between $S$ and $G-F$ colored green}
\label{fig:btreemarking}
\end{subfigure}

\caption{An example of Lemma~\ref{lemma:mark_in_tree}}
\label{fig:treemarking}
\end{figure}

%% file: figure_CPPath.tex
 \begin{figure}

        \centering

\begin{tikzpicture}[rotate=90, scale=1.3]

\node [green] (a) at (0.2,7) {$P$};
\node [green] (a) at (0,0.75) {$\bullet$};
\node [green] (a) at (0,1.5) {$\bullet$};
\node [green] (a) at (0,2.25) {$\bullet$};
\node [green] (a) at (0,3) {$\bullet$};
\node [green] (a) at (0,3.75) {$\bullet$};
\node [green] (a) at (0,4.5) {$\bullet$};
\node [green] (a) at (0,5.25) {$\bullet$};
\node [green] (a) at (0,6) {$\bullet$};
\node [green] (a) at (0,6.75) {$\bullet$};

\node [blue] (a) at (2.6,6) {$Z$};
\node [blue] (a1) at (2.5,2) {$\bullet$};
\node [blue] (a2) at (2.5,2.75) {$\bullet$};
\node [blue] (a) at (2.5,3.5) {$\bullet$};
\node [blue] (a) at (2.5,4.25) {$\bullet$};
\node [blue] (a) at (2.5,5) {$\bullet$};

\node [] (a) at (2.8,2) {$u_5$};
\node [] (a) at (2.8,2.75) {$u_4$};
\node [] (a) at (2.8,3.5) {$u_3$};
\node [] (a) at (2.8,4.25) {$u_2$};
\node [] (a) at (2.8,5) {$u_1$};

\draw[blue] (2.5,2)--(2.5,2.75);

 \draw[green] (0,0.75)-- (0,1.5) --(0,2.25) --
(0,3) --
 (0,3.75)-- 
 (0,4.5)-- 
 (0,5.25)--(0,6) --(0,6.75); 

 \draw[black]
(0,0.75)--(2.5,2.75)
(0,0.75)--(2.5,3.5) 
(0,1.5) --(2.5,5)
(0,2.25) --(2.5,3.5)
(0,3) --(2.5,4.25)
 (0,3.75)-- (2.5,2)
 (0,4.5)-- (2.5,5)
 (0,5.25)--(2.5,3.5)
(0,6) --(2.5,2)
(0,6.75)-- (2.5,3.5)
(0,3.75)--(2.5,3.5)
; 

\draw(0,4.5) to[out=45,in=115] (2.5,5);

\node [red] (a) at (2.5,-4) {$\bullet$};
\node [red] (a) at (2.5,-3.25) {$\bullet$};
\node [red] (a) at (2.5,-2.5) {$\bullet$};
\node [red] (a) at (2.5,-1.75) {$\bullet$};
\node [red] (a) at (2.5,-1) {$\bullet$};
\node [red] (a) at (2.5,0) {$R$};

\draw[red] (2.5,-2.5)--(0,0.75);
\draw[red] (2.5,-1)--(0,6.75);

\draw[red] (2.5,-4) to[out=0,in=45] (2.5,-2.5);

\end{tikzpicture}

\caption{An example of a path $P$ in $\cP$, $Z$ and $R$}
\label{figure_CP_paths}
\end{figure}

%% file: prevWork.tex
\section{Approximate Kernelization in Previous Work}\label{sec:previousWork}
In this section we show how some of the existing approximation algorithms and \FPT{} approximation algorithms can be re-interpreted as first computing an $\alpha$-approximate kernel, and then running a brute force search or an approximation algorithm on the reduced instance.

\subsection{Partial Vertex Cover}
\input{partial_vc}
\subsection{Steiner Tree}
\input{SteinerTree.tex}
\subsection{Optimal Linear Arrangement}
\input{OLA.tex}

%% file: partial_vc.tex
In the \PVC{} problem the input is a graph $G$ on $n$ vertices, and an integer $k$. 
The task is to find a vertex set $S \subseteq V(G)$ of size $k$, maximizing the number of edges with at least one 
end-point in $S$. We will consider the problem parameterized by the solution {\em size} $k$. 
Note that the solution size is {\em not} the objective function value. We define \PVC{} as a parameterized 
optimization problem as follows. 
\[
    PVC(G,k,S)= 
\begin{cases}
    -\infty & \text{$\vert S \vert >k$} \\
    \text{Number of edges incident on } S & \text{Otherwise}
\end{cases}
\]

\PVC{} is \WOne-hard~\cite{GuoNW07}, thus we do not expect an \FPT{} algorithm or a kernel of {\em any} size to exist for this problem. On the other hand, Marx~\cite{MarxFPT-AS} gave a $(1+\epsilon)$-approximation algorithm for the problem with running time $f(k,\epsilon)n^{\OO(1)}$. We show here that the approximation algorithm of Marx~\cite{MarxFPT-AS} can be re-interpreted as a PSAKS.

\begin{theorem}
\PVC{} admits a strict time and size efficient PSAKS. 
\end{theorem}

\begin{proof}
We give an $\alpha$-approximate kernelization algorithm for the problem for every $\alpha>1$. 
Let $\epsilon=1-\frac{1}{\alpha}$ and $\beta=\frac{1}{\epsilon}$. 
Let $(G,k)$ be the input instance. Let $v_1,v_2,\ldots,v_n$ be the vertices of $G$ in the non-increasing order of 
degree, i.e 
$d_G(v_i)\geq d_G(v_j)$ for all $1\geq i>j\geq n$. 
The kernelization algorithm has two cases based on degree of $v_1$. 

\noindent
\textbf{Case 1: $d_G(v_1)\geq \beta {k \choose 2}$.}
In this case $S=\{v_1,\ldots,v_k\}$ is a $\alpha$-approximate solution. 
The number of edges incident to $S$ is at least $(\sum_{i=1}^k d_G(v_i))-{k \choose 2}$, because 
at most ${k \choose 2}$ edges have both end points in $S$ and they are counted twice in the sum $(\sum_{i=1}^k d_G(v_i))$. 
The value of the optimum solution is at most $\sum_{i=1}^k d_G(v_i)$. 
Now consider the value, $PVC(G,k,S)/OPT(G,k)$. 
$$\frac{PVC(G,k,S)}{OPT(G,k)}\geq \frac{(\sum_{i=1}^k d_G(v_i))-{k \choose 2}}{\sum_{i=1}^k d_G(v_i)}\geq 1-\frac{{k \choose 2}}{d_G(v_1)}\geq 1-\frac{1}{\beta}=\frac{1}{\alpha}$$
The above inequality implies that $S$ is an $\alpha$-approximate solution. 
So the kernelization algorithm outputs a trivial instance $(\emptyset,0)$ in this case.

\noindent
\textbf{Case 2: $d_G(v_1)< \beta {k \choose 2}$.}
Let $V'=\{v_1,v_2,\ldots,v_{k \lceil \beta {k \choose 2} \rceil +1}\}$. 
In this case the algorithm outputs $(G',k)$, where $G'=G[N_G[V']]$. 
We first clam that $OPT(G',k)=OPT(G,k)$. Since $G'$ is a subgraph of $G$, 
$OPT(G',k)\leq OPT(G,k)$. Now it is enough to show that $OPT(G',k)\geq OPT(G,k)$. 
Towards that, we prove that there is an optimum solution that contains only vertices from the set $V'$. Suppose not, then 
 consider the solution $S$ which is lexicographically smallest in the ordered list $v_1,\ldots v_n$. The set $S$ contains at most $k-1$ vertices from $V'$ and at least 
 one from $V\setminus V'$. 
Since degree of each vertex in $G$ is at most $\lceil \beta {k \choose 2} \rceil-1$ and 
$\vert S \vert \leq k$, we have that $\vert N_G[S]\vert \leq k \lceil \beta {k \choose 2} \rceil$.
This implies that there exists a vertex $v\in V'$ such that $v\notin N_G[S]$. 
 Hence 
 by including the vertex $v$ and removing a vertex from $S\setminus V'$, we can cover at least as many edges as $S$ 
 can cover. 
 This contradicts our assumption that $S$ is lexicographically smallest. 
  Since $G'$ is a subgraph of $G$ any solution of $G'$ is also a solution 
 of $G$. Thus we have shown that $OPT(G',k)=OPT(G,k)$. 
 So the algorithm returns the instance $(G',k)$  as the reduced instance. 
Since $G'$ is a subgraph of $G$, in this case, the solution lifting algorithm 
takes a solution $S'$ of $(G',k)$ as input and outputs $S'$ as a solution of $(G,k)$. 
Since $OPT(G',k)=OPT(G,k)$, it follows that 
$\frac{PVC(G,k,S')}{OPT(G,k)}=\frac{PVC(G',k,S')}{OPT(G',k)}$.

 The number of vertices in the reduced instance is 
 $\OO(k\cdot \lceil \frac{1}{\epsilon}{k\choose 2}\rceil^2)=\OO(k^5)$. 
 The running time of the algorithm is polynomial in the size of $G$. 
 Since the algorithm either finds an $\alpha$-approximate solution (Case 1) or reduces the instance by a \onesafe{}  reduction rule (Case 2), 
 this kernelization scheme is strict.
\end{proof}

%% file: SteinerTree.tex
In the \STREE{} problem we are given as input a graph $G$, a subset $R$ of $V(G)$ called the {\em terminals} and a weight function $w : E(G) \rightarrow \mathbb{N}$. A {\em Steiner tree} is a subtree $T$ of $G$ such that $R \subseteq V(T)$, and the {\em cost} of a tree $T$ is defined as $w(T) = \sum_{e \in E(T)} w(e)$. The task is to find a Steiner tree of minimum cost. We may assume without loss of generality that the input graph $G$ is complete and that $w$ satisfies the triangle inequality: for all $u, v, w \in V(G)$  we have $w(uw) \leq w(uv) + w(vw)$. This assumption can be justified by adding for every pair of vertices $u$,$v$ the edge $uv$ to $G$ and making the weight of $uv$ equal the shortest path distance between $u$ and $v$. If multiple edges are created between the same pair of vertices, only the lightest edge is kept.

Most approximation algorithms for the {\sc Steiner Tree} problem rely on the notion of a $k$-restricted Steiner tree, defined as follows. A {\em component} is a tree whose leaves coincide with a subset of terminals, and a $k$-component is a component with at most $k$ leaves. A $k$-restricted Steiner tree ${\cal S}$ is a collection of $k$-components, such that the union of these components is a Steiner tree $T$. The cost of  ${\cal S}$ is the sum of the costs of all the $k$-components in ${\cal S}$. Thus an edge that appears in several different $k$-components of ${\cal S}$ will contribute several times to the cost of ${\cal S}$, but only once to the cost of $T$. The following result by Borchers and Du~\cite{BorchersD97} shows that for every $\epsilon > 0$ there exists a $k$ such that the cost of the best $k$-restricted Steiner tree ${\cal S}$ is not more than $(1+\epsilon)$ times the cost of the best Steiner tree. Thus approximation algorithms for {\sc Steiner Tree} only need to focus on the best possible way to ``piece together" $k$-components to connect all the terminals.

\begin{proposition}[\cite{BorchersD97}]\label{prop:restricted} For every $k \geq 1$, graph $G$, terminal set $R$, weight function $w : E(G) \rightarrow \mathbb{N}$ and Steiner tree $T$, there is a $k$-restricted Steiner Tree ${\cal S}$ in $G$ of cost at most $(1+\frac{1}{\lfloor \log_2 k \rfloor}) \cdot w(T)$.
\end{proposition}

Proposition~\ref{prop:restricted} can easily be turned into a PSAKS for {\sc Steiner Tree} parameterized by the number of terminals, defined below.
\[
    ST((G,R),k',T)= \left\{
\begin{array}{rl}
    -\infty & \text{if } \text{$\vert R \vert >k'$} \\
   \infty & \text{if }  T  \text{ is not a Steiner tree for } R \\
    w(T) & \text{otherwise}
\end{array}\right.
\]
To get a $(1 + \epsilon)$-approximate kernel it is sufficent to pick $k$ based on $\epsilon$, compute for each $k$-sized subset $R' \subseteq R$ of terminals an optimal Steiner tree for $R'$, and only keep vertices in $G$ that appear in these Steiner trees. This reduces the number of vertices of $G$ to $\OO(|R|^k)$, but the edge weights can still be large making the bitsize of the kernel super-polynomial in $|R|$. However, it is quite easy to show that keeping only $\OO(\log |R|)$ bits for each weight is more than sufficient for the desired precision.

\begin{theorem}
{\sc Steiner Tree} parameterized by the number of terminals admits a PSAKS.
\end{theorem}

\begin{proof}
Start by computing a $2$-approximate Steiner tree $T_2$ using the classic factor $2$ approximation algorithm~\cite{bookApprox}. For every vertex $v \notin R$ such that $\min_{x \in R} w(vx) \geq w(T_2)$ delete $v$ from $G$ as $v$ may never participate in any optimal solution. By the triangle inequality we may now assume without loss of generality that for every edge $uv \in E(G)$ we have  $w(uv) \leq 6OPT(G,R,w)$.

Working towards a $(1+\epsilon)$-approximate kernel of polynomial size, set $k$ to be the smallest integer such that $\frac{1}{\lfloor \log_2 k \rfloor} \leq \epsilon/2$. For each subset $R'$ of $R$ of size at most $k$, compute an optimal steiner tree $T_{R'}$ for the instance $(G,R',w)$ in time $\OO(3^k|E(G)||V(G)|)$ using the algorithm of Dreyfus and Wagner\cite{DreyfusW71}. Mark all the vertices in $V(T_{R'})$. After this process is done, some $\OO(k|R|^k)$ vertices in $G$ are marked. Obtain $G'$ from $G$ by deleting all the unmarked vertices in $V(G) \setminus R$. Clearly every Steiner tree in $G'$ is also a Steiner tree in $G$, we argue that $OPT(G',R,w) \leq (1+\frac{\epsilon}{2})OPT(G,R,w)$.

Consider an optimal Steiner tree $T$ for the instance $(G,R,w)$. By Proposition~\ref{prop:restricted} there is a $k$-restricted Steiner Tree ${\cal S}$ in $G$ of cost at most $(1+\frac{1}{\lfloor \log_2 k \rfloor}) \cdot w(T) \leq (1+\frac{\epsilon}{2})OPT(G,R,w)$. Consider a $k$-component $C \in {\cal S}$, and let $R'$ be the set of leaves of $C$ - note that these are exactly the terminals appearing in $C$. $C$ is a Steiner tree for $R'$, and so $T_{R'}$ is a Steiner tree for $R'$ with $w(T_{R'}) \leq w(C)$. Then ${\cal S}' = ({\cal S} \setminus \{C\}) \cup \{T_{R'}\}$ is a $k$-restricted Steiner Tree of cost no more than $(1+\frac{\epsilon}{2})OPT(G,R,w)$.  Repeating this argument for all $k$-components of ${\cal S}$ we conclude that there exists a $k$-restricted Steiner Tree ${\cal S}$ in $G$ of cost at most $(1+\frac{\epsilon}{2})OPT(G,R,w)$, such that all $k$-components in ${\cal S}$ only use marked vertices. The union of all of the $k$-components in ${\cal S}$ is then a Steiner tree in $G'$ of cost at most  $(1+\frac{\epsilon}{2})OPT(G,R,w)$.

We now define a new weight function $\hat{w} : E(G') \rightarrow \mathbb{N}$, by setting 
$$\hat{w}(e) = \left\lfloor w(e) \cdot \frac{4|R|}{\epsilon \cdot  OPT(G,R,w)} \right\rfloor$$
Note that since $w(e) \leq 6 \cdot OPT(G,R,w)$ it follows that $\hat{w}(e) \leq \frac{24|R|}{\epsilon}$. Thus it takes only $\OO(\log |R| + \log \frac{1}{\epsilon})$ bits to store each edge weight. It follows that the bitsize of the instance $(G',R,\hat{w})$ is $|R|^{2^{\OO(1/\epsilon)}}$.
We now argue that, for every $c \geq 1$, a $c$-approximate Steiner tree $T'$ for the instance $(G',R,\hat{w})$ is also a $c(1+\epsilon)$-approximate Steiner tree for the instance $(G,R,w)$.

First, observe that the definition of $\hat{w}$ implies that for every edge $e$ we have the inequality
$$w(e) \leq \hat{w}(e) \cdot \frac{\epsilon \cdot OPT(G,R,w)}{4|R|} + \frac{\epsilon \cdot OPT(G,R,w)}{4|R|}.$$
In a complete graph that satisfies the triangle inequality, a Steiner tree on $|R|$ terminals has at most $|R|-1$ non-terminal vertices. Thus it follows that $T'$ has at most $2|R|$ edges. Therefore,
$$w(T') \leq \hat{w}(T') \cdot \frac{\epsilon \cdot OPT(G,R,w)}{4|R|} + \frac{\epsilon}{2}OPT(G,R,w).$$
Consider now an optimal Steiner tree $Q$ for the instance $(G',R,w)$. We have that 
$$w(Q) \cdot \frac{4|R|}{\epsilon \cdot  OPT(G,R,w)} \geq \hat{w}(Q),$$
which in turn implies that
$$OPT(G',R,w) \geq OPT(G',R,\hat{w}) \cdot \frac{\epsilon \cdot  OPT(G,R,w)}{4|R|}.$$
We can now wrap up the analysis by comparing $w(T')$ with $OPT(G,R,w)$.
\begin{eqnarray*}
w(T') & \leq & \hat{w}(T') \cdot \frac{\epsilon \cdot OPT(G,R,w)}{4|R|} + \frac{\epsilon}{2}OPT(G,R,w) \\
& \leq & c \cdot OPT(G',R, \hat{w}) \cdot \frac{\epsilon \cdot OPT(G,R,w)}{4|R|} + \frac{\epsilon}{2}OPT(G,R,w) \\
& \leq & c \cdot OPT(G',R, w) + \frac{\epsilon}{2}OPT(G,R,w) \\
& \leq & c \cdot (1 + \epsilon/2) \cdot OPT(G,R, w) + \frac{\epsilon}{2}OPT(G,R,w) \\
& \leq & c \cdot (1 + \epsilon) \cdot OPT(G,R, w)
\end{eqnarray*}
This implies that a $T'$ is a $c(1+\epsilon)$-approximate Steiner tree for the instance $(G,R,w)$, concluding the proof.
\end{proof}

%% file: OLA.tex
In the {\sc Optimal Linear Arrangement} problem we are given as input an undirected graph $G$ on $n$ vertices. The task is to find a permutation $\sigma : V(G) \rightarrow \{1, \ldots, n\}$ minimizing the {\em cost} of $\sigma$. Here the cost of a permutation $\sigma$ is $val(\sigma,G) = \sum_{uv \in E(G)} |\sigma(u)-\sigma(v)|$. 
Recall the problem  {\sc Optimal Linear Arrangement} parameterized by vertex cover: 

\begin{equation*}
OLA((G,C),k,\sigma) = \left\{
\begin{array}{rl}
-\infty & \text{if } C \text{ is not vertex cover of } G \text{ of size at most } k,\\ 
\infty & \text{if } \sigma \text{ is not a linear layout},\\
val(\sigma,G) & \text{ otherwise.}
\end{array} \right.
\end{equation*}

Fellows et al.~\cite{FellowsHRS13} gave an FPT approximation scheme for {\sc Optimal Linear Arrangement} parameterized by vertex cover. 
Their algorithm can be interpreted as a $2^k \cdot (\frac{1}{\epsilon})$ size $(1+\epsilon)$-approximate kernel combined with a brute force algorithm on the reduced instance. Next we explain how to turn the approximation algorithm of Fellows et al. into a $(1+\epsilon)$-approximate kernel.

Let $((G,C),k)$ be the given instance of  {\sc Optimal Linear Arrangement} and $C$ be a vertex cover of $G$. 
The remaining set of vertices $I = V(G) \setminus C$ forms an independent set. Furthermore, $I$ can be partitioned into at most $2^k$ sets: for each subset $S$ of $C$ we define $I_S = \{v \in I ~:~ N(v) = S\}$.  Let $m=\vert E(G)\vert$. 

Based on $\epsilon$ we pick an integer $x = \lfloor \frac{\epsilon n}{4k^2 \cdot k^2 \cdot 2^{k+4}} \rfloor$. From $G$ we make a new graph $G_1$ by deleting for each $S \subseteq C$ at most $x$ vertices from $I_S$, such that the size of $I_S$ becomes divisible by $x$. Clearly $OPT(G_1) \leq OPT(G)$ since $G_1$ is an induced subgraph of $G$. Furthermore, for any ordering $\sigma_1$ of $G_1$ one can make an ordering $\sigma$ of $G$ by appending all the vertices in $V(G) \setminus V(G_1)$ at the end of the ordering. Since there are $2^k$ choices for $S \subseteq C$, each vertex in $V(G) \setminus V(G_1)$ has degree at most $k$ it follows that
\begin{equation}
val(\sigma,G) \leq val(\sigma_1,G_1) +  2^k \cdot x \cdot k \cdot n 
\label{eqn:Gbounded}
\end{equation}
One might think that an additive error of  $2^k \cdot x \cdot k \cdot n$ is quite a bit, however Fellows et al. show that the optimum value is so large that this is quite insignificant.
\begin{lemma}[\cite{FellowsHRS13}]\label{lem:olaOptBound}
$OPT(G) \geq \frac{m^2}{4k^2}$
\end{lemma}

In the following discussion let $I^1$ be the $V(G_1) \setminus C$ and let $I^1_S = \{v \in I^1 ~:~ N(v) = S\}$ for every $S \subseteq C$. Proceed as follows for each $S \subseteq C$. Since $|I^1_S|$ is divisible by $x$, we can group $I^1_S$ into $\frac{|I^1_S|}{x}$ groups, each of size $x$. Define $\widehat{OPT}(G_1)$ to be the value of the best ordering of $G_1$ among the orderings where, for every group, the vertices in that group appear consecutively. Next we prove the following lemma, which states  that there is a near-optimal solution for $G_1$ , where vertices in the same group appear consecutively.

\begin{lemma}\label{lem:olaGroup}  $\widehat{OPT}(G_1) \leq OPT(G_1) +  (kn + m) \cdot 2^k \cdot (k+1) \cdot x$\end{lemma}

To prove Lemma~\ref{lem:olaGroup} we first need an intermediate result. We say that an ordering $\sigma$ is {\em homogenous} if, for every $u$, $v \in I$ such that $N(u) = N(v)$, $\sigma(u) < \sigma(v)$ and there is no $c \in C$ such that $\sigma(u) < \sigma(c) < \sigma(v)$, we have that for every $w$ such that  $\sigma(u) < \sigma(w) < \sigma(v)$, $N(w) = N(u)$. Informally this means that between two consecutive vertices of $C$, the vertices from different sets $I_{S}$ and $I_{S'}$ ``don't mix''.

\begin{lemma}[\cite{FellowsHRS13}] \label{lem:olahom}
There exists a homogenous optimal linear arrangement of $G_1$.
\end{lemma}

\begin{proof}[Proof of Lemma~\ref{lem:olaGroup}]
Consider an optimal linear arrangement of $G_1$ that is homogenous, as guaranteed by Lemma~\ref{lem:olahom}. From such an arrangement one can move around at most $2^k \cdot (k+1) \cdot x$ vertices (at most  $(k+1) \cdot x$ vertices for each set $I^1_S$) and make a new one where for every group, the vertices in that group appear consecutively. Since moving a single vertex of degree at most $k$ in an ordering $\sigma$ can only increase the cost of $\sigma$ by at most $(kn + m)$, this concludes the proof.
\end{proof}

In the graph $G_1$, each set $I^1_S$ is partitioned into $\frac{|I^1_S|}{x}$ groups, each of size $x$. From $G_1$ we can make a new graph $G_2$ by keeping $C$ and exactly one vertex from each group, and deleting all other vertices. Thus $G_2$ is a graph on $|C| + \frac{n-|C|}{x}$ vertices. Each ordering $\sigma_2$ of $V(G_2)$ corresponds to an ordering  $\sigma_1$ of $V(G_1)$ where for every group, the vertices in that group appear consecutively. We will say that the ordering $\sigma_1$ is the ordering of $V(G_1)$ {\em corresponding} to $\sigma_2$. Note that for every ordering of $\sigma_1$ of $V(G_1)$ where for every group, the vertices in that group appear consecutively there is an ordering $\sigma_2$ of $G_2$ such that $\sigma_1$ corresponds to $\sigma_2$. The next lemma summarizes the relationship between the cost of $\sigma_1$ (in $G_1$) and the cost of $\sigma_2$ (in $G_2$).

\begin{lemma}\label{lem:olaIneqLast}
Let $\sigma_2$ be an ordering of $G_2$ and $\sigma_1$ be the ordering of $G_1$ corresponding to $\sigma_2$. Then the following two inequalities hold.
\begin{eqnarray}
val(\sigma_1,G_1) &\leq& x^2 \cdot val(\sigma_2, G_2) \label{eqn:G1bounded}\\
val(\sigma_2,G_2)& \leq& \frac{val(\sigma_1, G_1)}{x^2} + (k+1)\frac{m}{x} + k^2\frac{n}{x}\label{eqn:G2bounded}
\end{eqnarray}
\end{lemma}

\begin{proof}
For the first inequality observe that there is a natural correspondence between edges in $G_1$ and edges in $G_2$. Each edge in $G_2$ corresponds to either $1$ or $x$ edges, depending on whether it goes between two vertices of $C$ or between a vertex in $C$ and a vertex in $V(G_2) \setminus C$. Furthermore, for any edge $uv$ in $G_1$ corresponding to an edge $u'v'$ in $G_2$ we have that $|\sigma_1(u)-\sigma_1(v)| \leq x \cdot |\sigma_2(u')-\sigma_2(v')|$. This concludes the proof of the first inequality.

For the second inequality, observe that for every edge $uv$ in $G_1$ corresponding to an edge $u'v'$ in $G_2$ we have that 
$$|\sigma_2(u')-\sigma_2(v')| \leq \frac{|\sigma_1(u)-\sigma_1(v)|}{x} + k+1.$$
For each edge $u'v'$ in $G_2$ between a vertex in $C$ and a vertex in $V(G_2) \setminus C$, there are exactly $x$ edges in $G_1$ corresponding to it. These $x$ edges contribute at least $x(|\sigma_2(u')-\sigma_2(v')| - k - 1)$ each to $val(\sigma_1,G_1)$, thus contributing at least $|\sigma_2(u')-\sigma_2(v')| - k - 1$ to $\frac{val(\sigma_1, G_1)}{x^2}$. Since there are at most $\frac{m}{x}$ edges in $G_2$, and at most $k^2$ edges between vertices in $C$ (and thus un-accounted for in the argument above), each contributing at most $\frac{n}{x}$ to $val(\sigma_2,G_2)$, the second inequality follows.
\end{proof}
Observe that the second inequality of Lemma~\ref{lem:olaIneqLast} immediately implies that 
\begin{equation}
OPT(G_2) \leq  \frac{\widehat{OPT}(G_1)}{x^2} + (k+1)\frac{m}{x} + k^2\frac{n}{x}\label{eqn:G2bounded_sec}
\end{equation}
We are now ready to state our main result.
\begin{theorem} {\sc Optimal Linear Arrangement} parameterized by vertex cover has a $(1+\epsilon)$-approximate kernel of size 
$\OO(\frac{1}{\epsilon}2^kk^4)$. 
\end{theorem}
\begin{proof}
The kernelization algorithm outputs the graph $G_2$ as described above. $G_2$ has at most $k + n/x \leq \OO(\frac{1}{\epsilon}2^kk^4)$ vertices, so it remains to show how a $c$-approximate solution $\sigma_2$ to $G_2$ can be turned into a $c(1+\epsilon)$-approximate solution $\sigma$ of $G$.

Given a $c$-approximate solution $\sigma_2$ to $G_2$, this solution corresponds to a solution $\sigma_1$ of $G_1$. The ordering $\sigma_1$ of $V(G_1)$ corresponds to an ordering $\sigma$ of $V(G)$, as described in the paragraph right before Lemma~\ref{lem:olaOptBound}. We claim that this ordering $\sigma$ is in fact a $c(1+\epsilon)$-approximate solution $\sigma$ of $G$. 
\begin{eqnarray*}
val(\sigma,G) & \leq & val(\sigma_1, G_1) +  2^k \cdot x \cdot k \cdot n \qquad\qquad\qquad\qquad\qquad\qquad\qquad (\text{By Equation~}(\ref{eqn:Gbounded}))\\
& \leq & x^2 \cdot val(\sigma_2, G_2) +  2^k \cdot x \cdot k \cdot n 
\qquad\qquad\qquad\qquad\qquad\qquad (\text{By Equation~}(\ref{eqn:G1bounded}))\\
& \leq & x^2 \cdot c \cdot OPT(G_2) +  2^k \cdot x \cdot k \cdot n \\
& \leq & x^2 \cdot c \left( \frac{\widehat{OPT}(G_1)}{x^2} + (k+1)\frac{m}{x} + k^2\frac{n}{x}\right) +  2^k \cdot x \cdot k \cdot n 
\qquad (\text{By Equation~}(\ref{eqn:G2bounded_sec}))\\
& \leq & c \cdot \widehat{OPT}( G_1) + c \cdot 3k^2 \cdot n \cdot x + 2^k \cdot x \cdot k \cdot n 
\qquad\qquad\qquad\qquad (\text{Because } m\leq k\cdot n)\\
& \leq & c \cdot \left(OPT(G_1) +  (kn + m) (k+1)  2^k  \cdot x \right) + c \cdot 
2^{k+2} \cdot x \cdot k \cdot n 
\qquad(\text{By Lemma~}\ref{lem:olaGroup})\\
& \leq & c \cdot OPT(G_1) +  c \cdot k^22^{k+4} \cdot x \cdot n 
\qquad\qquad\qquad\qquad\qquad\qquad (\text{Because } m\leq k\cdot n)\\
& \leq & c \cdot OPT(G) +  c \cdot k^22^{k+4} \cdot x \cdot n \\
& \leq & c \cdot OPT(G) +  c \cdot \epsilon \frac{n^2}{4k^2 }  \\
& \leq & c \cdot (1+ \epsilon) \cdot OPT(G)
\qquad\qquad\qquad\qquad\qquad\qquad\qquad (\text{By Lemma~}(\ref{lem:olaOptBound}))\\
\end{eqnarray*}
This concludes the proof.
\end{proof}

%% file: lowerbound.tex
\newcommand{\apptlong}[1]{{#1}-approximate polynomial parameter transformation}
\newcommand{\apptshort}[1]{{#1}-{\sf appt}}
\newcommand{\gappath}{$\alpha$-{\sc Gap Long Path}}

\section{Lower Bounds for Approximate Kernelization}\label{sec:sizelb}
In this section we set up a framework for proving lower bounds on the size of $\alpha$-approximate kernels for a parameterized optimization problem. For normal kernelization, the most commonly used tool for establishing kernel lower bounds is by using cross compositions~\cite{BJK11}. In particular, Bodlaender et al.~\cite{BJK11} defined cross composition and showed that if an \NP{}-hard language $L$ admits a cross composition into a parameterized (decision) problem $\Pi$ and $\Pi$ admits a polynomial kernel, then $L$ has an {\em OR-distillation} algorithm. Fortnow and Santhanam~\cite{FortnowS11} proved that if an \NP{}-hard language $L$ has an OR-distillation, then \NP{} $\subseteq$ \coNPbypoly{}. 

In order to prove a kernelization lower bound for a parameterized decision problem $\Pi$, all we have to do is to find an \NP{}-hard langluage $L$ and give a cross composition from $L$ into $\Pi$. Then, if  $\Pi$ has a polynomial kernel, then combining the cross composition and the kernel with the results of Bodlaender et al.~\cite{BJK11} and Fortnow and Santhanam~\cite{FortnowS11} would prove that \NP{} $\subseteq$ \coNPbypoly{}. In other words a cross composition from $L$ into $\Pi$ proves that $\Pi$ does not have a polynomial kernel unless \NP{} $\subseteq$ \coNPbypoly{}. 

In order to prove lower bounds on the size of $\alpha$-approximate kernels, we generalize the notion of cross compositions to  $\alpha$-gap cross compositions, which are  hybrid of cross compositions and gap creating reductions found in hardness of approximation proofs. To give the formal definition of $\alpha$-gap cross compositions, we first need to recall the definition of Bodlaender et al.~\cite{BJK11} of 
polynomial equivalence relations on $\Sigma^*$, where $\Sigma$ is a finite alphabet.
\begin{definition}[polynomial equivalence relation~\cite{BJK11}] \label{polyEquivalenceRelation}
An equivalence relation~$R$ on~$\Sigma^*$, where $\Sigma$ is a finite alphabet, is called a \emph{polynomial equivalence relation} if (i) equivalence of any~$x,y \in \Sigma^*$ can be checked in time polynomial in~$|x|+|y|$, and (ii) any finite set~$S \subseteq \Sigma^*$ has at most~$(\max _{x \in S} |x|)^{\OO(1)}$ equivalence classes.
\end{definition}
Now we define the notion of $\alpha$-gap cross composition. 
\begin{definition}[$\alpha$-gap cross composition for maximization problem] \label{def:gapcrossComposition}
Let~$L \subseteq \Sigma^*$ be a language, where $\Sigma$ is a finite alphabet and let~$\Pi$ be a parameterized maximization problem. We say 
that~$L$ \emph{$\alpha$-gap cross composes} into~$\Pi$ (where $\alpha\geq 1$), if there is a polynomial equivalence relation~$R$ and an algorithm which, 
given~$t$ strings~$x_1,  \ldots, x_t$ belonging to the same equivalence class of~$R$, computes an 
instance~$(y,k)$ of $\Pi$ and $r\in {\mathbb R}$, in time polynomial in~$\sum _{i=1}^t |x_i|$ such that the following holds: 
\begin{enumerate}[label=(\roman*)]
\setlength{\itemsep}{-2pt}
\item $OPT(y, k) \geq {r}$  if and only if $x_i \in L$ for some $1 \leq i \leq t$; 
\item $OPT(y, k) < \frac{r}{\alpha}  $  if and only if $ x_i \notin L$ for all $1 \leq i \leq t$; and 
\item  $k$ is bounded by a polynomial in $\log t+\max_{\substack{1\leq i\leq t}} |x_i|$.
\end{enumerate}
If such an algorithm exists, then we say that $L$ $\alpha$-gap cross composes to $\Pi$. 
\end{definition}
%

One can similarly define $\alpha$-gap cross compositions for minimization problems. 
\begin{definition}
 \label{def:gapcrossCompositionmin}
The definition of $\alpha$-gap cross composition for minimization problem $\Pi$ can be obtained by replacing conditions $(i)$ and 
$(ii)$ of 
Definition~\ref{def:gapcrossComposition} with the following conditions $(a)$ and $(b)$ respectively: 
$(a)$ $OPT(y, k) \leq {r}$  if and only if $x_i \in L$ for some $1 \leq i \leq t$, 
and 
$(b)$  $OPT(y, k) > {r}\cdot {\alpha} $  if and only if $ x_i \notin L$ for all $1 \leq i \leq t$.
\end{definition}

Similarly to the definition of $\alpha$-approximate kernels, Definition~\ref{def:gapcrossCompositionmin} can be extended to encompass $\alpha$-gap cross composition where $\alpha$ is not a constant, but rather a function of the (output) instance $(y,k)$. Such compositions can be used to prove lower bounds on the size of $\alpha$-approximate kernels where $\alpha$ is super-constant.

%

One of the main ingredient to prove {\em hardness} about computations in different algorithmic models is an appropriate notion of a {\em reduction} from a problem to another. Next, we define a notion of a polynomial 
time reduction appropriate for obtaining lower bounds for $\alpha$-approximate kernels.  As we will see this is very similar to the definition of $\alpha$-approximate polynomial time pre-processing algorithm (Definition~\ref{def:polyTimePreProcessAppx}).   
\begin{definition}
\label{def:approxreduction}
Let $\alpha \geq 1$ be a real number. Let $\Pi$  and $\Pi'$ be two parameterized optimization problems.  
An {\bf \apptlong{$\alpha$}} (\apptshort{$\alpha$} for short)  ${\cal A}$ from $\Pi$ to $\Pi'$ is a pair of polynomial time algorithms,  
called 
reduction algorithm ${\cal R}_{\cal A}$ and solution lifting algorithm. 
Given as input an instance $(I,k)$ of $\Pi$ the reduction algorithm
outputs an instance $(I',k')$ of $\Pi'$.
The solution lifting algorithm 
takes as input an instance $(I,k)$ of $\Pi$, the output instance $(I',k')={\cal R}_{\cal A}(I,k)$ of $\Pi'$, 
and a solution $s'$ to the instance $I'$ and 
outputs a solution $s$ to $(I,k)$. If $\Pi$ is a minimization problem then
\[
\frac{\Pi(I,k,s)}{OPT_{\Pi}(I,k)} \leq \alpha \cdot \frac{\Pi'((I',k'),s')}{OPT_{\Pi'}(I',k')}.
\]
If $\Pi$ is a maximization problem then
\[
\frac{\Pi(I,k,s)}{OPT_{\Pi}(I,k)} \cdot \alpha \geq \frac{\Pi'((I',k'),s')}{OPT_{\Pi'}(I',k')}.
\]
If there is a an \apptshort{$\alpha$} from $\Pi$ to $\Pi'$ then in short we denote it by 
$\Pi \prec_{\apptshort{\alpha}} \Pi' $. 
\end{definition}

In the standard kernelization setting lower bounds machinery also rules out existence of compression algorithms. 
Similar to this our lower bound machinery also rules out existence of  
compression algorithms. Towards that we need to generalize the definition of $\alpha$-approximate kernel
to $\alpha$-approximate compression. The only difference is that in the later case the reduced instance can be an instance of any parameterized optimization problem. 

\begin{definition}
\label{def:approxcompression}
Let $\alpha \geq 1$ be a real number. Let $\Pi$  and $\Pi'$ be two parameterized optimization problems.  
An {\bf $\alpha$-approximate compression}  from $\Pi$ to $\Pi'$ 
is an \apptshort{$\alpha$} ${\cal A}$ from $\Pi$ to $\Pi'$ such that  
$\text{size}_{\cal A}(k) = \sup\{|I'|+k'  : (I',k') = {\cal R}_{\cal A}(I,k), I \in \Sigma^*\}$, 
is upper bounded by a computable function $g : \mathbb{N} \rightarrow \mathbb{N}$, 
where  ${\cal R}_{\cal A}$ is the reduction algorithm in ${\cal A}$. 
\end{definition}

For the sake of proving approximate kernel lower bounds, it is  immaterial 
that $\Pi'$ in the Definition~\ref{def:approxcompression} is a parameterized optimization problem and in fact it can {\em also be a  unparameterized optimization problem}. However, for clarity of presentation we will stick to  parameterized optimization problem in this paper. 
Whenever we talk about an existence of an $\alpha$-approximate compression and we
do not specify the target problem $\Pi'$, we mean the existence of $\alpha$-approximate compression into any 
optimization problem $\Pi'$.
For more detailed exposition about lower bound machinery about polynomial compression 
for decision problems we refer to the textbook~\cite{CyganFKLMPPS15}. 

Towards building a framework for lower bounds we would like to prove a theorem analogous to the one by  Bodlaender et al.~\cite{BJK11}. In particular, we would like to show that  an $\alpha$-gap cross composition from an \NP{}-hard language $L$ into a parameterized optimization problem $\Pi$, together with an $\alpha$-approximate 
compression 
of polynomial size yield an OR-distillation for the language $L$. Then the result of Fortnow and Santhanam~\cite{FortnowS11} would immediately imply that any parameterized optimization problem $\Pi$ that has an  $\alpha$-gap cross composition from an \NP{}-hard language $L$ can not have an $\alpha$-approximate 
compression 
 unless \NP{} $\subseteq$ \coNPbypoly{}. Unfortunately, for technical reasons, it seems difficult to make such an argument. Luckily, we {\em can} complete a very similar argument yielding essentially the same conclusion, but instead of relying on ``OR-distillations'' and the result of Fortnow and Santhanam~\cite{FortnowS11}, we make use of the more general result of Dell and van Melkebeek~\cite{DellM10} that rules  out cheap oracle communication protocols for \NP{}-hard problems. We first give  necessary definitions that allow us to formulate our statements. 

\begin{definition}[Oracle Communication Protocol~\cite{DellM10}]
Let $L\subseteq \Sigma^*$ be a language, where $\Sigma$ is a finite alphabet. An oracle communication protocol 
for the language $L$ is a communication protocol between two players. The first player is given the input $x$
and has to run in time polynomial in the length of $x$; the second player is computationally
unbounded but is not given any part of $x$. At the end of the protocol the first player should be able
to decide whether $x \in L$. The cost of the protocol is the number of bits of communication from the
first player to the second player.
\end{definition}

\begin{lemma}[Complementary Witness Lemma~\cite{DellM10}]
\label{lemma:complementary}
Let $L$ be a language and $t : {\mathbb N} \rightarrow {\mathbb N}$ be
polynomial function such that the problem of deciding whether at least one out of $t(s)$ inputs of
length at most $s$ belongs to $L$ has an oracle communication protocol of cost $\OO(t(s) \log t(s))$, where
the first player can be conondeterministic. Then $L \in$ \coNPbypoly.
\end{lemma}

Our lower bound technique for $\alpha$-approximate 
compression 
for a parameterized optimization problem $\Pi$ requires the  problem $\Pi$ to be {\em polynomial time verifiable}. By this we mean that the function $\Pi$ is computable in polynomial time. 
We call such problems {\em nice parameterized optimization problems}. 
We are now in position 
to prove the main lemma of this section. 
\begin{lemma}
\label{lem:gapcrossComposition}
Let $L$ be 
a 
language and $\Pi$ be a nice parameterized optimization problem. If $L$ $\alpha$-gap cross 
composes to $\Pi$, and $\Pi$ has a polynomial sized $\alpha$-approximate 
compression, 
then 
$L\in$ \coNPbypoly. 
\end{lemma}

\begin{proof}
We prove the theorem for the case when $\Pi$ is a maximization problem. 
The proof when $\Pi$ is a minimization problem is analogous and thus it is omitted. 
By our assumption $L$ 
$\alpha$-gap cross composes to $\Pi$. 
That is, there exists a polynomial time algorithm $\AAA$ that given $t(s)$ strings $x_1,\ldots,x_{t(s)}$, each of 
length at most $s$, outputs an instance $(y,k)$ 
of $\Pi$ and a number $r\in \mathbb{R}$ such that the following holds. 
\begin{itemize}\setlength\itemsep{-.7mm}
 \item[(i)] $OPT(y, k) \geq {r}  $  if and only if $ x_i \in L$ for some $1 \leq i \leq t(s)$
 \item[(ii)] $OPT(y, k) < \frac{r}{\alpha}  $  if and only if $ x_i \notin L$ for all $1 \leq i \leq t(s)$
 \item[(iii)] $k$ is upper bounded by a polynomial ${\sf P}_1$ of $s+\log (t(s))$. That is, 
 $k\leq {\sf P}_1(s+\log (t(s)))$. 
\end{itemize}
By our assumption $\Pi$ has a polynomial sized $\alpha$-approximate 
compression. 
That is, there 
is a pair of polynomial time algorithms $\BB$ and $\CC$, and 
an 
optimization problem $\Pi'$  with the following properties: 
(a) $\BB$ is a reduction algorithm, which takes input $(y,k)$  of $\Pi$ and outputs an instance 
$(y',k')$ 
of $\Pi'$ such that 
 $|y'| + k' \leq {\sf P}_2(k)$ 
for a polynomial ${\sf P}_2$ and (b) $\CC$ is a solution lifting algorithm, which given an instance $(y,k)$ of $\Pi$, an instance $(y',k')$of $\Pi'$ and a solution $S'$ to  $(y',k')$, 
outputs a  solution $S$ of $(y,k)$ such that 
$$\frac{\Pi(y,k,S)}{OPT_{\Pi}(y,k)}\cdot \alpha \geq \frac{\Pi'(y',k',S')}{OPT_{\Pi'}(y',k')}.$$ 
Let $t = {\sf P}_1\circ {\sf P}_2$, that is,  $t(s)={\sf P}_2({\sf P}_1(s))$. We design an oracle communication protocol for the language $L$ using algorithms $\AAA,\BB$ and $\CC$. The oracle communication protocol for $L$ works as follows. 
\begin{description}
\setlength\itemsep{-.7mm}
\item[Step 1:] The first player runs the algorithm $\AAA$ on the $t(s)$ input strings $x_1,\ldots,x_{t(s)}$, each of length at most $s$, and produces in polynomial time, an instance $(y,k)$ of $\Pi$ and a number $r\in {\mathbb R}$. Here the value $k$ is upper bounded by ${\sf P}_1(s+\log (t(s)))$ (by condition (iii)). 

\item[Step 2:]
The first player runs the reduction algorithm $\BB$ on $(y,k)$, producing an instance  $(y',k')$ 
of $\Pi'$. Then the first player sends the instance 
$(y',k')$ 
to the second player.  By the property of algorithm $\BB$, the size of the instance 
$(y',k')$ 
 is upper bounded by ${\sf P}_2(k) = {\sf P}_2({\sf P}_1(s+\log (t(s)) ) )$, which in turn is equal to $t(s+\log (t(s)))$. 

\item[Step 3:]
The (computationally unbounded) second player sends an optimum solution $S'$ of 
$(y',k')$ 
 back to the first player. 

\item[Step 4:]
The first player runs the solution lifting algorithm $\CC$ on input 
$(y,k),(y',k')$
 and $S'$, and it outputs a solution $S$ of $(y,k)$. Then, if $\Pi(y,k,S)\geq \frac{r}{\alpha}$ the first player declares that there exists an $i$ such that $x_i\in L$. Otherwise the first player declares that $x_i\notin L$ for all $i$.
\end{description}
%
%
%

\smallskip
All the actions of the first player are performed in polynomial time. The cost of communication is  $t(s+\log (t(s))) = \OO(t(s))$, since $t$ is a polynomial. We now show that the protocol is correct. Let $x_i\in L$ for some $1\leq i\leq t(s)$. Since $\AAA$ is an $\alpha$-gap cross composition we have that $OPT(y,k)\geq r$ (by condition (i)). 
Since $S'$ is an optimum solution, by the property of solution lifting  algorithm $\CC$,  $S$ is a solution of $(y,k)$ 
such that  $\Pi(y,k,S)\geq \frac{OPT(y,k)}{\alpha}\geq \frac{r}{\alpha}$. This implies that in Step 4, the first player declares that $x_i\in L$ for some $i$. 
Suppose now that $x_i\notin L$ for all $i$. Then, by the definition of $\alpha$-gap cross composition algorithms $\AAA$, we have that $OPT(y,k)< \frac{r}{\alpha}$. This implies that for any $S$, $\Pi(y,k,S)<\frac{r}{\alpha}$. Thus in \
Step 4, the first player declares that $x_i\notin L$ for all $i$. We have just verified that the described oracle communication protocol satisfies all the conditions of Lemma~\ref{lemma:complementary}. Thus, by Lemma~\ref{lemma:complementary}, we have that $L \in$ \coNPbypoly. This completes the proof.
\end{proof}

The main theorem of the section follows from Lemma~\ref{lem:gapcrossComposition}.

\begin{theorem}
\label{thm:gapcrossComposition}
Let $L$ be an \NPp-hard 
language and $\Pi$ be a nice parameterized optimization problem. If $L$ $\alpha$-gap cross 
composes to $\Pi$, and $\Pi$ has a polynomial sized $\alpha$-approximate 
compression, 
then 
\NPp\ $\subseteq$ \coNPbypoly. 
\end{theorem}

We note that Lemma~\ref{lemma:complementary} applies even if the first player works in co-nondeterministic polynomial time. Thus, a co-nondeterministic $\alpha$-gap cross composition together with an $\alpha$-approximate compression from an \NP-hard 
language would still yield that \NP{} $\subseteq$ \coNPbypoly. For clarity we formally define co-nondeterministic $\alpha$-gap cross composition for 
minimization problem which we later use  in this section to derive a lower bound on  \Setcover.
\begin{definition}[co-nondeterministic $\alpha$-gap cross composition for minimization problem] \label{def:conongapcrossCompositionmin}
Let~$L \subseteq \Sigma^*$ be a language, where $\Sigma$ is a finite alphabet and let~$\Pi$ be a parameterized minimization problem. We say 
that~$L$ \emph{co-nondeterministically $\alpha$-gap cross composes} into~$\Pi$ (where $\alpha\geq 1$), if there is a polynomial equivalence relation~$R$ and a nondeterministic  algorithm ${\cal A}$ which, 
given~$t$ strings~$x_1, x_2, \ldots, x_t$ belonging to the same equivalence class of~$R$, computes an 
instance~$(y,k)$ of $\Pi$ and $r\in {\mathbb R}$, in time polynomial in~$\sum _{i=1}^t |x_i|$ such that the following holds. 
\begin{itemize}
\setlength{\itemsep}{-2pt}
\item[(i)]  if $x_i \in L$ for some $i \in [t]$, then in all the computation paths of ${\cal A}$, $OPT(y, k) \leq {r}$, 
\item[(ii)]  if $ x_i \notin L$ for all $i \in [t]$, then there is a computation path in ${\cal A}$ with $OPT(y, k) > {r}\cdot {\alpha} $, and
\item[(iii)] $k$ is bounded by a polynomial in $\log t+\max_{\substack{1\leq i\leq t}} |x_i|$.
\end{itemize}
If such an algorithm exists, then we say $L$ co-nondeterministically $\alpha$-gap cross composes to $\Pi$. 
\end{definition}

\section{Longest Path}\label{sec:lp}
In this Section we show that {\sc Longest Path} does not admit an $\alpha$-approximate compression 
of polynomial size for any $\alpha\geq1$ unless \NP{} $\subseteq$ \coNPbypoly. The parameterized optimization version of   the {\sc Longest Path} problem, that we call {\sc Path}, is defined as follows. 
\[
    \mbox{\sc Path}(G,k,P)= 
\begin{cases}
    -\infty & \text{if $P$ is not a path in $G$} \\
    \min\left\{k+1,|V(P)|-1\right\} & \text{otherwise}
\end{cases}
\]

We show that {\sc Path} does not have a polynomial sized $\alpha$-approximate 
compression for any constant $\alpha\geq 1$. 
We prove this by giving an $\alpha$-gap cross composition from a \gappath. The problem \gappath\  is a promise problem which is defined as follows. 

\begin{definition}
The \gappath\ problem is to determine, given a graph $G$ and an integer $k$ whether: 
\begin{itemize}
\setlength{\itemsep}{-2pt}
\item $G$ has a path of length at least $k$, in which case we say that $(G,k)$ is a \Yes{} instance of \gappath.
\item the longest path in $G$ has length strictly less than $\frac{k}{\alpha}$, in which case we say that 
$(G,k)$ is a \No{} instance of \gappath.
\end{itemize}
\end{definition} 
It is known that \gappath\ is \NP{}-hard~\cite{KargerMR97}. 
\begin{lemma}\label{lemm:path_composition}
\gappath\ $\alpha$-gap cross composes to {\sc Path} for any $\alpha\geq 1$.
\end{lemma}
\begin{proof}
First we make the following polynomial equivalence relation: two instances $(G_1,k_1)$ and $(G_2,k_2)$ are in the same equivalence class if $k_1=k_2$. 
Now given $t$ instances $(G_1, k),\ldots,(G_t, k)$ of \gappath, the $\alpha$-gap cross composition algorithm $\AAA$ just outputs an instance  $(G,k)$ of {\sc Path}, where $G$ is the disjoint union of $G_1,\ldots,G_t$. 

Clearly, $G$ contains a path of length $k$ if and only if there exists an $i$ such that $G_i$ contains a path of length $k$. Thus, $OPT(G,k) \geq r$ if and only if there is an $i$ such that $(G_i, k)$ is a yes instance of \gappath. For the same reason $OPT(G,k) < \frac{r}{\alpha}$ if and only if $(G_i, k)$ is a \No{} instance of \gappath\ for every $i$. Finally the parameter $k$ of the output instance is upper bounded by the size of the graphs $G_i$. This concludes the proof.
%
%
\end{proof}

Theorem~\ref{thm:gapcrossComposition} and Lemma~\ref{lemm:path_composition} yields the following theorem. 
\begin{theorem}
{\sc Path} does not have a polynomial size $\alpha$-approximate 
compression 
 for any $\alpha \geq 1$, unless \NPp\ $\subseteq$ \coNPbypoly. 
\end{theorem}

%% file: setcover.tex
\newcommand{\Setcoverparam}{{\sc Set Cover$/n$}}
\newcommand{\SCparam}{{\sc SC$/n$}}
\newcommand{\gapsc}{$\alpha$-{\sc Gap $d$-Set Cover}}
\section{Set Cover}\label{sec:setCover}
In this section we show that parameterized optimization version of \Setcover{} parameterized by universe size does not admit an $\alpha$-approximate 
compression 
of polynomial size for any $\alpha\geq1$ unless \NP{} $\subseteq$ \coNPbypoly. 
The input of \Setcover{} is a family ${\cal S}$ of subsets of a universe $U$ and the objective is to choose a minimum 
sized subfamily ${\cal F}$ of ${\cal S}$ such that $\bigcup_{S\in {\cal F}}S=U$.  Such a set $\FF$ 
is called a {\em set cover} of $(\SSS,U)$.  
Since 
the parameter used here is a structural parameter, both \Setcover\ (\SC) and its parameterized version 
\Setcoverparam\ (\SCparam) can be defined as follows.
\[
    \mbox{\SCparam}((\SSS,U),\vert U \vert,\FF)= \mbox{\SC}((\SSS,U),\FF)=
\begin{cases}
     \vert \FF\vert & \text{if $\FF$ is a set cover of $(\SSS,U)$} \\
    \infty & \text{otherwise}
\end{cases}
\]

We show \Setcoverparam{} does not have a polynomial sized $\alpha$-approximate compression for any constant $\alpha\geq 1$. Towards this we first define the  $d$-\Setcover\ problem, $d\in {\mathbb N}$.  The $d$-\Setcover\ problem is a restriction of \Setcover, where each set in the family $\SSS$ is bounded by $d$. 
We show the desired lower bound on  \Setcoverparam{} by giving a co-nondeterministic $\alpha$-gap cross composition from a gap version of  $d$-\Setcover. 
The problem \gapsc\  is a promise problem defined as follows. 
\begin{definition}
The \gapsc\ problem is to determine, given a set family $\SSS$ over a universe $U$, where  the size of each set in $\SSS$ is upper bounded by $d$, and an integer $r$ whether: 
\begin{itemize}
\setlength{\itemsep}{-2pt}
\item $OPT_{\SC}(\SSS,U)\leq r$, in which case we say that $((\SSS,U),r)$ is a \Yes{} instance of \gapsc.
\item $OPT_{\mbox{\SC}}(\SSS,U)> r\alpha$, in which case we say that 
$((\SSS,U),r)$ is a \No{} instance of \gapsc.
\end{itemize}
\end{definition} 

We will use the following known result regarding \gapsc\ for our purpose. 

\begin{theorem}[\cite{Trevisan01,ChlebikC08}]
\label{thm:gapSChard}
For any $\alpha\geq 1$, there is a constant $d$ such that 
\gapsc\ is \NPp-hard. 
\end{theorem}

To show a  lower bound  of $\alpha$-approximate compression for \Setcoverparam, 
our aim here is to give co-nondeterministic $\alpha$-gap cross composition from 
\gapsc.
To give a co-nondeterministic cross composition algorithm it is enough to give a randomized cross
composition algorithm which is always correct when it returns a \No{} instance. We give a formal 
proof about this after the following lemma.  

\begin{lemma}
\label{lem:rand:SC}
Given $t$ instances of \gapsc, $((\SSS_1,U_1),r), \ldots, ((\SSS_t,U_t),r)$ of size $s$ each, 
$\vert U_1\vert=\cdots =\vert U_t\vert=n$, and $\vert \SSS_1 \vert=\cdots =\vert \SSS_t\vert=m$, 
there is a randomized polynomial 
time algorithm (i.e, polynomial in $t\cdot s$) with one sided error, which outputs an instance $(\SSS,U)$ of \Setcover{} with following guarantees. 
 \begin{itemize}
 \setlength{\itemsep}{-2pt}
\item[$(a)$]  if 
$((\SSS_i,U_i),r)$ is an \Yes{} instance of \gapsc\   
for some $i \in [t]$, then 
$$\Pr[OPT_{\mbox{\SC}}(\SSS,U)\leq {r}]=1,$$
\item[$(b)$]  if 
$((\SSS_i,U_i),r)$ is a \No{} instance of \gapsc\ 
for all $i \in [t]$, then 
$$\Pr[OPT_{\mbox{\SC}}(\SSS,U )> {r\alpha}]>0, \mbox{and}$$
\item[$(c)$] $\vert U \vert =n^{2d}\cdot 4^d \cdot 2 \log \binom{m\cdot t}{r\cdot \alpha}$.
\end{itemize}
\end{lemma}
\begin{proof}
We design an algorithm ${\cal A}$ with properties mentioned in the statement of the lemma. 
We know that $n=\vert U_1\vert =\ldots =\vert U_t \vert$.  
For any $i\in [t]$, let $\SSS_i=\{\SSS_{i1},\ldots,\SSS_{im}\}$. 
Algorithm ${\cal A}$ creates a universe 
$U$ with $N=n^{2d}\cdot 4^d \cdot 2 \log \binom{m\cdot t}{r\cdot \alpha}$ elements. Now we describe the random process by which we construct the set family $\SSS$. 
\begin{quote}
For each $u\in U$ and $i\in [t]$, {\em uniformly at random} assign an element 
from $U_i$ to $u$. 
\end{quote}
That is, in this random process $t$ elements are assigned to each element 
$u\in U$, one from each $U_i$. We use $\Gamma_i : U \rightarrow U_i$ to represent the random assignment. That is, for each $u\in U$, $i\in [t]$, 
$\Gamma_i(u)$ denotes the element in $U_i$ that is assigned to $u$. Observe that an element 
$w \in U_i$  can be assigned to several elements of $U$. In other words,  the set  $
\Gamma_i^{-1}(w)$ can have arbitrary size.
For each $S_{ij}, i\in [t], j\in [m]$, algorithm ${\cal A}$, creates a set 
$$S'_{ij}=\bigcup_{w\in S_{ij}} \Gamma_i^{-1}(w).$$ 
Notice that $\vert S'_{ij}\vert$ need not be 
bounded by any function of $d$. Let $\SSS=\{S'_{ij}  : i\in [t],j\in[m]\}$. Algorithm ${\cal A}$ outputs $(\SSS,U)$. 
An illustration is given in Figure~\ref{fig:composition}.

\input{figure_composition}

Now we prove the correctness of the algorithm. Suppose there exists  $i\in [t]$ such that $((\SSS_i,U_i),r)$ is a \Yes{} instance of \gapsc. That is, there exist $S_{ij_1},\ldots,S_{ij_r}\in \SSS_i$ such that $S_{ij_1}\cup \cdots\cup S_{ij_r}=U_i$. 
We know that for each $u\in U$, there is a $w\in U_i$ such that $\Gamma_i(u)=w$.  Since 
$S_{ij_1}\cup \cdots\cup S_{ij_r}=U_i$ and for each $u\in U$, there is a $w\in U_i$ with $\Gamma_i(u)=w$, we can conclude that 
$$S'_{ij_1}\cup \cdots\cup S'_{ij_r}=\bigcup_{\ell\in[r],w\in S_{ij_\ell}} \Gamma_i^{-1}(w)=\bigcup_{w\in U_i} \Gamma_i^{-1}(w)= U.$$
This implies that $\{S'_{ij_1}, \ldots, S'_{ij_r}\}$ is a set cover of $(\SSS,U)$. This proves condition $(a)$ of the lemma. 

Now we prove condition $(b)$. In this case, for all $i \in [t]$, 
$((\SSS_i,U_i),r)$ is a \No{} instance of \gapsc. That is, for any $i\in [t]$, $(\SSS_i,U_i)$ does not 
have a set cover of cardinality at most $r \alpha$. Let $\SSS_{\sf input}=\bigcup_{i\in [t]} \SSS_i$. 
Each set in $S_{ij}' \in \SSS$, $i\in[t],j\in[m]$ is a random variable defined by $\bigcup_{w\in S_{ij}} \Gamma_i^{-1}(w).$ We call $S_{ij}'$ a {\em set-random variable}. Thus, $\SSS$ is a set of $mt$-set random variables where the domain of each set-random variable is the power set of $U$ (that is, $2^{U}$). 
Thus, to show that $\Pr[OPT_{\mbox{\SC}}(\SSS,U )> {r\alpha}]>0$, we need to show that probability of 
union of any $\alpha r$ set-random variables covering $U$ is strictly less than 
$\frac{1}{\binom{m t}{r\alpha}}$. For an ease of presentation, for any $S_{ij}\in \SSS_{\sf input}$, we call $\bigcup_{w\in S_{ij}} \Gamma_i^{-1}(w)$, as ${\sf Image}(S_{ij})$. Observe that ${\sf Image}(S_{ij})$ is also a random variable and it is same as $S_{ij}'$. For a subset $\FF \subseteq \SSS_{\sf input}$, by ${\sf Image}(\FF)$ we mean $\{{\sf Image}(S)~|~S\in \FF\} $. 
Now we are ready to state and prove our main claim. 

\begin{claim} 
\label{claim:PrFsetcover}
For any $\FF\subseteq \SSS_{\sf input}$ of cardinality $r \alpha$, $\Pr[{\sf Image}(\FF) \mbox{ is a set cover of }(\SSS,U)]<\frac{1}{\binom{m t}{r \alpha}}$.   
\end{claim}
\begin{proof}
We can partition $\FF = \FF_1\uplus \ldots \uplus \FF_t$ such that 
$\FF_i=\FF \cap \{S_{i1},\ldots , S_{im}\}$. Similarly, we can partition ${\sf Image}(\FF)$ into ${\sf Image}(\FF_1)\uplus \ldots \uplus {\sf Image}(\FF_t)$ such that ${\sf Image}(\FF_i)={\sf Image}(\FF) \cap \{S'_{i1},\ldots , S'_{im}\}$. 
Le   
$U_i'$ be the subset of elements of $U_i$ covered by the sets in $\FF_i$. 
Let $X_i=\vert U'_i\vert$.   
Because of our assumption  that $((\SSS_i,U_i),r)$ is a \No{} 
instance of \gapsc, we have that 
the subset $U'_i$ covered by $\FF_i$ is a strict subset of $U_i$ and hence 
\begin{equation}
\mbox{for all } i\in [t], X_i<n \label{eqn:SCXi} 
\end{equation}
 Since  
$\vert \bigcup_{i\in[t]}\FF_i\vert=\vert \FF \vert= r \alpha$ and the cardinality of each set in 
$\FF$ is at most $d$, we have  
\begin{equation}
\sum_{i\in [t]} X_i<r \alpha  d < nd \label{eqn:SCXiSUM} 
\end{equation} 
Since $((\SSS_i,U_i),r)$ is a \No{} instance of \gapsc\ for any $i\in [t]$, 
we have $r \alpha <n$ and hence the last inequality of Equation~\ref{eqn:SCXiSUM} follows. 
 %
Towards bounding the probability mentioned in the claim, we first lower bound the following probability, 
$\Pr[u \mbox{ is not covered by }{\sf Image}(\FF)]$, for any fixed $u\in U$. 
\begin{eqnarray}
\Pr[u \mbox{ is not covered by }{\sf Image}(\FF)] &=& \bigwedge_{i\in [t]} \Pr[u \mbox{ is not covered by }{\sf Image}(\FF_i)] \nonumber\\
&=&\bigwedge_{i\in [t]} \left(1-\Pr[u \mbox{ is covered by }{\sf Image}(\FF_i)]\right) \nonumber\\ 
&=&\bigwedge_{i\in [t]} \left(1-\Pr[\Gamma_i(u)\in U'_i]\right) \nonumber\\ 
&=&\prod_{i\in [t]} \left(1-\frac{\vert U_i'\vert}{n}\right)  ~~(\mbox{Because, }\forall w\in U_i, \Pr[\Gamma_i(u)=w]=\frac{1}{n})\nonumber\\
&=&\prod_{i\in [t]} \left(1-\frac{X_i}{n}\right)\nonumber\\
&=&\prod_{\substack{i\in[t]\; \mbox{ such that } \\ X_i\leq \frac{n}{2}}} \left(1-\frac{X_i}{n}\right) \cdot \prod_{\substack{i\in[t]\; \mbox{ such that } \\ X_i > \frac{n}{2}}} \left(1-\frac{X_i}{n}\right)\label{eqn:SCPr1}
\end{eqnarray}
We know, by Equation~\ref{eqn:SCXiSUM}, that  $\sum_{i\in [t]} X_i< n  d$. This implies that 
the number of $X_i$'s  such that $X_i> \frac{n}{2}$ is at most 
$2d$.
By Equation~\ref{eqn:SCXi}, we have that for any 
$i\in [t]$,  
$\left(1-\frac{X_i}{n}\right)\geq \left(1-\frac{n-1}{n}\right)=\frac{1}{n}$.  
Since the number of $X_i$'s with  $X_i> \frac{n}{2}$ is at most $2d$ and $\left(1-\frac{X_i}{n}\right)\geq \frac{1}{n}$, we can rewrite Equation~\ref{eqn:SCPr1}, as follows. 
\begin{eqnarray}
\Pr[u \mbox{ is not covered by }{\sf Image}(\FF)]&\geq &\left(\prod_{\substack{i\in [t]\; \mbox{ such that }  \\X_i\leq \frac{n}{2}}} \left(1-\frac{X_i}{n}\right)\right) \cdot \frac{1}{n^{2d}}\nonumber\\
&\geq& \frac{1}{n^{2d}} \prod_{\substack{i\in [t]\; \mbox{ such that }  \\X_i\leq \frac{n}{2}}} \left(\frac{1}{4}\right)^{\frac{X_i}{n}} \qquad\qquad\quad(\mbox{By Fact~\ref{fact2}})\nonumber\\
&\geq&  \frac{1}{n^{2d}} \cdot \left(\frac{1}{4}\right)^{\frac{\sum X_i}{n}}\nonumber\\
&\geq&  \frac{1}{n^{2d}} \cdot \left(\frac{1}{4}\right)^{d}\qquad\qquad\qquad\qquad(\mbox{By Equation~\ref{eqn:SCXiSUM}}) \label{eqn:u_not_covered}
\end{eqnarray}
Since the set of events ``$u$ is covered by ${\sf Image}(\FF)$'', where $u\in U$, are independent events, we have
\begin{eqnarray}
\Pr[{\sf Image}(\FF) \mbox{ is a set cover of }(\SSS,U)]&\leq & \prod_{u\in U}\Pr[\mbox{$u$ is covered by }{\sf Image}(\FF)]\nonumber\\
&\leq & \prod_{u\in U} \left( 1-\frac{1}{n^{2d}\cdot 4^d}\right) \nonumber\qquad\qquad\qquad (\mbox{By Equation~\ref{eqn:u_not_covered}})\\
&\leq&\prod_{u\in U} e^ \frac{-1}{n^{2d}\cdot 4^d} \nonumber \\
&<&\frac{1}{\binom{m t}{r\alpha}} \qquad\qquad (\mbox{Because }\vert U \vert=n^{2d}\cdot 4^d \cdot 2 \log \binom{m t}{r\alpha})\nonumber
\end{eqnarray}
This completes the proof of the claim.
\end{proof}
Since the number of subsets of cardinality $r \alpha$ of $\SSS_{\sf input}$, is at most $\binom{m t}{r \alpha}$, by 
Claim~\ref{claim:PrFsetcover} and union bound we get that   
\[\Pr[\mbox{$\exists$ a set $\FF\subseteq \SSS_{\sf input}$ of cardinality $r \alpha$ such that } 
{\sf Image}(\FF) \mbox{ is a set cover of }(\SSS,U)]<1.\] 
This completes the proof of condition $(b)$. Condition $(c)$ trivially follows from the construction of $U$. 
\end{proof}

Now we will use the construction given in Lemma~\ref{lem:rand:SC} to prove the main lemma of this section.

\begin{lemma}\label{lemm:sc_composition}
\gapsc\  co-nondeterministically $\alpha$-gap cross composes to \Setcoverparam{} for any $\alpha\geq 1$.
\end{lemma}
\begin{proof}
First we make the following polynomial equivalence relation: two instances $((\SSS,U),r)$ and $((\SSS',U'),r')$ of  \gapsc\  are in the same equivalence class if $\vert \SSS\vert=\vert\SSS'\vert, \vert U \vert=\vert U'\vert$ and 
$r=r'$. 
Towards proving the lemma we need to design an algorithm ${\cal B}$ with the properties of Definition~\ref{def:conongapcrossCompositionmin}.   
Let  $((\SSS_1,U_1),r), \ldots, ((\SSS_t,U_t),r)$ be instances of \gapsc\ of size $s$ each, 
$\vert U_1\vert=\cdots =\vert U_t\vert=n$, and $\vert \SSS_1 \vert=\cdots =\vert \SSS_t\vert=m$.  
 Here, $t$ is polynomially bounded  in $s$.  Since $((\SSS_i,U_i),r)$, $i\in [t]$, are instances 
of  \gapsc,   $m\leq \binom{n}{d}$ and hence $t\leq n^c$ for some constant $c$.  
Now ${\cal B}$ runs the algorithm ${\cal A}$ mentioned in the Lemma~\ref{lem:rand:SC}, but instead of using the 
random bits it nondeterministically guesses these bits while running ${\cal A}$. 
Algorithm ${\cal B}$ returns $(\SSS,U)$ and $r$ as output, where $(\SSS,U )$ is the output of ${\cal A}$.  

If there exists $i\in [t]$ such that $((\SSS_i,U_i),r)$ 
is a \Yes{} instance of \gapsc, by condition $(a)$ of Lemma~\ref{lem:rand:SC}, we can 
conclude that $OPT_{\mbox{\SCparam}}((\SSS,U),\vert U \vert)\leq r$ and hence satisfies property $(i)$ of 
Definition~\ref{def:conongapcrossCompositionmin}. Suppose  $((\SSS_i,U_i),r)$ is a \No{} instance for all 
$i\in [t]$. Because of condition $(b)$ of Lemma~\ref{lem:rand:SC}, there is a choice of random bits $B$ 
such that if ${\cal A}$ uses $B$ as the random bits then  
$OPT_{\mbox{\SC}}(\SSS, U)> r \alpha$. 
Hence, for the nondeterministic guess $B$ of the algorithm ${\cal B}$, we get that 
$OPT_{\mbox{\SCparam}}((\SSS,U),\vert U \vert)=OPT_{\mbox{\SC}}(\SSS, U)> r \alpha$. This proves property $(ii)$  of 
Definition~\ref{def:conongapcrossCompositionmin}.
By condition $(c)$ of Lemma~\ref{lem:rand:SC}, and the facts that $m,t \in n^{\OO(1)}$, we get that $\vert U \vert = n^{\OO(1)}$. This 
implies the property $(iii)$ of   Definition~\ref{def:conongapcrossCompositionmin}. This 
completes the proof of the lemma. 
\end{proof}

Theorems~\ref{thm:gapcrossComposition} and \ref{thm:gapSChard}, and Lemma~\ref{lemm:sc_composition} yields the following theorem. 
\begin{theorem}
\Setcoverparam{} does not have a polynomial size $\alpha$-approximate 
compression 
 for any $\alpha \geq 1$, unless \NPp $\subseteq$ \coNPbypoly. 
\end{theorem}


%% file: figure_composition.tex
 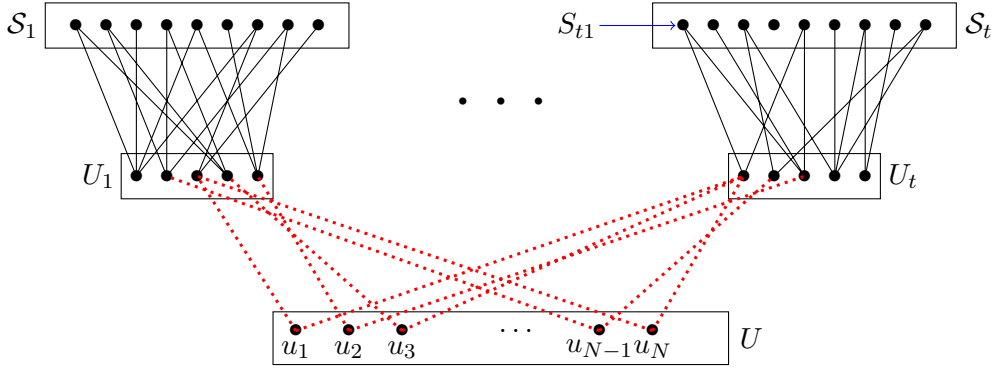
\begin{figure}
 \centering

\begin{tikzpicture}[scale=1]

\draw [] (-1,2) rectangle (3,2.6);
\draw [] (0,0) rectangle (2,0.6);

\node[] (a) at (-0.3,0.3) {$U_1$};
\node[] (a) at (0.2,0.3) {$\bullet$};
\node[] (a) at (0.6,0.3) {$\bullet$};
\node[] (a) at (1,0.3) {$\bullet$};
\node[] (a) at (1.4,0.3) {$\bullet$};
\node[] (a) at (1.8,0.3) {$\bullet$};

\node[] (a) at (-1.3,2.3) {$\SSS_1$};
\node[] (a) at (-0.6,2.3) {$\bullet$};
\node[] (a) at (-0.2,2.3) {$\bullet$};
\node[] (a) at (0.2,2.3) {$\bullet$};
\node[] (a) at (0.6,2.3) {$\bullet$};
\node[] (a) at (1,2.3) {$\bullet$};
\node[] (a) at (1.4,2.3) {$\bullet$};
\node[] (a) at (1.8,2.3) {$\bullet$};
\node[] (a) at (2.2,2.3) {$\bullet$};
\node[] (a) at (2.6,2.3) {$\bullet$};

\draw 
(-0.6,2.3) --(0.2,0.3)
(-0.2,2.3) --(1.4,0.3)
(0.2,2.3) -- (0.2,0.3)
(0.6,2.3) --(0.6,0.3)
(1,2.3) --(1.8,0.3)
(1.4,2.3)-- (1.8,0.3)
(1.8,2.3) --(1,0.3)
(2.2,2.3) -- (0.6,0.3)
(-0.6,2.3) --(1.4,0.3)
(-0.2,2.3) --(0.6,0.3)
(0.6,2.3) -- (1.4,0.3)
(1,2.3) --(0.2,0.3)
(1.8,2.3)--(0.2,0.3) 
(2.2,2.3) --(1.8,0.3)
(2.6,2.3) --(1,0.3)
;

\node[] (a) at (4.5,1.3) {\tiny$\bullet$};
\node[] (a) at (5.5,1.3) {\tiny$\bullet$};
\node[] (a) at (5,1.3) {\tiny$\bullet$};


\draw [] (7,2) rectangle (11,2.6);
\draw [] (8,0) rectangle (10,0.6);

\node[] (a) at (8.2,0.3) {$\bullet$};
\node[] (a) at (8.6,0.3) {$\bullet$};
\node[] (a) at (9,0.3) {$\bullet$};
\node[] (a) at (9.4,0.3) {$\bullet$};
\node[] (a) at (9.8,0.3) {$\bullet$};
\node[] (a) at (10.3,0.3) {$U_t$};

\node[] (a) at (6,2.3) {$S_{t1}$};

\node[] (a) at (7.4,2.3) {$\bullet$};
\node[] (a) at (7.8,2.3) {$\bullet$};
\node[] (a) at (8.2,2.3) {$\bullet$};
\node[] (a) at (8.6,2.3) {$\bullet$};
\node[] (a) at (9,2.3) {$\bullet$};
\node[] (a) at (9.4,2.3) {$\bullet$};
\node[] (a) at (9.8,2.3) {$\bullet$};
\node[] (a) at (10.2,2.3) {$\bullet$};
\node[] (a) at (10.6,2.3) {$\bullet$};
\node[] (a) at (11.3,2.3) {$\SSS_t$};

\draw[blue,->] (6.3,2.3)--(7.3,2.3);

\draw
(7.4,2.3) --(8.2,0.3)
(7.8,2.3) --(9,0.3)
(8.2,2.3) --(8.6,0.3)
(9,2.3) --(9,0.3)
(9.4,2.3)-- (9.4,0.3)
(9.8,2.3) --(9.8,0.3)
(10.2,2.3)--(9.8,0.3)
(10.6,2.3)--(9.4,0.3)
%
(7.4,2.3) --(9,0.3)
(8.2,2.3) --(9.4,0.3)
(9,2.3) --(8.2,0.3)
(9.8,2.3) --(9.4,0.3)
(10.6,2.3)--(8.6,0.3)
;

\draw [] (2,-2.2) rectangle (8,-1.5);
\node[] (a) at (2.3,-1.75) {$\bullet$};
\node[] (a) at (3,-1.75) {$\bullet$};
\node[] (a) at (3.7,-1.75) {$\bullet$};
\node[] (a) at (5.2,-1.75) {$\ldots$};
\node[] (a) at (6.3,-1.75) {$\bullet$};
\node[] (a) at (7,-1.75) {$\bullet$};
\node[] (a) at (8.3,-1.85) {$U$};

\node[] (a) at (2.3,-2) {$u_1$};
\node[] (a) at (3,-2) {$u_2$};
\node[] (a) at (3.7,-2) {$u_3$};
\node[] (a) at (6.3,-2) {$u_{N-1}$};
\node[] (a) at (7,-2) {$u_{N}$};

\draw[red,line width=0.4mm,dotted]
(2.3,-1.75)-- (8.2,0.3)
(3,-1.75) --(9,0.3)
(3.7,-1.75)-- (8.2,0.3)
(6.3,-1.75) --(8.6,0.3)
(7,-1.75)--(8.2,0.3)
(2.3,-1.75)-- (1,0.3) 
(3,-1.75) --(1.8,0.3) 
(3.7,-1.75)-- (1.4,0.3) 
(6.3,-1.75) --(0.6,0.3) 
(7,-1.75)--(1,0.3) 
;

\end{tikzpicture}
\caption{An illustration of proof of Lemma~\ref{lem:rand:SC}. Every set and every element is represented using bullets. 
An element in a set is represented using an edge between the element and the set.  
The random assignment to each element in $U$ is represented using dotted lines. 
The set $S'_{t1}$ created from $S_{t1}$ is $\{u_1,u_2,u_3,u_N\}$.}
\label{fig:composition}
\end{figure}

%% file: hittingset.tex
\newcommand{\Hittingsetparam}{\Hittingset$/n$}
\newcommand{\HSparam}{\HS$/n$}

\section{Hitting Set}\label{sec:hitSet}
In this section we show that a parameterized optimization version of \Hittingset{} does not admit an $\OO( 2^{\log^c n})$-approximate kernel 
of polynomial size for any $c<1$, unless CNF-SAT can be solved in slightly subexponential time, where  universe size $n$ 
of the input instance is the parameter.
Compare to \Setcover\ our result in this section are much more stronger, but unlike \Setcover\ here we can only rule out an existence of an approximate kernel and not an approximate compression.    
%
The input of \Hittingset{} is a family ${\cal S}$ of subsets of a universe $U$ and the objective is to choose a minimum 
cardinality subset $X\subseteq U$ such that for all $S\in \SSS$, $S\cap X\neq \emptyset$.  Such a subset $X$ 
is called a {\em hitting set} of $(\SSS,U)$. Since 
the parameter used here is a structural parameter, both \Hittingset~(\HS) and its parameterized version 
\Hittingsetparam~(\HSparam) can be defined as follows. 
\[
    \mbox{\HSparam}((\SSS,U),\vert U \vert,X)=\mbox{\HS}((\SSS,U),X) =
\begin{cases}
     \vert X\vert & \text{if $X$ is a hitting set of  $(\SSS,U)$} \\
    \infty & \text{otherwise}
\end{cases}
\]

The following lemma shows that in fact  \Hittingset{} is same as \Setcover\ but with  a different parameter.  

\begin{lemma}
\label{lem:SCeqHS}
Let $(\SSS,U)$ be an instance of \Hittingset. Let  $F_u=\{S\in \SSS : u\in S\}$ for all $u\in U$ and let $\FF=\{F_u : u\in U\}$. 
Then $OPT_{\mbox{\HS}}(\SSS,U)=OPT_{\mbox{\SC}}(\FF,\SSS)$
\end{lemma}
\begin{proof}
Let $X\subseteq U$ be  a hitting set of $(\SSS,U)$. 
Consider the set $\FF_X=\{F_u\in \FF : u\in X\}$. Since $X$ is a hitting set of $\SSS$, for 
any $S\in \SSS$, there is an element $u\in X$ such that $S\in F_u$. This implies that $\FF_X$ 
is a set cover of $(\FF,\SSS)$. 

Let $\FF'\subseteq \FF$ be a set cover of $(\FF,\SSS)$. Let $X=\{u\in U : F_u\in \FF' \}$.
Since $\FF'$ is a set cover of $(\FF,\SSS)$, for any $S\in \SSS$, there is a set $F_u\in \FF'$ such that 
$S\in F_u$. This implies that $X$ is a hitting set of $(\SSS,U)$. 
This completes the proof of the lemma. 
\end{proof}

The following Lemma follows from the $\OO(\log n)$-approximation algorithm of 
\Setcover~\cite{Chvatal79} and Lemma~\ref{lem:SCeqHS}

\begin{lemma}[\cite{Chvatal79}]
\label{lem:HSapproxalgo}
There is a polynomial time algorithm which given an instance $(\SSS,U)$ of \Hittingset, outputs a hitting set 
of cardinality bounded by $\OO(OPT_{\mbox{\HS}}(\SSS,U)\cdot \log \vert \SSS\vert)$.  
\end{lemma}

The following theorem is a slight weakening of a result by Nelson~\cite{Nelson07}, 
which we use to prove our theorem. 

\begin{theorem}[\cite{Nelson07}]
\label{thm:hittingsethard}
For any $c<1$, 
\Hittingset{} has no polynomial time $\OO(2^{\log^c n})$-approximation unless 
CNF-SAT with $n$-variables can be solved in time $2^{ \OO(  2^{     \log^{1-1/(\log \log n)^{{1}/{3}}} n   }  )  }$. 
\end{theorem}

The assumption used in Theorem~\ref{thm:hittingsethard}, implies the Exponential Time Hypothesis (ETH)
of Impagliazzo, Paturi and Zane~\cite{ImpagliazzoPZ01} and hence it is 
weaker than ETH.  
\begin{theorem} 
For any $c<1$, 
\Hittingsetparam\ does not admits a $\OO(2^{\log^{c} n} )$-approximate kernel, unless  CNF-SAT with $n$-variables can be solved in time 
$2^{ \OO(  2^{     \log^{1-1/(\log \log n)^{{1}/{3}}} n   }  )  }$. 
\end{theorem}
\begin{proof}
Suppose there is a  
$\OO(2^{\log^{c} n} )$-approximate kernel ${\cal A}$ for \Hittingsetparam{}  for some  
$c<1$. Then, we argue that we can solve CNF-SAT on $n$ variables in time 
$2^{ \OO(  2^{     \log^{1-1/(\log \log n)^{{1}/{3}}} n   }  )  }$.
Towards that, by Theorem~\ref{thm:hittingsethard}, 
it is enough to give a 
$\OO(2^{\log^{c'} n} )$
-approximation algorithm for \Hittingset{} for some $c'<1$, where 
$n$ is the cardinality of the universe in the input instance. 

Fix a constant $c'$ such that $c<c'<1$. 
We design a 
$\OO(2^{\log^{c'} n} )$-approximation algorithm for \Hittingset{} using ${\cal A}$.  
Let $(\SSS,U)$ be an instance of \HS{} and let $\vert U \vert=n$. 
Let ${\cal R}_{\cal A}$ and ${\cal L}_{\cal A}$ be the reduction algorithm and solution lifting algorithm 
of ${\cal A}$ respectively. We run the algorithm ${\cal R}_{\cal A}$ on $((\SSS,U),n)$ and let $((\SSS',U'),\vert U'\vert)$ be the 
output of ${\cal R}_{\cal A}$. We know that $\vert \SSS'\vert + \vert U'\vert =n^{\OO(1)}$. Then, by Lemma~\ref{lem:HSapproxalgo}, 
we compute a hitting set $W$ of $(\SSS',U')$, of cardinality bounded by $\OO(OPT_{\mbox{\HS}}(\SSS',U') \cdot \log n)$. 
Then, by using solution lifting algorithm ${\cal L}_{\cal A}$, we compute a hitting set $X$ of $((\SSS,U),n)$. By the 
property of 
$\OO(2^{\log^{c} n} )$
-approximate kernel ${\cal A}$, we can conclude 
that the cardinality of $X$ is bounded by 
$\OO(2^{\log^{c} n} \cdot \log n \cdot OPT_{\mbox{\HSparam}}((\SSS,U), n))=\OO(2^{\log^{c'} n}  \cdot OPT_{\mbox{\HSparam}}((\SSS,U), n))$. 
This implies that $X$ is a 
$\OO(2^{\log^{c'} n} )$
-approximate solution of $(\SSS,U)$.  
This completes the proof of the theorem.  
\end{proof}

%% file: conclusion.tex
\section{Conclusion and Discussions}\label{sec:conclusion}
In this paper we have set up a framework for studying lossy kernelization, and showed that for several problems it is possible to obtain approximate kernels with better approximation ratio than that of the best possible approximation algorithms, and better size bound than what is achievable by regular kernels. We have also developed methods for showing lower bounds for approximate kernelization. There are plenty of problems that are waiting to be attacked within this new framework. Indeed, one can systematically go through the list of all parameterized problems, and investigate their approximate kernelization complexity. For problems that provably do not admit polynomial size kernels but do admit constant factor approximation algorithms, one should search for PSAKSes. For problems with PSAKSes one should search for efficient PSAKSes. For problems with no polynomial kernel and no constant factor approximation, one may look for a constant factor approximate kernel of polynomial size. For problems that do have polynomial kernels, one can search for approximate kernels that are even smaller. We conclude with a list of concrete interesting problems.

\begin{itemize}
\setlength\itemsep{-2pt}
\item Does \cvc{}, {\sc Disjoint Factors} or {\sc Disjoint Cycle Packing} admit an EPSAKS?
\item Does {\sc Edge Clique Cover} admit a constant factor approximate kernel of polynomial size?
\item Does {\sc Directed Feedback Vertex Set} admit a  constant factor approximate kernel of polynomial size?
\item Does {\sc Multiway Cut} or {\sc Subset Feedback Vertex Set} have a PSAKS?
\item Does {\sc Disjoint Hole Packing} admit a PSAKS? Here a {\em hole} in a graph $G$ is an induced cycle of length $4$ or more.
\item Does {\sc Optimal Linear Arrangement} parameterized by vertex cover admit a constant factor approximate kernel of polynomial size, or even a PSAKS? 
\item Does {\sc Maximum Disjoint Paths} admit a constant factor approximate kernel, or even a PSAKS? Here the input is a graph $G$ together with a set of vertex pairs $(s_1, t_1), (s_2, t_2)$, $\ldots$, $(s_\ell, t_\ell)$. The goal is to find a maximum size subset $R \subseteq \{1, \ldots, \ell\}$ and, for every $i \in R$ a path $P_i$ from $s_i$ to $t_i$, such that for every $i,j \in R$ with $i \neq j$ the paths $P_i$ and $P_j$ are vertex disjoint. What happens to this problem when input is restricted to be a planar graph? Or a graph excluding a fixed graph $H$ as a minor? What about chordal graphs, or interval graphs?

\item It is known that {\sc $d$-Hitting Set} admits a kernel if size $\OO(k^{d})$, this kernel is also a strict $1$-approximate kernel. {\sc $d$-Hitting Set} also admits a factor $d$-approximation in polynomial time, this is a $d$-approximate kernel of constant size. Can one interpolate between these two extremes by giving an $\alpha$-approximate kernel of size $\OO(k^{f(\alpha)})$ with $f(1) = d$, $f(d) = \OO(1)$, and $f$ being a continuous function?

\item Our lower bound for approximate kernelization of \Hittingset{} parameterized by universe size $n$ does not apply to compressions. Can one rule out polynomial size constant factor approximate compressions of  \Hittingset{} parameterized by universe size $n$ assuming \NP{} $\not\subseteq$ \coNPbypoly{} or another reasonable complexity theoretic assumption?

\item One may extend the notion of approximate kernels to approximate Turing kernels~\cite{CyganFKLMPPS15} in a natural way. Does {\sc Independent Set} parameterized by treewidth admit a polynomial size approximate Turing kernel with a constant approximation ratio? What about a Turing PSAKS?

\item Does {\sc Treewidth} admit an constant factor approximate kernel of polynomial size? Here even a Turing kernel (with a constant factor approximation) would be very interesting.

\item What is the complexity of approximate kernelization of {\sc Unique Label Cover}?~\cite{AroraBS15,Khot02a}

\item The notion of $\alpha$-gap cross compositions can be modified to ``AND $\alpha$-gap cross compositions'' in the same way that AND-compositions relate to OR-compositions~\cite{BodlaenderDFH09}. In order to directly use such ``AND $\alpha$-gap cross compositions'' to show lower bounds for approximate kernelization, one needs an analogue of Lemma~\ref{lemma:complementary} for the problem of deciding whether {\em all} of the $t(s)$ inputs belong to $L$. This is essentially a strengthening of the AND-distillation conjecture~\cite{BodlaenderDFH09,Drucker12} to oracle communication protocols (see the conclusion section of Drucker~\cite{Drucker12}, open question number $1$). Can this strengthening of the AND-distillation conjecture be related to a well known complexity theoretic assumption?

\end{itemize}

\noindent 
{\bf Acknowledgement.} The authors thank D\'{a}niel Marx for enlightening discussions on related work in the literature, and  Magnus Wahlstr\"{o}m for pointing out the remark
about randomized pre-processing algorithms following Definition~\ref{def:polyTimePreProcess}.
